\documentclass[smallextended, natbib]{svjour3} 
\usepackage[utf8]{inputenc}
\usepackage[T1]{fontenc}        
\usepackage[normalem]{ulem}
\usepackage{verbatim}
\usepackage[dvipsnames]{xcolor}
\usepackage{amsfonts}
\usepackage{amsmath}
\usepackage{amsfonts}
\usepackage{amssymb}
\usepackage{mathtools}
\usepackage[linesnumbered, ruled]{algorithm2e}
\usepackage{enumerate}
\usepackage{rotating}

\smartqed  

\usepackage{graphicx}
\usepackage{marvosym}	
\usepackage{csquotes}
\usepackage{float}
\usepackage{url}

\usepackage{tikz}
\usetikzlibrary{snakes,shapes}
\tikzstyle{tre}=[circle,draw,minimum size=3mm,inner sep=1pt]
\tikzstyle{trepp}=[circle,draw,minimum size=1.5mm,inner sep=0.1pt]
\tikzstyle{treppp}=[circle,draw,minimum size=1mm,inner sep=0pt]

\renewcommand{\leq}{\leqslant}
\renewcommand{\geq}{\geqslant}
\newcommand{\NN}{\mathbb{N}}

\newcommand{\oC}{\overline{c}}

\newcommand{\treeset}{\mathit{treeset}}

\newcommand{\TT}{\mathcal{T}}

\newcommand{\ZZ}{\mathbb{Z}}

\newcommand{\RR}{\mathbb{R}}
\newcommand{\Tgfb}{T^{\mathit{gfb}}_n}

\DeclarePairedDelimiter\ceil{\lceil}{\rceil}
\DeclarePairedDelimiter\floor{\lfloor}{\rfloor}

\begin{document}
\title{%
On the minimum value of the Colless index and the bifurcating trees that achieve it}
\author{Tom\'as M. Coronado
		\and Mareike Fischer
		\and Lina Herbst
   		\and Francesc Rossell\'o
   		\and Kristina Wicke}

\institute{Francesc Rossell\'o (\Letter) \email{cesc.rossello@uib.eu} \\
Tom\'as M. Coronado, Francesc Rossell\'o \at Dept. of Mathematics and Computer Science, University of the Balearic Islands, E-07122 Palma, Spain, and Balearic Islands Health Research Institute (IdISBa), E-07010 Palma, Spain \\
Mareike Fischer, Lina Herbst, Kristina Wicke \at Institute of Mathematics and Computer Science, University of Greifswald, Greifswald, Germany}

\titlerunning{Minimum Colless index}

\date{Received:  / Revised version: }

\maketitle

\begin{abstract}
Measures of tree balance play an important role in the analysis of phylogenetic trees. One of the oldest and most popular indices in this regard is the Colless index for rooted bifurcating trees, introduced by \citet{Colless1982}. While many of its  statistical properties under different probabilistic models for phylogenetic trees have already been established, little is  known about its minimum value and the trees that achieve it. In this manuscript, we fill this gap in the literature. To begin with, we derive both recursive and closed expressions for the minimum Colless index of a tree with $n$ leaves. Surprisingly, these expressions show a connection between the minimum Colless index and the so-called Blancmange curve, a fractal curve. 
We then fully characterize the tree shapes that achieve this minimum value and we introduce both an algorithm to generate them and a recurrence to count them. After focusing on two extremal classes of trees with minimum Colless index (the maximally balanced trees and the greedy from the bottom trees), we conclude by showing that all trees with minimum Colless index also have minimum Sackin index, another popular balance index.
\end{abstract}

\section{Introduction}
{One of the main goals of evolutionary biology is to understand  the forces that influence speciation and extinction processes and their effect on macroevolution \citep{Futuyma}. Since phylogenetic trees are the standard representation of joint evolutionary histories of groups of species, there has been a natural interest in the development of techniques that allow to assess the imprint of these forces in them \citep{Kubo95,Mooers1997,Stich09}. This imprint may be found in two aspects of a phylogenetic tree: in its branch lengths, which are determined by the timing of speciation events, and in its
\emph{shape}, or \emph{topology}, which is determined by the differences in the diversification rates among clades  \citep[Chap. 33]{fel:04}. Now, it turns out that accurately reconstructing branch lengths that associate a robust timeline to a phylogenetic tree is not easy \citep{Drummond}, whereas different phylogenetic reconstruction methods on the same empirical data tend to agree on the topology of the reconstructed tree \citep{BrowerRindal13,Hillis92,RindalBrower11}. It thus has been the shape of phylogenetic trees  which has become the focus of most studies in this regard, be it through the definition of indices that quantify topological features 
---see,  for instance,  \citep{Fusco95,Mooers1997,Shao:90} and the references on balance indices given below---
or through the frequency distribution of small rooted subtrees \citep{cherries,Savage,Slowinski90,Wu15}. 

Since the early observation by \citet{Yule} that taxonomic trees tend to be asymmetric, with many small clades and only a few large ones at every taxonomic level, the most popular topological feature used to describe the shape of a phylogenetic tree  has been its \emph{balance}: the tendency of the children of any given node to have the same number of descendant leaves. In this way, the \emph{imbalance} of a phylogenetic tree reflects the propensity of diversification events to occur preferentially along specific lineages \citep{Nelson,Shao:90}.
Several \emph{balance indices} have been proposed so far to quantify the balance (or rather, in most cases, the imbalance) of a phylogenetic tree: see, for instance,   \citep{Colless1982,CMR,Fischer2015,Fusco95,KiSl:93,cherries,Mir2013,Mir2018,Sackin1972,Shao:90} and the section ``Measures of overall asymmetry'' in \citet{fel:04} (pp.  562--563). These indices have then been used, among other applications, to test evolutionary models \citep{Aldous01,Blum2005,duchene2018,KiSl:93,Mooers1997,Purvis1996,Verboom2019};  to assess biases in the distribution of shapes obtained through different phylogenetic tree reconstruction methods \citep{Colless1995,Farris98,Holton2014,Sober93,Stam02}; as a tool to discriminate between input parameters in phylogenetic tree simulations 
\citep{Poon2015,Saulnier16}; to compare tree shapes \citep{Avino18,Goloboff17,Kayondo}; or simply to describe phylogenies \citep{chalmandrier2018,Cunha2019,metzig2019,Purvis11}. 

One of the most popular balance indices is the \emph{Colless index}, introduced by \citet{Colless1982}. The Colless index of a rooted bifurcating tree $T$ is defined as  the sum, over all the internal nodes $v$ of $T$, of the absolute value of the difference between the numbers of descendant leaves of the pair of children of $v$; for a recent sound extension to multifurcating trees, see \citep{Mir2018}. The popularity of this index is due to several reasons. First, it is one of  the first balance indices introduced in the literature.   Second, being a sum of values reflecting the  ``local imbalance'' of each internal node in $T$, it measures the global imbalance  of $T$ in a very intuitive way. Moreover,  it has been proved to be one of the most powerful tree shape  indices  in goodness-of-fit tests of probabilistic  models of phylogenetic trees    \citep{Agapow02,KiSl:93,Matsen06} as well as one of the most shape-discriminant balance indices \citep{Hayati19}.
 
As a consequence of this popularity, the statistical properties of the Colless index under several probabilistic models for phylogenetic trees have been thoroughly studied  \citep{Blum2006a,CMR2013,Ford,Heard1992}. In this manuscript we focus on its extremal properties.  More specifically, we solve several open problems related to the minimum Colless index for rooted bifurcating trees with a given number of leaves. Let us mention here that, as far as the maximum Colless index for a given number of leaves $n$ goes, it is folklore knowledge that it is reached at the \emph{caterpillar tree}, or \emph{comb}: the unique rooted bifurcating tree with $n$ leaves where all internal nodes have different numbers of descendant leaves (cf. Figure \ref{treetop}.(a)). Caterpillars are considered since the early paper by \citet{Sackin1972} to be the most imbalanced type of phylogenetic trees, and the fact that they have the maximum Colless index for any number of leaves $n$ was already hinted at by \citet{Colless1982}, who gave a wrong value for their Colless index that was later corrected by \citet{Heard1992} (and confirmed by \citet{Colless1995}) giving the correct maximum value of $(n-1)(n-2)/2$. For a formal proof of the maximality of this Colless index, see Lemma 1 in \citep{Mir2018}.

In contrast, the analysis of the minimum value of the Colless index is much more involved.  On the one hand, despite its popularity and wide use, the minimum Colless index of a bifurcating tree with $n$ leaves is unknown   beyond the often stated straightforward result that for numbers of leaves that are powers of 2 it is reached  at the fully symmetric  trees, which clearly have Colless index 0; see for instance \citep{Heard1992,KiSl:93,Mooers1997}. To have a closed formula for this minimum value is essential in order to normalize the Colless index to the range $[0,1]$  for every number of leaves, making its value independent of its size as it is recommended, for instance, by \citet{Shao:90} or \citet{Stam02}. Up to now, this normalization is performed by simply dividing by its maximum value, as it was suggested by \citet{Heard1992}, but then the normalized index only reaches 0 when $n$ is a power of 2. By subtracting the minimum value and then dividing by the maximum value minus the minimum value we guarantee to reach both ends of the interval $[0,1]$.

On the other hand, this minimum value may be achieved by several trees.
In fact, as we shall see, for every number $n$ of leaves, the \emph{maximally balanced tree} with $n$ leaves \citep{Mir2013}, which is characterized by the property that all its internal nodes are maximally balanced in the sense that the numbers of descendant leaves of their children differ by at most 1, always achieves the minimum Colless index among all bifurcating trees with $n$ leaves. These maximally balanced trees were called ``the most balanced trees'' by \citet{Shao:90}, and they are also classified as ``most balanced'' by the Sackin index \citep{Fischer2018}, the total cophenetic index \citep{Mir2013}, or the rooted quartets index \citep{CMR}, among other indices. But it turns out that, for every $n$ except those of the form $2^m$ or $2^m\pm 1$, there also exist other bifurcating trees with $n$ leaves that achieve the minimum Colless index without being maximally balanced.  In other words, the least global amount of imbalance is almost always achieved also at trees that do not minimize the  local imbalance at each internal node.
This raises the questions of characterizing the family of all ``most balanced trees'' according to the Colless index and counting them.

In this manuscript, we fill these gaps in the literature. To be precise, we first prove a recursive formula and two closed expressions for the minimum Colless index for a given number $n$ of leaves. One of the closed expressions is related to a fractal curve, namely the so-called Blancmange, or Takagi, curve, thus showing the fractal structure and symmetry of the minimum Colless index.
Next, we fully characterize all rooted bifurcating trees with $n$ leaves that have minimum Colless index, we prove that they include the maximally balanced trees, and we provide an efficient algorithm to generate them and a recursive formula to count them. We also focus on a particular class of trees with minimum Colless index,
which we call \emph{greedy from the bottom} (\emph{GFB}) trees. It turns out that there exists a GFB tree for every number $n$ of leaves and they are  almost never maximally balanced (in fact, they are only maximally balanced when $n$ has the form $2^m$ or $2^m\pm 1$, in which case there is only one tree that attains the minimum Colles index). Moreover, the GFB trees and the maximally balanced trees are extremal among those trees with minimum Colless index  in the following sense: for every $m$,
the difference  (in absolute value) between the numbers of descendant leaves of the pair of children of an internal node with  $m$ descendant leaves in a tree $T$ with minimum Colless index achieves its minimum value
when $T$ is maximally balanced and its maximum value when $T$ is greedy from the bottom. 
We conclude by showing that all trees with minimum Colless index also have minimum Sackin index  \citep{Sackin1972,Shao:90} and that the converse implication is false.

Before leaving this Introduction, we want to point out that, although the main motivation to study the Colless index is its application to the description and analysis of phylogenetic trees, it is actually a shape index, that is, its value does not depend on the specific labels at the leaves of the tree, only on the unlabeled tree underlying the phylogenetic tree. For this reason, in most of the rest of this manuscript we shall restrict ourselves to unlabeled trees, and we shall only deal with phylogenetic trees in some remarks.

}

\section{Basic definitions and preliminary results} \label{Sec_Preliminaries}
Before we can present our results, we need to introduce some definitions and notations. Throughout this manuscript, by {a \emph{tree} we mean a  non-empty \emph{rooted tree}: that is, a directed graph $T=(V(T),E(T))$, with node set $V(T)$ and edge set $E(T)$, containing exactly one node of indegree 0, which is called its \emph{root} (denoted henceforth by $\rho$) and such that for every $v\in V(T)$ there exists a unique path from $\rho$ to $v$.}
We use $V_L(T) \subseteq V(T)$ to denote the leaf set of $T$ (i.e. $V_L(T) = \{v \in V\mid \deg_{\mathit{out}}(v) =0\}$) and by $\mathring{V}(T)$ we denote the set of internal nodes, i.e. $\mathring{V}(T) = V(T) \setminus V_L(T)$. {Note in particular that if $|V(T)|\geq 2$, $\rho \in \mathring{V}(T)$.} If $|V(T)|=1$, $T$ consists of only one node, which is at the same time the root and the only leaf of the tree, and no edge.  Whenever there is no ambiguity we simply denote $E(T)$, $V(T)$, $\mathring{V}(T)$, and $V_L(T)$ by $E$, $V$, $\mathring{V}$, and $V_L$, respectively.  To simplify the language, we shall often say that two trees are \emph{equal} when they are actually only isomorphic as rooted trees; we shall also use the expression \emph{to have the same shape} as a synonym of being isomorphic.

Now, a \emph{bifurcating tree} is a rooted tree where all internal nodes have out-degree 2.   We denote by $\mathcal{T}_n$, for every $n\in \NN_{\geq 1}$,  the set of (isomorphism classes of)  bifurcating trees with $n$ leaves.{\footnote{{We always understand that 0 belongs to the set $\NN$ of natural numbers, and, for any given $m\in\NN\setminus\{0\}$, we use the notation $\NN_{\geq m}\coloneqq\{n\in \NN\mid n\geq m\}$.}}} Note that, for $n = 1$, $\TT_1$ consists only of the tree with one node and no edge.

Whenever there exists a path from $u$ to $v$ in a  tree $T$, we say that $u$ is an \emph{ancestor} of $v$ and that $v$ is a \emph{descendant} of $u$. In addition, whenever there exists an edge from $u$ to $v$, we say that $v$ is a  \emph{child} of $u$ and that $u$ is the \emph{parent} of $v$. Note that in a bifurcating tree with $n \geq 2$ leaves, each internal node has exactly two children. Two leaves $x$ and $y$ are said to form a \emph{cherry} when they have the same parent. Given a node $v$ of $T$, we denote by $T_v$ the subtree of $T$ rooted at $v$.

The \emph{depth} $\delta_T(v)$ of a node $v$ is the number of edges on the  path from $\rho$ to $v$ and the \emph{height} $h(T)$ of a tree $T$ is  the maximum depth of any leaf in it.

A  bifurcating tree $T$ {with $n\geq 2$ leaves} can be decomposed into its two \emph{maximal pending subtrees} $T_a$ and $T_b$ rooted at the children $a$ and $b$ of $\rho$, and we shall denote this decomposition by $T=(T_a,T_b)$; cf.\ Figure \ref{fig:star}. We shall usually denote by $n_a$ and $n_b$ the numbers of leaves of $T_a$ and $T_b$, respectively, and without any loss of generality we shall always assume, usually without any further notice, that $n_a \geq n_b \geq 1$.

\begin{figure}[htb]
\begin{center}
\begin{tikzpicture}[thick,>=stealth,yscale=0.25,xscale=0.4]
\draw(0,0) node[tre] (z1) {\scriptsize $a$}; 
\draw (z1)--(-2,-3)--(2,-3)--(z1);
\draw(0,-2) node  {\footnotesize $T_a$};
\draw(5,0) node[tre] (z2) {\scriptsize $a$}; 
\draw (z2)--(3,-3)--(7,-3)--(z2);
\draw(5,-2) node  {\footnotesize $T_b$};
\draw(2.5,2) node[tre] (z) {\scriptsize $\rho$}; 
\draw (z)--(z1);
\draw (z)--(z2);
\end{tikzpicture}
\end{center}
\caption{\label{fig:star} The decomposition $T=(T_a,T_b)$ of a  bifurcating tree into its two maximal pending subtrees.}
\end{figure}
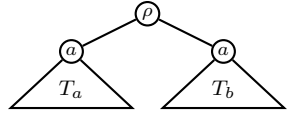

{An internal node $v$ of a bifurcating tree $T$ is a \emph{symmetry vertex}  when the subtrees rooted at its two children have the same shape ---hence, in particular, the same number of leaves. We shall denote by $s(T)$ the number of symmetry vertices in $T$.

Next two definitions introduce two concepts that play a key role in this paper.

\begin{definition}
Let $T$ be a bifurcating tree and let $v \in \mathring{V}$ with children $v_1$ and $v_2$. Then, the \emph{balance value} of $v$ is defined as $bal_T(v) = |\kappa_T(v_1)-\kappa_T(v_2)|$, where $\kappa_T(w)$ denotes the number of leaves of $T_w$, i.e.~the number of descendant leaves of $w\in V$.
We call an internal node $v$ \emph{balanced} if $bal_T(v) \leq 1$, i.e. when its two children have $\lceil \kappa_T(v)/2 \rceil$ and $\lfloor \kappa_T(v)/2 \rfloor$ descendant leaves, respectively.
\end{definition}

\begin{definition}
A bifurcating tree $T$ is called \emph{maximally balanced}  if all its internal nodes are balanced (cf. Figure \ref{treetop}.(b)). Recursively, a  bifurcating tree {with $n\geq 2$ leaves} is maximally balanced if its root is balanced and its two maximal pending subtrees are maximally balanced.
\end{definition}

Note that this last definition easily implies that any rooted subtree of a maximally balanced tree is again maximally balanced, by induction on the depth of the root of the subtree. It also implies that, for every $n \in \mathbb{N}_{\geq 1}$, there exists a unique maximally balanced tree with $n$ leaves, which we shall denote by $T_n^\mathit{mb}$, and that when $n\geq 2$, as we have just mentioned, $T^{\mathit{mb}}_n = (T^{\mathit{mb}}_{\lceil n/2 \rceil}, T^{\mathit{mb}}_{\lfloor n/2 \rfloor} )$.

Our maximally balanced trees were called by  \citet{Shao:90} the  ``most balanced'' bifurcating trees, and they are natural candidates to have the minimum Colless index for every number of leaves. As we shall see, this is indeed the case (see Theorem \ref{thm:minC}), but it will also turn out that for almost all numbers $n$ of leaves there are also other trees with $n$ leaves and minimum Colless index (cf. Proposition \ref{GFB_is_minimal} and Corollary \ref{cor:morethan1}).}

Two other particular families of trees appearing in this manuscript are the caterpillar trees  and the fully symmetric  trees (cf. Figures \ref{treetop}.(a) and (c)). The \emph{caterpillar tree} with $n$ leaves, $T^{\mathit{cat}}_n$, is the unique  bifurcating tree with $n$ leaves all  of whose internal nodes have different numbers of descendant leaves. 
As to the \emph{fully symmetric  tree of height $k$}, $T^{\mathit{fs}}_{k}$, it is the unique tree with $n=2^k$ leaves in which all leaves have depth  $k$. Note that  if $k\geq 1$, $T^{\mathit{fs}}_{k}=(T^{\mathit{fs}}_{k-1}, T^{\mathit{fs}}_{k-1})$, i.e.
 the maximal pending subtrees of a fully symmetric  tree of height $k$ are fully symmetric  trees of height $k-1$. Note also that $T^{\mathit{fs}}_{k} = T^{\mathit{mb}}_{2^k}$, because in the special case when $n=2^k$, $T_k^{\mathit{fs}}$ is the unique tree all of whose internal nodes have balance value 0.

\begin{figure}[htb]
\begin{center}
\begin{tikzpicture}[thick,>=stealth,scale=0.25]
\draw(0,0) node [trepp] (1) {};
\draw(2,0) node [trepp] (2) {};
\draw(4,0) node [trepp] (3) {};
\draw(6,0) node [trepp] (4) {};
\draw(8,0) node [trepp] (5) {}; 
\draw(10,0) node [trepp] (6) {};
\draw(12,0) node [trepp] (7) {};
\draw(11,1) node[trepp] (a) {};
\draw(10,2) node[trepp] (b) {};
\draw(9,3) node[trepp] (c) {};
\draw(8,4) node[trepp] (d) {};
\draw(7,5) node[trepp] (e) {};
\draw(6,6) node[trepp] (r) {};
\draw  (a)--(6);
\draw  (a)--(7);
\draw  (b)--(a);
\draw  (b)--(5);
\draw  (c)--(b);
\draw  (c)--(4);
\draw  (d)--(3);
\draw  (d)--(c);
\draw  (e)--(d);
\draw  (e)--(2);
\draw  (r)--(e);
\draw  (r)--(1);
\draw(6,-2) node {\footnotesize (a) $T^{\mathit{cat}}_7$};
\end{tikzpicture}
\quad
\begin{tikzpicture}[thick,>=stealth,scale=0.25]
\draw(0,0) node [trepp] (1) {};
\draw(2,0) node [trepp] (2) {};
\draw(4,0) node [trepp] (3) {};
\draw(6,0) node [trepp] (4) {};
\draw(8,0) node [trepp] (5) {}; 
\draw(10,0) node [trepp] (6) {};
\draw(12,0) node [trepp] (7) {};
\draw(1,1.5) node[trepp] (a) {};
\draw(5,1.5) node[trepp] (b) {};
\draw(3,3) node[trepp] (c) {};
\draw(11,1.5) node[trepp] (d) {};
\draw(10,3) node[trepp] (e) {};
\draw(6.5,4.5) node[trepp] (r) {};
\draw  (r)--(c);
\draw  (r)--(e);
\draw  (c)--(a);
\draw  (c)--(b);
\draw  (a)--(1);
\draw  (a)--(2);
\draw  (b)--(3);
\draw  (b)--(4);
\draw  (e)--(d);
\draw  (e)--(5);
\draw  (d)--(6);
\draw  (d)--(7);
\draw(6,-2) node {\footnotesize (b) $T^{\mathit{mb}}_7$};
\end{tikzpicture}
\quad
\begin{tikzpicture}[thick,>=stealth,scale=0.25]
\draw(0,0) node [trepp] (1) {};
\draw(2,0) node [trepp] (2) {};
\draw(4,0) node [trepp] (3) {};
\draw(6,0) node [trepp] (4) {};
\draw(8,0) node [trepp] (5) {}; 
\draw(10,0) node [trepp] (6) {};
\draw(12,0) node [trepp] (7) {};
\draw(14,0) node [trepp] (8) {};
\draw(1,1.5) node[trepp] (a) {};
\draw(5,1.5) node[trepp] (b) {};
\draw(3,3) node[trepp] (c) {};
\draw(9,1.5) node[trepp] (d) {};
\draw(13,1.5) node[trepp] (e) {};
\draw(11,3) node[trepp] (f) {};
\draw(7,4.5) node[trepp] (r) {};
\draw  (a)--(1);
\draw  (a)--(2);
\draw  (b)--(3);
\draw  (b)--(4);
\draw  (c)--(a);
\draw  (c)--(b);
\draw  (d)--(5);
\draw  (d)--(6);
\draw  (e)--(7);
\draw  (e)--(8);
\draw  (f)--(d);
\draw  (f)--(e);
\draw  (r)--(c);
\draw  (r)--(f);
\draw(7,-2) node {\footnotesize (c) $T^{\mathit{fs}}_3 = T^{\mathit{mb}}_8$};
\end{tikzpicture}
\end{center}
\caption{\label{treetop} From left to right, the caterpillar tree $T^{\mathit{cat}}_7$ with 7 leaves, the maximally balanced tree $T^{\mathit{mb}}_7$ with 7 leaves, and the fully symmetric  tree $T^{\mathit{fs}}_3 = T^{\mathit{mb}}_8$ of depth 3, with $2^3=8$ leaves.}
\end{figure}
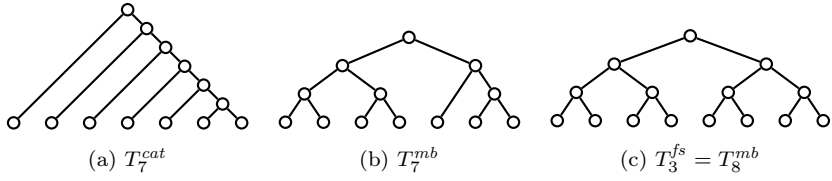

We are now in a position to define the focus of this manuscript: 

\begin{definition}[\citet{Colless1982}] \label{Def_Colless}
The \emph{Colless index} of a  bifurcating tree $T$ is the sum of the balance values of its internal nodes:
$$
\mathcal{C}(T) = \sum_{v \in \mathring{V}(T)}  bal_T(v) = \sum_{v \in \mathring{V}(T)} \vert  \kappa_T(v_1) - \kappa_T(v_2) \vert,
$$
where $v_1$ and $v_2$ denote the children of each $v \in \mathring{V}(T)$.
\end{definition}

Note that $\mathcal{C}(T) \geq 0$, because it is defined as a sum of absolute values.
For instance, consider the three trees depicted in Figure \ref{treetop}. Here, we have: $\mathcal{C}(T^{\mathit{cat}}_7)=15$, $\mathcal{C}(T^{\mathit{mb}}_7)=2$, and $\mathcal{C}(T^{\mathit{fs}}_3)=0$.

Since the Colless index of a tree measures its global imbalance,  the smaller the Colless index of a tree is, the \emph{more balanced} we consider it  to be. In other words, for every pair of trees $T_1,T_2\in \TT_n$, if $\mathcal{C}(T_1) < \mathcal{C}(T_2)$, then $T_1$ is \emph{more balanced} than $T_2$. 
For example, in Figure \ref{treetop}, $T_7^{\mathit{mb}}$ is more balanced than $T_7^{\mathit{cat}}$. {Notice that this comparison is meaningful only  if both trees have the same number of leaves.}

It is easy to see that the Colless index satisfies the following recurrence   \citep{Rogers1993}.

\begin{lemma} \label{colless_sum}
If $T=(T_a,T_b)$ is a  bifurcating tree with $T_a\in \TT_{n_a}$ and $T_b\in \TT_{n_b}$, where $n_a \geq n_b$, then 
$$
\mathcal{C}(T) = \mathcal{C}(T_a) + \mathcal{C}(T_b) +  n_a - n_b .
$$
\end{lemma}

{\begin{corollary}\label{min_colless}
For every $n \geq 1$ and for every $T \in \mathcal{T}_n$, $\mathcal{C}(T) = 0 $ if, and only if, $n$ is a power of 2 and $T$ is fully symmetric.
\end{corollary}

\begin{proof}
The ``if'' implication is a direct consequence of the fact that, in a fully symmetric tree, both children of each internal node have the same number of descendant leaves. We prove now the ``only if'' implication by induction on $n$. The base case $n=1$ being obvious, let $n\geq 2$ and let us assume that the assertion is true for every $1\leq n'<n$. Let $T\in \mathcal{T}_n$ be such that $\mathcal{C}(T)=0$, and let $T=(T_a,T_b)$, with $T_a\in \mathcal{T}_{n_a}$ and $T_b\in \mathcal{T}_{n_b}$,  be its decomposition into its maximal pending subtrees. Then, by Lemma \ref{colless_sum}, $\mathcal{C}(T)=0$ is equivalent to $n_a=n_b$ and $\mathcal{C}(T_a)=\mathcal{C}(T_b)=0$. By the induction hypothesis, this implies that $n_a=n_b$ is a power of 2, and hence that $n=n_a+n_b$ is also a power of 2, and that both $T_a$ and $T_b$ are fully symmetric, and hence that $T=(T_a,T_b)$ is fully symmetric, too. \qed
\end{proof}}

\section{The minimum Colless index}

We shall denote throughout this manuscript by  $c_n$  the \emph{minimum Colless index of a  bifurcating tree with $n$ leaves}:
$$
c_n= \min\big\{\mathcal{C}(T)\mid T \in \mathcal{T}_n\big\}.
$$
{Notice that, by Corollary \ref{min_colless}, $c_n=0$ if, and only if, $n$ is a power of 2.} The main aim of this section is to study the sequence $c_n$. We derive both a recurrence and two closed formulas for this sequence and we point out  both its fractal structure and its symmetry. 
We start by showing that if a  bifurcating tree $T=(T_a, T_b)$ has minimum Colless index, its two maximal pending subtrees also have minimum Colless index.

\begin{lemma} \label{max_subtrees}
Let $T=(T_a,T_b)$ be a  bifurcating tree with $n$ leaves. If $T$ has minimum Colless index on $\TT_n$,  then  $T_a$ and $T_b$ have minimum Colless indices on $\TT_{n_a}$ and $\TT_{n_b}$, respectively. 
\end{lemma}

\begin{proof}
Assume that $\mathcal{C}(T_a)$ is not minimal; the case when $\mathcal{C}(T_b)$ is not minimal is symmetrical.
Then, there exists $\widehat{T}\in \TT_{n_a}$  such that $\mathcal{C}(\widehat{T}) < \mathcal{C}(T_a)$. 
Consider the tree $\widetilde{T}= (\widehat{T},T_b)\in \TT_n$ obtained by replacing in $T$ the rooted subtree  $T_a$ by $\widehat{T}$. Then,  by Lemma \ref{colless_sum},
$$
\mathcal{C}(\widetilde{T}) = \mathcal{C}(\widehat{T}) + \mathcal{C}(T_b) +  n_a  - n_b  < \mathcal{C}(T_a) + \mathcal{C}(T_b) +  n_a - n_b  = \mathcal{C}(T),
$$
which implies that $\mathcal{C}(T)$ is not minimal. Thus, if $\mathcal{C}(T)$ is  minimal, $\mathcal{C}(T_a)$ must be minimal, too. \qed
\end{proof}

\begin{remark}\label{rem:inheritance}
Lemma \ref{max_subtrees} easily implies that every rooted subtree of a tree with minimum Colless index has also minimum Colless index, by induction on the depth of the root of the subtree.
\end{remark}

Lemmas \ref{colless_sum} and \ref{max_subtrees} directly imply that
\begin{align}
 c_n &= \min\{c_{n_a} + c_{n_b} + n_a - n_b \mid n_a \geq n_b \geq 1,\ n_a + n_b=n\}.  \label{cn_min}
\end{align}
In particular, 
\begin{align}
c_n &\leq c_{n_a} + c_{n_b} + n_a - n_b \text{ for every }   n_a \geq n_b \geq 1 \text{ with } n_a + n_b=n,
\label{cn_leq}
\end{align}
a fact that will be useful in subsequent proofs. 

\subsection{The maximally balanced trees have minimum Colless index}

In this subsection we prove that the Colless index of a maximally balanced tree $T^{\mathit{mb}}_n$ is $c_n$. The proof relies on the following lemma, which shows that the sequence $\mathcal{C}(T^{\mathit{mb}}_{n})$ also satisfies the Inequalities (\ref{cn_leq}).

\begin{lemma}\label{lem:key}
For every $n \in \NN_{\geq 2}$ and for every $n_a \geq n_b \geq 1$ such that $n_a + n_b=n$, 
$$
 \mathcal{C}(T^{\mathit{mb}}_{n})\leq \mathcal{C}(T^{\mathit{mb}}_{n_a}) + \mathcal{C}(T^{\mathit{mb}}_{n_b}) + n_a-n_b.
$$
\end{lemma}

\begin{proof}
To simplify the notations, throughout this proof we shall denote $\mathcal{C}(T^{\mathit{mb}}_{n})$ by $C(n)$.
By Lemma \ref{colless_sum} and the equality $T^{\mathit{mb}}_n = (T^{\mathit{mb}}_{\lceil n/2 \rceil}, T^{\mathit{mb}}_{\lfloor n/2 \rfloor} )$, we have that, for every $n\geq 2$,
$$
C(n)=C(\lceil n/2 \rceil)+C(\lfloor n/2 \rfloor)+\lceil n/2 \rceil-\lfloor n/2 \rfloor,
$$
or, equivalently, for every $n\geq 1$,
\begin{equation}
C(2n)=2C(n)\quad\mbox{and}\quad C(2n+1)=C(n+1)+C(n)+1.
\label{eqn:recuforC}
\end{equation}
We shall use this recurrence to prove by induction on $m$ that, for every $m\geq 1$,  the inequality
\begin{equation}
C(m + s) + C(m) + s \geq C(2m + s)
\label{eq:goal}
\end{equation}
holds for every $s\in \NN$. Taking $n_a=m+s$ and $n_b=m$, this clearly entails the statement.

{Since $C(1)=0$, the base case $m=1$ says that, for every $s\geq 0$,
\begin{equation}
C(1+s) + s \geq C(2+s).
\label{eq:base1}
\end{equation}
We prove it by induction on $s$.
The cases  $s=0$ and $s=1$ are obviously true, because  $C(1) + 0 = 0 = C(2)$ and $C(2) + 1 = 1 = C(3)$. Let us now consider the case $s\geq 2$ and let us assume that,  for every $s'< s$,
\begin{equation}
C(1+s') +  s'  \geq C(2+s').
\label{eqn:IHs=1}
\end{equation}
To prove the induction step, we distinguish two cases.
\begin{itemize}
\item If $s$ is even, say $s=2s'$ with $s'\geq 1$, then, 
\begin{align*}
& C(1+s) +s = C(2s'+1)+2s'\\
&\quad =C(s'+1)+C(s')+ 1 + 2s'&\quad \mbox{(by (\ref{eqn:recuforC}))}&\\
&\quad =  C(s'+1) + C(1+(s'-1)) + (s'-1) + s' + 2 \\
&\quad\geq C(s'+1) + C(2+(s'-1))+s'+2 &\quad\mbox{(by (\ref{eqn:IHs=1}))}&\\
&\quad=2C(s'+1)+s'+2=C(2s'+2)+s'+2 &\quad \mbox{(again by (\ref{eqn:recuforC}))}&\\
&\quad >C(2s'+2)= C(2+s).
\end{align*}

\item If $s$ is odd, say $s=2s'+1$ with $s'\geq 1$, 
\begin{align*}
& C(1+s) +s =C(2s'+2)+2s'+1 \\
&\quad = C(s'+1)+C(s'+1)+2s'+1  &\qquad \mbox{(by (\ref{eqn:recuforC}))}&\\
&\quad\geq C(s'+1) + C(s'+2)+s'+1 &\qquad\mbox{(by (\ref{eqn:IHs=1}))}&\\
&\quad=C(2s'+3)+s' &\qquad \mbox{(again by (\ref{eqn:recuforC}))}&\\
& \quad = C(2+s) +s'> C(2+s).
\end{align*}
\end{itemize}
This completes the proof of the base case $m=1$.

Let us consider now the case $m\geq 2$ and let us assume that, for every $1\leq m' < m$ and $s\geq 0$, 
\begin{equation}
C(m' + s) + C(m') + s \geq C(2m' + s).
\label{eqn:IHm}
\end{equation} 
To prove that (\ref{eq:goal}) is true for every $s\in\NN$ we distinguish four cases:

\begin{itemize}
\item  $m$ and $s$  even: say, $m=2m'$ and $s=2s'$. Then, 
\begin{align*}
 & C(m+s) +C(m)+ s=C(2m'+2s')+C(2m')+2s'\\
 &\quad =2C(m'+s')+2C(m')+2s'  &\qquad \mbox{(by (\ref{eqn:recuforC}))}&\\
 &\quad \geq 2C(2m'+s') &\qquad\mbox{(by (\ref{eqn:IHm}))}&\\
 &\quad =C(4m'+2s')=C(2m+s)
\end{align*}
where the second last equality is due again to (\ref{eqn:recuforC}).

\item $m$  even and $s$  odd: say, $m=2m'$ and $s=2s'+1$. Then
\begin{align*}
 & C(m+s) +C(m)+ s\\
 &\quad =C(2m'+2s'+1)+C(2m')+2s'+1\\
 &\quad =C(m'+s'+1)+C(m'+s')+1+2C(m')+2s'+1  &\  \mbox{(by (\ref{eqn:recuforC}))}&\\
  &\quad =\big(C(m'+s'+1)+C(m')+s'+1\big)\\
  &\qquad +\big(C(m'+s')+C(m')+s'\big)+1 \\
  &\quad\geq C(2m'+s'+1)+C(2m'+s')+1 &\ \mbox{(by \quad(\ref{eqn:IHm}))}&\\
 &\quad =C(4m'+2s'+1)=C(2m+s)
\end{align*}
where the second last equality is due again to (\ref{eqn:recuforC}).

\item $m$  odd and $s$ even: say, $m=2m'+1$ and $s=2s'$. 
If $s'=0$,  the desired Inequality (\ref{eq:goal}) amounts to $C(m)+C(m)\geq C(2m)$, which is true because it is actually an equality. So, assume that $s'\geq 1$. Then
\begin{align*}
 & C(m+s) +C(m)+ s\\
 &\quad=C(2m'+2s'+1)+C(2m'+1)+2s'\\
 &\quad =C(m'+s'+1)+C(m'+s')+1\\
 &\qquad+C(m'+1)+C(m')+1+2s'  & \mbox{(by (\ref{eqn:recuforC}))}&\\
  &\quad =\big(C(m'+s'+1)+C(m')+s'+1\big)\\
  &\qquad +\big(C((m'+1)+(s'-1))+C(m'+1)+s'-1\big)+2 \\
  &\quad\geq C(2m'+s'+1)+C(2(m'+1)+s'-1)+2 & 
  \mbox{(by (\ref{eqn:IHm}))}&\\
 &\quad =C(2m'+s'+1)+C(2m'+s'+1)+2\\
 &\quad=C(4m'+2+2s')+2& \mbox{(again by (\ref{eqn:recuforC}))}&\\
 &\quad =C(2m+s)+2>C(2m+s)
 \end{align*}

\item $m$ and $s$  odd: say, $m=2m'+1$ and $s=2s'+1$. Then
\begin{align*}
 & C(m+s) +C(m)+ s\\
 &\quad=C(2m'+2s'+2)+C(2m'+1)+2s'+1\\
 &\quad =2C(m'+s'+1)+C(m'+1)+C(m')+1+2s'+1 &\  \mbox{(by (\ref{eqn:recuforC}))}&\\
  &\quad =\big(C(m'+s'+1)+C(m'+1)+s'\big)\\
  &\qquad +\big(C(m'+s'+1)+C(m')+s'+1\big)+1\\
  &\quad\geq C(2m'+s'+2)+C(2m'+s'+1)+1 &\ \mbox{(by (\ref{eqn:IHm}))}&\\
 &\quad =C(4m'+2s'+3)=C(2m+s)
 \end{align*}
where the second last equality is due again to (\ref{eqn:recuforC}).
\end{itemize}
This completes the proof of the inductive step.} \qed
\end{proof}

{\begin{remark}\label{rem:postkey}
Notice that in the proof of the last lemma  we have established the following two facts, which will be used later:
\begin{enumerate}[(a)]
\item  The inequalities obtained in both cases of the induction step of the proof of (\ref{eq:base1}) are strict. This implies that
$$
C(1+s) + s > C(2+s) \mbox{\quad if, and only if, $s>1$}.
$$

\item The inequality obtained in the induction step of the proof of (\ref{eq:goal}) when  $m\geq 2$ is odd and $s$ is even and greater than 0 is strict. Combining this fact with (a) when $m=1$, we deduce that
$$
C(m+s) +C(m)+ s>C(2m+s)\mbox{\quad if $m$ is odd and $s\geq 2$ is even.}
$$
\end{enumerate}
\end{remark}}

We are now in a position to establish our first main result.

\begin{theorem}\label{thm:minC}
For every $n\geq 1$, $\mathcal{C}(T^{\mathit{mb}}_n)=c_n$.
\end{theorem}

\begin{proof}
We shall prove  by induction on $n$ that $\mathcal{C}(T)\geq \mathcal{C}(T^{\mathit{mb}}_n)$ for every $T\in \TT_n$. The case when $n=1$ is obvious, because $\TT_1=\{T^{\mathit{mb}}_1\}$. Assume now that $n\geq 2$ and that the assertion is true for every number of leaves smaller than $n$ and let $T=(T_{a},T_{b})\in \TT_n$, with $T_a\in\TT_{n_a}$ and $T_b\in \TT_{n_b}$. Then, by Lemma \ref{colless_sum},
$$
\mathcal{C}(T)  = \mathcal{C}(T_a) + \mathcal{C}(T_b) + n_a-n_b  
\geq   \mathcal{C}(T^{\mathit{mb}}_a) + \mathcal{C}(T^{\mathit{mb}}_b) + n_a-n_b \geq \mathcal{C}(T^{\mathit{mb}}_n),
$$
where the first inequality holds by the induction hypothesis and the second inequality by the previous lemma. \qed
\end{proof}

Next corollary says that the sequence $c_n$ is the sequence A296062 in the \textsl{On-Line Encyclopedia of Integer Sequences} \citep{OEIS}.

\begin{corollary}\label{cor:autom}
Let $A(T^{\mathit{mb}}_{n})$ be the number of automorphisms of $T^{\mathit{mb}}_{n}$. Then,
$c_n=n-1-\log_2(A(T^{\mathit{mb}}_{n}))$.
\end{corollary}

\begin{proof}
Since, by definition, the balance value of every internal node in $T_n^{\mathit{mb}}$ is 0 or 1, $c_n=\mathcal{C}(T_n^{\mathit{mb}})$ is equal to the number of internal nodes of $T_n^{\mathit{mb}}$ with non zero balance value. Now, for every internal node $u$ of $T_n^{\mathit{mb}}$, its balance value is 0 if, and only if, the subtrees of $T_n^{\mathit{mb}}$ rooted at its children are isomorphic, that is, if, and only if,  $u$ is a symmetry vertex. Indeed, as we mentioned in Section \ref{Sec_Preliminaries},  the subtrees rooted at the children of $u$ are again maximally balanced, and therefore they have the same numbers of leaves if, and only if, they are isomorphic.

So, the number of symmetry vertices in $T^{\mathit{mb}}_n$ is $s(T^{\mathit{mb}}_n)=n-1-c_n$. 
{Since the number of  automorphisms  of  a  tree  is 2 raised to the number of symmetry vertices in it (see, for instance, Proposition 2.4.2 in \citep{Semple2003}), we conclude that  $A(T^{\mathit{mb}}_{n})=2^{n-1-c_n}$, as stated. \qed}
\end{proof}

Theorem \ref{thm:minC}, together with Lemma \ref{colless_sum}, directly imply the following recurrence for $c_n$, which was already used, for $\mathcal{C}(T_n^{\mathit{mb}})$, in the proof of Lemma \ref{lem:key}: cf. Eqns.~(\ref{eqn:recuforC}).  

\begin{corollary} \label{colless_minimum}
The sequence $c_n$ satisfies that $c_1=0$ and, for every $n\geq 2$,
$$
c_n = c_{\lceil n/2 \rceil} + c_{\lfloor n/2 \rfloor} + \left\lceil n/2 \right\rceil - \left\lfloor n/2 \right\rfloor
$$
or, equivalently, 
$c_{2n}=2c_n$ and $c_{2n+1}=c_{n+1}+c_{n}+1$ for every $n \geq 1$.
\end{corollary}

\subsection{Two closed formulas for the minimum Colless index}

Corollary \ref{colless_minimum} implies that we can recurrently compute $c_n$ for any desired $n$. In this subsection, however, we derive from that recurrence two different closed expressions for $c_n$ and we prove some properties of this sequence. Our first closed formula for $c_n$ is given in terms of the binary expansion of $n$.

\begin{theorem} \label{thm_binaryExpansion}
If $n=\sum_{j=1}^\ell 2^{m_j}$, with $\ell\geq 1$ and $m_1,\ldots,m_\ell\in \NN$ such that $m_1>\cdots>m_\ell$, then
$$ 
c_n = \sum_{j=2}^\ell 2^{m_j}(m_1 - m_j - 2(j-2)).
$$
\end{theorem}

\begin{proof}
For every  $n\geq 1$, let $\oC_n=\sum_{j=2}^\ell 2^{m_j}(m_1-m_j-2(j-2))$, where $n=\sum_{j=1}^\ell 2^{m_j}$ with $m_1>\cdots>m_\ell$. We shall prove that $c_n = \oC_n$ by induction on $n$. 

If $n=1$, $\oC_1=\oC_{2^0} = 0 = c_1$, which proves the base case of the induction.
Now, we assume that the claim holds for every $n'\leq n-1$ and we prove it for $n$ by distinguishing two cases: $n$ even and $n$ odd.

If $n$ is even, i.e. if $m_\ell>0$, we have $\floor*{n/2}=\ceil*{n/2}=n/2 = \sum_{j=1}^\ell 2^{m_j-1}$ with $m_1-1>\cdots>m_\ell-1\geq 0$ and thus
\begin{align*}
c_n &= 2 \cdot c_{n/2} \quad \text{(by Corollary \ref{colless_minimum})} \\
&= 2 \cdot \oC_{n/2} \quad \text{(by the induction hypothesis)}\\
&= 2 \cdot \sum_{j=2}^\ell 2^{m_j-1}\big(m_1-1-(m_j-1)-2(j-2)\big) \\
&= \sum_{j=2}^\ell 2^{m_j}(m_1-m_j-2(j-2))=\oC_n.
\end{align*}

Assume now that $n$ is odd, i.e. that  $m_\ell=0$. Let $k=\min\{j\mid m_j=\ell-j\}$ (which exists because $m_\ell=\ell-\ell$). Then, $\floor*{n/2}=\sum_{j=1}^{\ell-1} 2^{m_j-1}$, with $m_1-1>\cdots>m_{\ell-1}-1$,
and
$$
\ceil*{n/2}=\sum_{j=1}^{\ell-1} 2^{m_j-1}+1=\sum_{j=1}^{k-1} 2^{m_j-1}+\sum_{j=k}^{\ell-1} 2^{\ell-j-1}+1=\sum_{j=1}^{k-1}2^{m_j-1}+2^{\ell-k}
$$
with $m_1-1>\cdots>m_{k-1}-1>\ell-k\geq 0$. 
In this case, 
\begin{align*}
c_n &= c_{\lceil n/2 \rceil} + c_{\lfloor n/2 \rfloor} + \left\lceil n/2 \right\rceil - \left\lfloor n/2 \right\rfloor \quad \text{(by Corollary \ref{colless_minimum})} \\
&= \oC_{\ceil*{n/2}} + \oC_{\floor*{n/2}} + \ceil*{n/2} - \floor*{n/2} \quad \text{(by the induction hypothesis)} \\
&= \sum_{j=2}^{k-1}2^{m_j-1}\big((m_1-1)-(m_j-1)-2(j-2))  \\
&\qquad +2^{\ell-k}(m_1-1-(\ell-k)-2(k-2)\big)\\
&\qquad+\sum_{j=2}^{\ell-1}2^{m_j-1}\big((m_1-1)-(m_j-1)-2(j-2)\big)+1 \\
\hphantom{c_n }&= \sum_{j=2}^{k-1}2^{m_j-1}(m_1-m_j-2({j}-2))+2^{m_k}(m_1-m_k-2(k-2))-2^{\ell-k}\\
&\qquad+\sum_{j=2}^{k-1}2^{m_j-1}(m_1-m_j-2(j-2))\\
&\qquad+\sum_{j=k}^{\ell-1} 2^{\ell-j-1}(m_1-(\ell-j)-2(j-2))+1\\
&\mbox{(because $m_j=\ell-j$ for every $j\geq k$)}\\
&=\sum_{j=2}^{k}2^{m_j}(m_1-m_j-2(j-2))\\
&\qquad+\sum_{j=k}^{\ell-1}2^{\ell-j-1}(m_1-(\ell-j)-2(j-2))+1-2^{\ell-k}
\end{align*}
\begin{align*}
&= \sum_{j=2}^{k}2^{m_j}(m_1-m_j-2(j-2))\\
&\qquad +\sum_{i=k+1}^{\ell}2^{\ell-i}(m_1-(\ell-i)-2(i-2)+1)+1-2^{\ell-k}\\
&= \sum_{j=2}^{k}2^{m_j}(m_1-m_j-2(j-2))\\
&\qquad+\sum_{i=k+1}^{\ell}2^{m_i}(m_1-m_i-2(i-2))+\sum_{i=k+1}^{\ell}2^{\ell-i} +1-2^{\ell-k}\\
&= \sum_{j=2}^{\ell}2^{m_j}(m_1-m_j-2(j-2))=\oC_n.
\end{align*}
This completes the proof of the inductive step. \qed
\end{proof}


\begin{figure}[htbp]
	\centering
	\includegraphics[width=\linewidth]{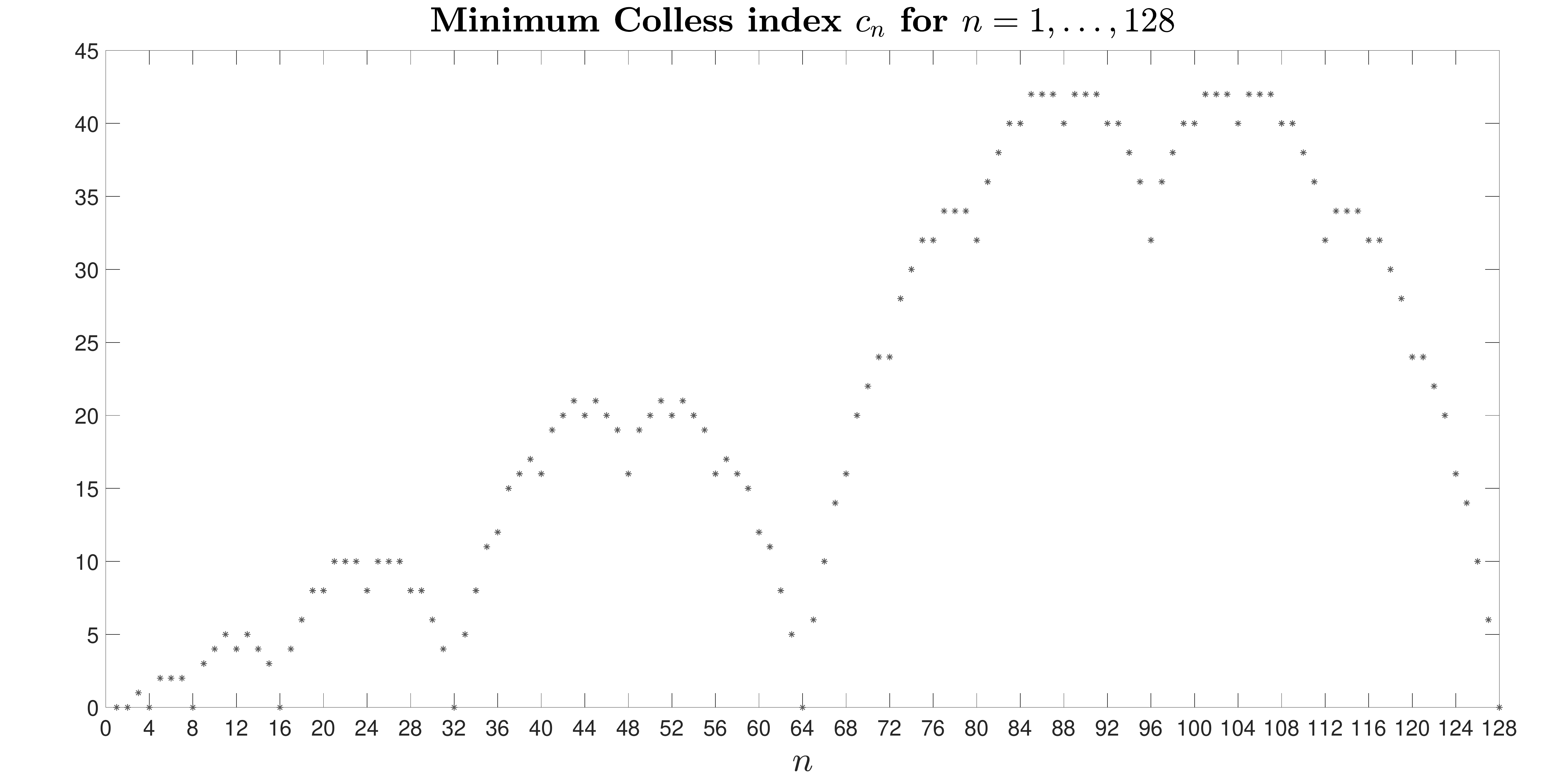}
	\caption{Plot of $c_n$ for $n=1,\ldots, 128$.}
	\label{Fig_MinimumColless}
\end{figure}

Figure \ref{Fig_MinimumColless} depicts the value of $c_n$ for $n=1, \ldots, 128$. Surprisingly, the minimum Colless index exhibits a fractal structure. In the next theorem we provide a second closed formula for $c_n$ that {will explain this fractal structure by entailing} a connection between the sequence $c_n$ and the so-called \emph{Blancmange curve}, a fractal curve also known as the \emph{Takagi curve} (cf. \citet{Takagi1901}).  This curve plays an important role in different areas such as combinatorics, number theory and analysis \citep{Allaart2012} and it is defined as the graph of the function $T: [0,1] \rightarrow \mathbb{R}$ with
\begin{align}\label{Tx}
T(x) = \sum_{i=0}^{\infty} 2^{-i}\cdot s(2^i \cdot x), 
\end{align}
where $s(x) = \min\limits_{z \in \mathbb{Z}}\vert  x-z \vert$  is the distance from $x$ to its nearest integer. {Note that $s(x) \in [0,1/2]$.}
{Moreover, recall that $s$} satisfies the following straightforward properties: $s(n)=0$ for every $n\in \ZZ$;
$s(n+x)=s(x)$ for every $n\in \ZZ$ and $x\in \RR$; $s(x)=s(-x)$ for every $x\in \RR$; if $0\leq x\leq 1/2$, then $s(x)=x$; and  if $1/2\leq x\leq 1$, then $s(x)=1-x$.

\begin{theorem} \label{colless_explicit}
For every $n \geq 1$, let $k_n \coloneqq \lceil \log_2 (n) \rceil$. Then,
$$ 
c_n = \sum_{j=1}^{k_n-1} 2^j\cdot s(2^{-j} \cdot n),
$$
where $s(x)$ is the distance from $x\in \RR$ to its nearest integer.
\end{theorem}

\begin{proof}
We shall prove that the expression for $c_n$ given in the statement is equal to the expression
provided in Theorem \ref{thm_binaryExpansion}. In this proof, it is convenient to write the binary expansion of $n$ as
$n = \sum_{i=1}^\ell 2^{n_i}$ with $n_1 < \cdots < n_\ell$. In this way, the formula given in Theorem \ref{thm_binaryExpansion} becomes
$$
c_n = \sum_{i=1}^{\ell-1} 2^{n_i}(n_\ell - n_i - 2(\ell -i-1)).
$$

With these notations, for every $j\in \NN$,  if $j\leq n_1$, then $2^{-j}\cdot n\in \NN$ and thus $s(2^{-j}\cdot n)=0$, while if $n_{t}< j\leq n_{t+1}$ for some $t=1,\ldots,\ell-1$, then
$$
2^{-j}\cdot n=\sum_{i=1}^t 2^{n_i-j}+\sum_{i=t+1}^\ell 2^{n_i-j},
$$
where $\sum_{i=t+1}^\ell 2^{n_i-j}\in \NN$ and, as far as $\sum_{i=1}^t 2^{n_i-j}$ goes:
\begin{itemize}
\item If $j>n_{t}+1$
$$
\sum_{i=1}^t 2^{n_i-j}=\frac{\sum_{i=1}^t 2^{n_i-n_1}}{2^{j-n_1}}\leq \frac{\sum_{s=0}^{n_t-n_1} 2^{s}}{2^{n_t+2-n_1}}=
\frac{2^{n_t-n_1+1}-1}{2^{n_t-n_1+2}}< \frac{1}{2}
$$

\item If $j=n_{t}+1$
$$
\sum_{i=1}^t 2^{n_i-j}=\sum_{i=1}^t 2^{n_i-n_{t}-1}=\frac{1}{2}+\sum_{i=1}^{t-1} 2^{n_i-n_{t}-1}
$$
where 
$$
0\leq \sum_{i=1}^{t-1} 2^{n_i-n_{t}-1}\leq \frac{\sum_{s=0}^{n_{t-1}} 2^s}{2^{n_t+1}}=\frac{2^{n_{t-1}+1}-1}{2^{n_t+1}}<\frac{1}{2}
$$
and therefore in this case
$1/2\leq \sum_{i=1}^t 2^{n_i-j}<1$.
\end{itemize}

This implies that, if $n_{t}+1< j\leq n_{t+1}$,
\begin{equation}
2^j\cdot s(2^{-j}\cdot n)=2^j  \sum_{i=1}^t 2^{n_i-j} =\sum_{i=1}^t 2^{n_i}
\label{eq:s1}
\end{equation}
and if $j=n_{t}+1$,
\begin{equation}
2^{n_t+1}\cdot s(2^{-n_t-1}\cdot n)  =2^{n_t+1}\Big(\frac{1}{2}-\sum_{i=1}^{t-1} 2^{n_i-n_t-1}\Big)=2^{n_t}-\sum_{i=1}^{t-1} 2^{n_i}.
\label{eq:s2}
\end{equation}

Now, on the one hand, if $n$ is a power of 2, i.e.  if $n=2^{n_1}$, then $k_n=n_1$ and the previous discussion shows that $s(2^{-j} \cdot n)=0$ for every $j\leq n_1-1$, which implies that 
$$
\sum_{j=1}^{k_n-1} 2^j\cdot s(2^{-j} \cdot n) =0=c_n.
$$
On the other hand,   if $n$ is not a power of 2, i.e. if $\ell>1$, then $k_n=n_\ell+1$ and, by the previous discussion,
$$
\sum_{j=1}^{k_n-1} 2^j\cdot s(2^{-j} \cdot n) =
\sum_{j=n_1+1}^{n_\ell}  2^j\cdot s(2^{-j} \cdot n)=
\sum_{t=1}^{\ell-1}  \sum_{j=n_{t}+1}^{n_{t+1}} 2^j\cdot s(2^{-j} \cdot n) 
$$
where, for each $t=1,\ldots,\ell-1$, 
\begin{align*}
& \sum_{j=n_{t}+1}^{n_{t+1}} 2^j\cdot s(2^{-j} \cdot n) =
2^{n_t+1}\cdot s(2^{-n_t-1} \cdot n) +\sum_{j=n_t+2}^{n_{t+1}} 2^j\cdot s(2^{-j} \cdot n)\\
& = 2^{n_t}-\sum_{i=1}^{t-1} 2^{n_i}+(n_{t+1}-n_t-1)\sum_{i=1}^t 2^{n_i}\quad \mbox{(by Eqns.~(\ref{eq:s1}) and (\ref{eq:s2}))}\\
& =(n_{t+1}-n_t)2^{n_t}+(n_{t+1}-n_t-2)\sum_{i=1}^{t-1} 2^{n_i}.
\end{align*}
Therefore
$$ 
\sum_{j=1}^{k_n-1} 2^j\cdot s(2^{-j} \cdot n)=\sum_{t=1}^{\ell-1}\Big((n_{t+1}-n_t)2^{n_t}+(n_{t+1}-n_t-2)\sum_{i=1}^{t-1} 2^{n_i}\Big)
$$
and the coefficient of each $2^{n_i}$, for $i=1,\ldots,\ell-1$, in this expression is
$$
n_{i+1}-n_i+\sum_{j=i+1}^{\ell-1} (n_{j+1}-n_j-2)=n_\ell-n_i-2(\ell-i-1)
$$
which proves that
$$
\sum_{j=1}^{k_n-1} 2^j\cdot s(2^{-j} \cdot n)=
\sum_{i=1}^{\ell-1} 2^{n_i}(n_\ell - n_i - 2(\ell -i-1))=c_n
$$
as we claimed. \qed
\end{proof}

{As we mentioned, there is a close relationship between the sequence $c_n$ and the Takagi curve, which we bring forth now.

\begin{corollary}\label{cor:takagi}
For every $n \in \mathbb{N}_{\geq 1}$, let $k_n \coloneqq \lceil \log_2(n) \rceil$, and let $T(x): [0,1]\to \RR$ be the function whose graph defines the Takagi curve (cf. Eqn.~\eqref{Tx}).
Then, 
$$ 
c_n = 2^{k_n-1} \cdot T\Big(\frac{n}{2^{k_n-1}}-1\Big).
$$
\end{corollary}

\begin{proof}
First of all, recall that $T(x)$ is defined on the unit interval $[0,1]$. Now, as $k_n = \lceil \log_2(n) \rceil$, we clearly have $n \in (2^{k_n-1}, 2^{k_n}]$ and thus, $\frac{n}{2^{k_n-1}}-1 \in (0,1]$. Thus, $T(\frac{n}{2^{k_n-1}}-1)$ is well-defined.
Now, by Theorem \ref{colless_explicit},
\begin{align*}
    c_n &= \sum_{j=1}^{k_n-1} 2^j\cdot s(2^{-j} \cdot n) = \sum_{i=0}^{k_n-2} 2^{k_n-1-i} \cdot s(2^{i-k_n+1} \cdot n) \\
    &= 2^{k_n-1} \cdot \sum_{i=0}^{\infty} 2^{-i} \cdot s(2^{i-k_n+1} \cdot n) \quad\mbox{(because $2^{i-k_n+1}\cdot n\in \NN$ if $i\geq k_n-1$)} \\
    &= 2^{k_n-1} \cdot \sum_{i=0}^{\infty} 2^{-i} \cdot s(2^{i-k_n+1} \cdot n - 2^{i}) \quad \mbox{(because each $2^i \in \mathbb{N}$)} \\
    &= 2^{k_n-1} \cdot \sum_{i=0}^{\infty} 2^{-i} \cdot s\left(2^{i} \cdot \left(\frac{n}{2^{k_n-1}} - 1 \right)\right)= 2^{k_n-1} \cdot T\left(\frac{n}{2^{k_n-1}}-1\right).
\end{align*}
\qed
\end{proof}}

We close this section with the following result, which establishes some properties of the minimum Colless index $c_n$ that are  reflected in Figure \ref{Fig_MinimumColless}, in particular its symmetry. 

\begin{corollary} \label{min_colless_properties}
The sequence $c_n$ satisfies the following properties:
\begin{enumerate}[(a)]
\item For every $m\geq 0$, $c_{2^{m}+1}=m$.
\item {For every $n\geq 1$,  $c_{n} <2^{\lceil\log_2(n)\rceil}/3$.}
\item {For every $n\geq 1$, $c_n<n/2$.}
\item For every $m\geq 1$ and for every $p=1,\ldots,2^m-1$,  $c_{2^{m}+p} = c_{2^{m+1}-p}$.
\end{enumerate}
\end{corollary}

\begin{proof}
Assertion (a) is a direct consequence of Theorem \ref{thm_binaryExpansion}. Indeed,  if $n=2^m+1$ then, with the notations of that theorem, $\ell=2$, $m_1=m$ and $m_2=0$, and therefore $c_{2^m+1} = 2^{0}(m - 0 - 2(2-2))=m$.

{As to (b), if $n=2^m$, then $c_n=0<2^m/3$, and if $n=2^m+p$ with $1\leq p\leq 2^m-1$, so that $\lceil\log_2(n)\rceil=m+1$, then, by Corollary \ref{cor:takagi},
$$
c_n= 2^{m}T\Big(\frac{n}{2^m}-1\Big)\leq \frac{2^{m+1}}{3}
$$
because, by Theorem 3.1 in \citep{Allaart2012}, $T(x)\leq 2/3$ for every $x\in [0,1]$. Moreover, from the explicit description of the numbers $x\in [0,1]$ such that $T(x)= 2/3$ given in the aforementioned theorem, we easily deduce that if $x$ has the form $n/2^m-1$ with $n\in \NN_{\geq 1}$, then $T(x)<2/3$.

Let us prove now (c) by induction on $n$ using Corollary \ref{colless_minimum}. The base case  $n=1$ holds because $c_{1}=0<1/2$. Assume now that $n\geq 2$ and that  the statement holds for every $1\leq n'<n$. Since $c_{n'}$ is a natural number, the inequality $c_{n'}<n'/2$ actually says that if $n'$ is even, say $n'=2n'_0$, then $c_{n'}\leq n'_0-1$, and if $n'$ is odd, say $n'=2n'_0+1$, then $c_{n'}\leq n'_0$. Now we distinguish three cases, depending on the congruence class of $n$ modulo 4:
\begin{itemize}
\item If $n$ is even, say $n=2n_0$, then $c_n=2c_{n_0}<2\cdot (n_0/2)=n_0=n/2$.
\item If $n=4n_0+1$ for some $n_0\in \NN$, then
$$
c_{n}=c_{2n_0+1}+c_{2n_0}+1\leq n_0+n_0-1+1=2n_0<\frac{n}{2}.
$$
\item If $n=4n_0+3$ for some $n_0\in \NN$, then
$$
c_{n}=c_{2n_0+2}+c_{2n_0+1}+1\leq n_0+n_0+1=2n_0+1<\frac{n}{2}.
$$
\end{itemize}
This concludes the proof of (c).}

Finally, as far as (d) goes, let $n=2^m+p$ for some $p=1,\ldots,2^m-1$. Then:
\begin{align*}
c_{2^{m}+p} &= \sum_{j=1}^{m} 2^j\cdot s(2^{-j}  (2^{m}+p)) \quad \text{(by Theorem \ref{colless_explicit})} \\
&= \sum_{j=1}^{m} 2^j\cdot s(2^{m-j}+2^{-j} \cdot p) \\
&= \sum_{j=1}^{m} 2^j\cdot s(2^{-j} \cdot p)\quad \text{(because each $2^{m-j} \in \mathbb{N}$)} \\
&= \sum_{j=1}^{m} 2^j\cdot s(-2^{-j} \cdot p)\quad \text{(because $s(x)=s(-x)$)} \\
&= \sum_{j=1}^{m} 2^j\cdot s(2^{m+1-j}-2^{-j} \cdot p) \quad \text{(because each $2^{m+1-j} \in \mathbb{N}$)}\\
&= \sum_{j=1}^{m} 2^j\cdot s(2^{-j}(2^{m+1}-p))= c_{2^{m+1}-p} \quad \text{(again by  Theorem \ref{colless_explicit}). \hspace*{\fill}\qed}
\end{align*}
%
\end{proof}

{Notice that, when $n\geq 4$, the bounds given in points (a), (b), and (c) in this corollary are stronger  than the upper bound $c_n\leq n-2$ that stems from  Corollary \ref{cor:autom}. Notice moreover that, depending on $n$, either $n/2$ or $2^{\lceil\log_2(n)\rceil}/3$ is a sharper strict bound for $c_n$ and, in general, they cannot be improved: for instance, when $n=11=2^3+3$, $c_n=5=(11-1)/2=\lfloor 2^4/3\rfloor$.}

\section{Minimal Colless trees}
We now turn our attention to the trees that achieve the minimum Colless index for their number of leaves, which we shall call henceforth \emph{minimal Colless} trees. While we have already seen in Theorem \ref{thm:minC} that, for every $n$, the maximally balanced tree $T_n^{\mathit{mb}}$ has minimum Colless index and in Corollary \ref{min_colless} that when $n$ is a power of 2 this is the only minimal Colless tree, for numbers $n$ of leaves that are not powers of 2 there may exist other minimal Colless trees in $\TT_n$. For instance, $c_6=2$ is reached at both trees depicted in Figure \ref{Fig_ExampleGFB}. Actually, as we shall see, for numbers of leaves $n$ that differ more than 1 from a power  of 2 there \emph{always} exist at least two minimal Colless trees (see Corollary \ref{cor:morethan1} below). So, the main goal of this section is to characterize all  minimal Colless trees and to provide an efficient way of  generating them for any given number $n$ of leaves as well as a recurrence to count them.

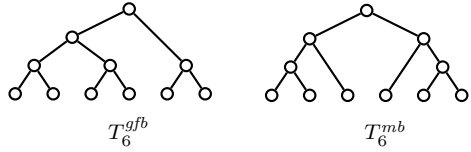
\begin{figure}[htbp]
\begin{center}
\begin{tikzpicture}[thick,>=stealth,scale=0.25]
\draw(0,0) node [trepp] (1) {};
\draw(2,0) node [trepp] (2) {};
\draw(4,0) node [trepp] (3) {};
\draw(6,0) node [trepp] (4) {};
\draw(8,0) node [trepp] (5) {}; 
\draw(10,0) node [trepp] (6) {};
\draw(1,1.5) node[trepp] (a) {};
\draw(5,1.5) node[trepp] (b) {};
\draw(3,3) node[trepp] (c) {};
\draw(9,1.5) node[trepp] (d) {};
\draw(6,4.5) node[trepp] (r) {};
\draw  (r)--(c);
\draw  (r)--(d);
\draw  (c)--(a);
\draw  (c)--(b);
\draw  (a)--(1);
\draw  (a)--(2);
\draw  (b)--(3);
\draw  (b)--(4);
\draw  (d)--(5);
\draw  (d)--(6);
\draw(6,-2) node {\footnotesize $T^{\mathit{gfb}}_6$};
\end{tikzpicture}
\qquad
\begin{tikzpicture}[thick,>=stealth,scale=0.25]
\draw(0,0) node [trepp] (1) {};
\draw(2,0) node [trepp] (2) {};
\draw(4,0) node [trepp] (3) {};
\draw(6,0) node [trepp] (4) {};
\draw(8,0) node [trepp] (5) {}; 
\draw(10,0) node [trepp] (6) {};
\draw(1,1.5) node[trepp] (a) {};
\draw(2,3) node[trepp] (b) {};
\draw(9,1.5) node[trepp] (c) {};
\draw(8,3) node[trepp] (d) {};
\draw(5,4.5) node[trepp] (r) {};
\draw  (r)--(b);
\draw  (r)--(d);
\draw  (b)--(a);
\draw  (b)--(3);
\draw  (a)--(1);
\draw  (a)--(2);
\draw  (d)--(4);
\draw  (d)--(c);
\draw  (c)--(5);
\draw  (c)--(6);
\draw(6,-2) node {\footnotesize $T^{\mathit{mb}}_6$};
\end{tikzpicture}
\end{center}	
\caption{\label{Fig_ExampleGFB} The GFB tree $T^{\mathit{gfb}}_6$ (cf. Subsection \ref{sec:GFB}) and the maximally balanced tree $T^{\mathit{mb}}_6$ with 6 leaves. Both trees have minimum Colless index in $\TT_6$, namely $c_6 = \mathcal{C}(T^{\mathit{gfb}}_6)=\mathcal{C}(T^{\mathit{mb}}_6)=2$, and they are the only trees in $\TT_6$ with Colless index 2.}	
\end{figure}

\subsection{Characterizing and generating minimal Colless trees}
Recall from Eqn.~\eqref{cn_min} that for $n \geq 2$
$$ 
c_n = \min\{c_{n_a} + c_{n_b} + n_a - n_b \, | \, n_a \geq n_b \geq 1, n_a + n_b=n\}. 
$$
To simplify the language, for every $n \geq 2$, let
$$
\begin{array}{rl}
QB(n)\coloneqq\big\{ (n_a,n_b)\in \NN^2\mid &  n_a\geq n_b\geq 1,\ n_a+n_b=n,\\ &ç
 c_{n_a}+c_{n_b}+n_a-n_b=c_n \big\}.
\end{array}
$$ 
Notice that $QB(n)\neq \emptyset$, because $(\ceil*{n/2},\floor*{n/2})\in QB(n)$ by  Corollary \ref{colless_minimum}.

The next proposition gives a characterization of the minimal Colless trees in terms of the sets $QB$ that will allow us to efficiently generate them.

{
\begin{proposition}\label{lem:charmin1}
Let $T=(T_a,T_b)\in \TT_n$, with  $T_a\in \TT_{n_a}$ and $T_b\in \TT_{n_b}$. The following three conditions are equivalent:
\begin{enumerate}[(a)]
    \item $T$ is a minimal Colless tree.
    \item $T_a$ and $T_b$ are minimal Colless trees and $(n_a,n_b)\in QB(n)$.
    \item $(\kappa_T(v_1),\kappa_T(v_2))\in QB(\kappa_T(v))$ for every $v\in \mathring{V}(T)$ with children $v_1,v_2$ so that $\kappa_T(v_1)\geq \kappa_T(v_2)$.
\end{enumerate}
\end{proposition}

\begin{proof}
(a)$\Rightarrow$(b): Let $T$ be a minimal Colless tree. Then, $T_a$ and $T_b$ are also minimal Colless by Lemma \ref{max_subtrees} and, by Lemma \ref{colless_sum},
$$
c_n=\mathcal{C}(T)=c_{n_a}+c_{n_b}+n_a-n_b,
$$
which implies that $(n_a,n_b)\in QB(n)$.

\noindent (b)$\Rightarrow$(c): We shall prove by induction on $m$ the following assertion:
\begin{quote}
If $T=(T_a,T_b)\in \TT_n$ is such that $T_a$ and $T_b$ are minimal Colless  trees and $(n_a,n_b)\in QB(n)$ and if $v\in \mathring{V}(T)$, with children $v_1,v_2$ so that $\kappa_T(v_1)\geq \kappa_T(v_2)$, has depth $\delta_T(v)=m$, then $(\kappa_T(v_1),\kappa_T(v_2))\in QB(\kappa_T(v))$. 
\end{quote}

The case when $m=0$ holds because if $\delta_T(v)=0$, then $v$ is the root of $T$ and
$$
(\kappa_T(v_1),\kappa_T(v_2))=(n_a,n_b)\in QB(n)=QB(\kappa_T(v))
$$ 
by assumption. Assume now that the assertion is true for $m-1$ and let 
$v$ be an internal node of depth $m>0$ of a tree $T$ such that $T_a$ and $T_b$ are minimal Colless  and $(n_a,n_b)\in QB(n)$. Then $v$ will be an internal node of either $T_a$ or $T_b$; without any loss of generality, we shall assume that $v\in \mathring{V}(T_a)$. Let $T_a=(T_{a,1},T_{a,2})$ be the decomposition of $T_a$ into its maximal pending subtrees, with $T_{a,i}\in \TT_{n_{a,i}}$, $i=1,2$. Then, since $T_a$ is minimal Colless, by the implication (a)$\Rightarrow$(b), which we have already proved, $T_a$ satisfies that $T_{a,1}$ and $T_{a,2}$ are minimal Colless and $(n_{a,1},n_{a,2})\in QB(n_a)$. Since $\delta_{T_a}(v)=m-1$, by the inductive hypothesis we conclude that $$
(\kappa_T(v_1),\kappa_T(v_2))=(\kappa_{T_a}(v_1),\kappa_{T_a}(v_2))\in QB(\kappa_{T_a}(v))=QB(\kappa_{T}(v)),
$$
as we wanted to prove.

\noindent (c)$\Rightarrow$(a):  We shall prove that if $T$ satisfies that 
 if $c_{\kappa_T(v_1)}+c_{\kappa_T(v_2)}+\kappa_T(v_1)-\kappa_T(v_2)= c_{\kappa_T(v)}$ 
for every $v\in \mathring{V}(T)$ with children $v_1,v_2$ so that $\kappa_T(v_1)\geq \kappa_T(v_2)$, then $\mathcal{C}(T)=c_n$,  by induction on the number $n$ of leaves in $T$. The case when $n=1$ is obvious, because $\TT_1=\{T^{\mathit{mb}}_1\}$. Assume now that this implication is true for every tree in $\TT_{n'}$ with $n'<n$, and let $T\in \TT_n$
be such that, for every $v\in \mathring{V}(T)$, 
$$c_{\kappa_T(v_1)}+c_{\kappa_T(v_2)}+\kappa_T(v_1)-\kappa_T(v_2)=c_{\kappa_T(v)},$$ where $v_1,v_2$ stand for the children of $v$ so that $\kappa_T(v_1)\geq \kappa_T(v_2)$. 
Let  $T=(T_a,T_b)$ be the decomposition of $T$ into its maximal pending subtrees, with
$T_a\in \TT_{n_a}$ and $T_b\in \TT_{n_b}$ so that $n_a\geq n_b$. 
Then, for every $v\in \mathring{V}(T_{a})$, with children $v_1,v_2$ so that $\kappa_T(v_1)\geq \kappa_T(v_2)$, 
$$
\begin{array}{l}
c_{\kappa_{T_{a}}(v_1)}+c_{\kappa_{T_{a}}(v_2)}+\kappa_{T_{a}}(v_1)-\kappa_{T_{a}}(v_2)\\
\quad =c_{\kappa_{T}(v_1)}+c_{\kappa_{T}(v_2)}+\kappa_{T}(v_1)-\kappa_{T}(v_2)=c_{\kappa_{T}(v)}=c_{\kappa_{T_{a}}(v)}.
\end{array}
$$
This implies, by the induction hypothesis, that $\mathcal{C}(T_{a})=c_{\kappa_T(a)}$. 
By symmetry, we also have that $\mathcal{C}(T_{b})=c_{\kappa_T(b)}$.
Finally,  
\begin{align*}
\mathcal{C}(T) & =\mathcal{C}(T_{a})+\mathcal{C}(T_{b})+n_a-n_b\\ & =
c_{\kappa_T(a)}+c_{\kappa_T(b)}+\kappa_T(a)-\kappa_T(b) =c_{\kappa_T(\rho)}=c_n
\end{align*}
as we wanted to prove. \qed
\end{proof}

}
Next result provides  a characterization of the  pairs $(n_a,n_b)\in QB(n)$, for every $n\geq 2$. 

\begin{proposition}\label{prop:eqC}
For every $n\geq 2$ and for every $n_a,n_b\in \NN_{\geq 1}$ such that $n_a\geq n_b$ and $n_a+n_b=n$:
\begin{enumerate}[(1)]

\item If $n_a=n_b=n/2$, then $(n_a,n_b)\in QB(n)$ always.

\item If $n_a>n_b$, then $(n_a,n_b)\in QB(n)$ if, and only if, one of the following three conditions is satisfied:
\begin{itemize}
\item There exist $k\in \NN$ and $p\in \NN_{\geq 1}$ such that $n=2^k(2p+1)$, $n_a=2^k(p+1)$ and $n_b=2^kp$.

\item There exist $k\in \NN$, $l\in \NN_{\geq 2}$, $p\in \NN_{\geq 1}$, and $t\in \NN$,  $0\leq t<2^{l-2}$, such that $n=2^k(2^l(2p+1)+2t+1)$, $n_a=2^{k+l}(p+1)$, and $n_b=2^{k}(2^lp+2t+1)$. 

\item There exist $k\in \NN$, $l\in \NN_{\geq 2}$, $p\in \NN_{\geq 1}$, and $t\in \NN$,  $0\leq t<2^{l-2}$, such that $n=2^k(2^l(2p+1)-(2t+1))$, $n_a=2^k(2^l(p+1)-(2t+1))$, and $n_b=2^{k+l}p$.
\end{itemize}
\end{enumerate}
\end{proposition}

{Our proof of this proposition is by induction on  $n_a-n_b$ and discussing three cases that depend on whether this difference is even or odd and, in this last case, on whether $n_b$ is even or odd. Since the resulting proof is long, in order not to lose the thread of the manuscript we postpone it until Appendix A.1.}

We now translate this proposition into an explicit and non-redundant description of $QB(n)$ from the binary expansion of $n$.

\begin{proposition}\label{cor:QB}
For every $n\geq 2$, let $k\in \NN$ be the exponent of the largest power of 2 that divides $n$, let $n_0=n/2^k$, and let $n_0=\sum_{i=1}^\ell 2^{m_i}$, with $\ell\geq 1$ and  $m_1>\cdots>m_{\ell-1}>m_\ell=0$, be the binary expansion of $n_0$. Then:
\begin{enumerate}[(a)]
\item If $\ell=1$, i.e. if $n=2^k$, then $QB(n)=\{(n/2,n/2)\}$.

\item If $\ell>1$:
\begin{itemize}
\item[(b.1)] $QB(n)$ always contains the pair
$$
\Big(2^k\Big(\sum_{i=1}^{\ell-1} 2^{m_i-1}+1\Big), 2^k\sum_{i=1}^{\ell-1} 2^{m_i-1}\Big).
$$

\item[(b.2)] For every $j=2,\ldots, \ell-1$ such that $m_j>m_{j+1}+1$, $QB(n)$ contains the pair 
$$
\Big(2^{k}\Big(\sum_{i=1}^{j-1} 2^{m_i-1}+2^{m_j}\Big),n-2^{k}\Big(\sum_{i=1}^{j-1} 2^{m_i-1}+2^{m_j}\Big)\Big).
$$

\item[(b.3)] For every $j=2,\ldots,\ell-1$ such that $m_j<m_{j-1}-1$, 
$QB(n)$ contains the pair 
$$
\Big(n-2^{k}\sum_{i=1}^{j-1} 2^{m_i-1},2^{k}\sum_{i=1}^{j-1} 2^{m_i-1}\Big).
$$

\item[(b.4)] If $k\geq 1$, then $QB(n)$ contains the pair $(n/2,n/2)$.
\end{itemize}\smallskip

Moreover, the pairs  described in  (b.1) to (b.4) are pairwise different and $QB(n)$ contains no other element.
\end{enumerate}
\end{proposition}

\begin{proof}
Assertion (a) is a consequence of {Corollary \ref{colless_minimum} and} the fact that if $QB(n)$ contains some $(n_a,n_b)$ with $n_a>n_b$, then by (2) in the last proposition $n$ cannot be a power of 2. So, assume henceforth that $\ell>1$. Let now $(n_a,n_b)\in \NN^2$ be such that $n=n_a+n_b$ and $1\leq n_b< n_a$. Then, by Proposition \ref{prop:eqC}, $(n_a,n_b)\in QB(n)$ if, and only if,  one of the following three conditions is satisfied:
\begin{itemize}
\item[(b.1)] There exist $k\in \NN$ and $p\in \NN_{\geq 1}$ such that $n=2^k(2p+1)$, and hence $n_0=2p+1$, and $n_a=2^k(p+1)$.
In this case 
$$
p=\frac{n_0-1}{2}=\sum_{i=1}^{\ell-1} 2^{m_i-1}
$$
 and this contributes to $QB(n)$ the pair
$(n_a,n_b)$ with
$$
n_a=2^{k}\Big(\sum_{i=1}^{\ell-1} 2^{m_i-1}+1\Big),\quad n_b=2^{k}\sum_{i=1}^{\ell-1} 2^{m_i-1}.
$$

\item[(b.2)] There exist $k\in \NN$, $l\in \NN_{\geq 2}$, $p\in \NN_{\geq 1}$, and $t\in \NN$,  $0\leq t<2^{l-2}$, such that 
$n=2^k(2^{l+1}p+2^l+2t+1)$, and hence $n_0=2^{l+1}p+2^l+2t+1$, and $n_a=2^{k+l}(p+1)$. Now, if $t<2^{l-2}$ and $p\geq 1$, then $2t+1<2^{l-1}$ and $2^{l+1}p\geq 2^{l+1}$. Therefore, the equality $$2^{l+1}p+2^l+2t+1=\sum_{i=1}^\ell 2^{m_i}$$ holds for some $p\geq 1$ and $t<2^{l-2}$ if, and only if,  $m_j=l\geq 2$ and $m_{j+1}<l-1$ for some $j=2,\ldots,\ell-1$, in which case  $$
p=\frac{\sum_{i=1}^{j-1} 2^{m_i}}{2^{m_j+1}}.
$$ 
This contributes to $QB(n)$ the pairs $(n_a,n_b)$ of the form 
\begin{equation}
\begin{array}{l}
\displaystyle n_a=2^{k+m_j}\Big(\frac{\sum_{i=1}^{j-1} 2^{m_i}}{2^{m_j+1}}+1\Big)=2^{k}\Big(\sum_{i=1}^{j-1} 2^{m_i-1}+2^{m_j}\Big),\\
\displaystyle  n_b=n-2^{k}\Big(\sum_{i=1}^{j-1} 2^{m_i-1}+2^{m_j}\Big),
\end{array}
\label{eq:QB.b.2}
\end{equation}
with $j=2,\ldots,\ell-1$ and $m_j\geq 2$ such that $m_{j+1}<m_j-1$. All these pairs are different, because
$\sum_{i=1}^{h-1} 2^{m_i-1}+2^{m_h}$ is {decreasing} on $h$ (because $m_{h+1}< m_h$) and
\begin{align*}
\sum_{i=1}^{h-1} 2^{m_i-1}+2^{m_h} & =\sum_{i=1}^{h} 2^{m_i-1}+2^{m_{h+1}}  \Longleftrightarrow 
 2^{m_h}=2^{m_h-1}+2^{m_{h+1}} \\
&\qquad  \Longleftrightarrow 2^{m_h-1}=2^{m_{h+1}}\Longleftrightarrow m_h=m_{h+1}+1.
\end{align*}

\item[(b.3)] There exist $k\in \NN$, $l\in \NN_{\geq 2}$, $p\in \NN_{\geq 1}$, and $t\in \NN$,  $0\leq t<2^{l-2}$ such that $n=2^k(2^{l+1}p+2^l-(2t+1))$, and hence $n_0=2^{l+1}p+2^l-(2t+1)$, and  $n_b=2^{k+l}p$. Since $t<2^{l-2}$, we have that $n_0=2^{l+1}p+2^{l-1}+2t_0+1$ with $2t_0+1<2^{l-1}$. Then, the equality $$2^{l+1}p+2^{l-1}+2t_0+1=\sum_{i=1}^\ell 2^{m_i}$$ holds for some $p\geq 1$ and $t_0<2^{l-2}$  if, and only if,  $l-1=m_j$ for some $j=2,\ldots,\ell-1$ such that $m_{j-1}\geq l+1=m_j+2$, and then  $$p=\frac{\sum_{i=1}^{j-1} 2^{m_i}}{2^{m_j+2}}.$$ This contributes to $QB(n)$  all pairs $(n_a,n_b)$ of the form 
$$
n_b=2^{k+m_j+1}\Big(\frac{\sum_{i=1}^{j-1} 2^{m_i}}{2^{m_j+2}}\Big)=2^k\sum_{i=1}^{j-1} 2^{m_i-1},\quad n_a=n-2^k\sum_{i=1}^{j-1} 2^{m_i-1},
$$
with $j=2,\ldots,\ell-1$ such that $m_j<m_{j-1}-1$, belong to $QB(n)$, and they are pairwise different because $n_b$ is strictly increasing on $j$. 
\end{itemize}
This gives all pairs $(n_a,n_b)$ in $QB(n)$ with $n_a>n_b$. If $n$ is even, we must add moreover to $QB(n)$ the pair $(n/2,n/2)$ and this completes the set of pairs belonging to $QB(n)$. To finish the proof of the statement, we verify that these pairs are pairwise different:
{\begin{itemize}
    \item Along our construction we have already checked that the pairs of the form (b.2), as well as those of the form (b.3),  are pairwise different. 
    \item The pairs of the form (b.2) are different from the pair (b.1) because their entry $n_a$ are strictly larger than the entry $n_a$ in (b.1). Indeed, since the~index $j = 2,\ldots,\ell-1$ defining a 
     pair of the form (b.2) satisfies that $m_j-1>m_{j+1}$, we have that
    \begin{align*}
    &\sum_{i=1}^{j-1} 2^{m_i-1}+2^{m_j} =\sum_{i=1}^{j} 2^{m_i-1}+2^{m_j-1}>\sum_{i=1}^{j} 2^{m_i-1}+2^{m_{j+1}}\\
    & \qquad =\sum_{i=1}^{j} 2^{m_i-1}+\sum_{s=0}^{m_{j+1}-1} 2^s+1\\
    &\qquad 
    \geq \sum_{i=1}^{j} 2^{m_i-1}+\sum_{i=j+1}^{\ell-1} 2^{m_i-1}+1=
    \sum_{i=1}^{\ell-1} 2^{m_i-1}+1.
    \end{align*}

\item The pairs of the form (b.3)  are  different from the pair (b.1) because their entry $n_b$ are strictly smaller than the entry $n_b$ in (b.1), a fact that is a direct consequence of their form and the assumption $j\leq \ell-1$  in (b.3).

\item If the pair 
$(n/2,n/2)$ is added to $QB(n)$, it is not of the form (b.1) to (b.3), because all these pairs have both entries divisible by $2^k$, while  the maximum power of 2 that divides $n/2$ is $2^{k-1}$. 
\item Finally, if $(n_a,n_b)$ is a pair of the form (b.2), then $n_a/2^k$ is even and $n_b/2^k$ is odd, while   if $(n_a,n_b)$ is a pair of the form (b.3), then $n_a/2^k$ is odd and $n_b/2^k$ is even. Therefore, no pair can simultaneously be of the form (b.2) and (b.3). \qed
\end{itemize}}
\end{proof}

\begin{example}
Let us find $QB(214)$. Since $214=2(2^6+2^5+2^3+2+1)$, with the notations of the last corollary we have that $k=1$, $\ell=5$, $m_1=6$, $m_2=5$, $m_3=3$, $m_4=1$, and $m_5=0$. Then:
\begin{itemize}
\item[(b.1)]  The pair of this type  in $QB(214)$ is 
$\big(2^k(\sum_{i=1}^4 2^{m_i-1}+1),
2^k\sum_{i=1}^4 2^{m_i-1}\big)=(108,106)$.

\item[(b.2)] The indices $j\in \{2,3,4\}$ such that $m_j>m_{j+1}+1$ are 2 and 3. Therefore, the pairs of this type  in $QB(214)$ are:
\begin{itemize}
\item For $j=2$, $\big(2^{k}(2^{m_1-1}+2^{m_2}),n-2^{k}(2^{m_1-1}+2^{m_2})\big)=(128,86)$.

\item For $j=3$, $\big(2^{k}(2^{m_1-1}+2^{m_2-1}+2^{m_3}),n-2^{k}(2^{m_1-1}+2^{m_2-1}+2^{m_3})\big)=(112,102)$.

\end{itemize}

\item[(b.3)] The indices $j\in \{2,3,4\}$ such $m_j<m_{j-1}-1$ are 3 and 4. Therefore, the pairs of this type  in $QB(214)$ are:
\begin{itemize}
\item For $j=3$, $\big(n-2^k(2^{m_1-1}+2^{m_2-1}),2^k(2^{m_1-1}+2^{m_2-1})\big)=(118,96)$.

\item For $j=4$, $\big(n-2^k(2^{m_1-1}+2^{m_2-1}+2^{m_3-1}),2^k(2^{m_1-1}+2^{m_2-1}+2^{m_3-1})\big)=(110,104)$.
\end{itemize}

\item[(b.4)]  Since $214=2\cdot 107$ is even, $QB(214)$ contains the pair $(107,107)$.

\end{itemize}
Therefore
$$
QB(214)=\big\{(107,107),(108,106),(110,104),(112,102),(118,96),(128,86)\big\}.
$$
\end{example}

\begin{corollary}
For every $n\geq 2$, the cardinality of $QB(n)$ is at most $ \floor*{\log_2(n)}$.
\end{corollary}

\begin{proof}
Let $n_{(2)}$ denote the binary representation of $n$. If $n$ is a power of 2, then $|QB(n)|=1\leq \floor*{\log_2(n)}$. Assume henceforth that $n$ is not a power of 2. {In this case, by construction, the number of pairs of type (b.2) in $QB(n)$ is the number of maximal sequences of zeroes in $n_{(2)}$ that do not start immediately after the leading 1 or that do not end in the units position;  the number of pairs of type (b.3) in $QB(n)$ is the number of maximal sequences of zeroes in $n_{(2)}$ that do not end immediately before the last 1 or  in the units position;}  there is one pair of type (b.4) in $QB(n)$ if $n_{(2)}$ contains a sequence of zeroes ending in the units position; and $QB(n)$ always contains a pair of the form (b.1). 
So, if we denote by $M_0(n)$ the number of maximal sequences of zeroes in $n_{(2)}$,  to compute the cardinality $|QB(n)|$:
\begin{itemize}
\item We count twice the number of  maximal sequences of zeroes in $n_{(2)}$ plus 1, $2M_0(n)+1$
\item We subtract 1 if $n_{(2)}$ contains a maximal sequence of zeroes  starting immediately after the leading 1
\item We subtract 1 if $n_{(2)}$ contains a maximal sequence of zeroes  ending immediately before the last 1
\item We subtract 2 and we add 1 (i.e. we subtract 1)  if $n_{(2)}$ contains a maximal sequence of zeroes  ending in the units position
\end{itemize}
For simplicity, we call any maximal sequence of zeroes in $n_{(2)}$ that starts immediately after the leading 1 or ends immediately before the last 1 or in the units position \emph{forbidden}. Using this notation we have
\begin{equation}
|QB(n)|\!=\! 2M_0(n)\!+\!1\! -\!\big|\{\mbox{forbidden maximal sequences of 0s in $n_{(2)}$}\}\big|.
\label{eqn:M0}
\end{equation}
In the subtraction in this formula we count each forbidden maximal sequence of zeroes as many times as it satisfies a ``forbidden'' property. So, a maximal sequence of zeroes starting immediately after the leading 1 and  ending immediately before the last 1 or in the units position subtracts 2. 

Now, on the one hand, if $\floor*{\log_2(n)}$ is an even number, by the pigeonhole principle we have that $M_0(n)\leq \floor*{\log_2(n)}/2$. But if $n_{(2)}$ does not contain any forbidden maximal sequence of zeroes, then $n_{(2)}$ starts with $11$ and ends with $11$ and the number of maximal sequences of zeroes in such an $n_{(2)}$ is at most $\floor*{\log_2(n)}/2-1$. So, if $M_0(n)=\floor*{\log_2(n)}/2$, then $n_{(2)}$ contains some forbidden maximal sequence of zeroes and then
$|QB(n)|\leq 2M_0(n)=\floor*{\log_2(n)}$, while if $M_0(n)\leq \floor*{\log_2(n)}/2-1$, then $|QB(n)|\leq 2M_0(n)+1\leq \floor*{\log_2(n)}-1$.

On the other hand,  if $\floor*{\log_2(n)}$ is an odd number, again by the pigeonhole principle we have that $M_0(n)\leq (\floor*{\log_2(n)}+1)/2$. Now, if $M_0(n)=(\floor*{\log_2(n)}+1)/2$, then $n_{(2)}$ contains at least 2 forbidden maximal sequences of zeroes. Indeed, let $\floor*{\log_2(n)}=2s+1$. If $n_{(2)}$ starts with $11$, avoiding a forbidden maximal sequence of zeroes at the beginning, then $M_0(n)\leq s=(\floor*{\log_2(n)}-1)/2$.  On the other hand, if it ends in $11$, avoiding a forbidden maximal sequence of zeroes at the end, then again $M_0(n)\leq s=(\floor*{\log_2(n)}-1)/2$.
So, to reach the maximum value of $M_0(n)$, $n_{(2)}$ must start with $10$ and end with $10$, $01$ or $00$, thus having at least 2 forbidden maximal sequences of zeroes. 
Thus, if $M_0(n)=(\floor*{\log_2(n)}+1)/2$, then
$|QB(n)|\leq 2M_0(n)-1=\floor*{\log_2(n)}$, while if $M_0(n)\leq (\floor*{\log_2(n)}+1)/2-1$, then $|QB(n)|\leq 2M_0(n)+1\leq \floor*{\log_2(n)}$. \qed
\end{proof}

{Proposition \ref{lem:charmin1}, together with Corollary \ref{min_colless}, provide the following algorithm to produce minimal Colless trees in $\TT_n$, which is reminiscent of Aldous' $\beta$-model~\citep{Ald1}. If the algorithm is run non-deterministically for all choices of a labeled leaf in line 3 and of a pair  $(m_a, m_b) \in QB(m)$ in line 8 (using Proposition \ref{cor:QB} to find all these pairs) in all executions of the \textbf{while} loop, one obtains all minimal Colless trees in $\TT_n$, possibly with repetitions that can be then removed (see Example \ref{ex:2} below). The non-deterministic choice of the leaf in line 3 can be made deterministic by considering \emph{oriented trees} (i.e., adding an orientation left--right to the pair of children of each internal node, with the number of descendant leaves decreasing from left to right) and then always choosing the left-most remaining labeled leaf, and at the end suppressing the orientations from the resulting trees.}

\begin{algorithm}[H] \label{alg_MinColless}
\SetAlgoLined
Start with a single node labeled $n$\;
\While{the current tree contains labeled leaves}{
 Choose a leaf with label $m$\;
	\If{$m$ is a power of 2}{
		replace this leaf by a fully symmetric  tree $T^{\mathit{fs}}_{\log_2(m)}$ with its nodes unlabeled\;
		}
	\Else{
		Find a pair of integers $(m_a, m_b) \in QB(m)$\;
		Split the leaf labeled $m$ into a cherry with unlabeled root and its leaves labeled $m_a$ and $m_b$, respectively.
	}
}
\caption{MinColless}
\end{algorithm}

\begin{example}\label{ex:2}
Let us use this Algorithm MinColless to find all minimal Colless trees with 20 leaves; we describe the trees by means of the usual Newick format,\footnote{See \url{http://evolution.genetics.washington.edu/phylip/newicktree.html}}   with the unlabeled leaves represented by a symbol $\cdot$ and omitting the semicolon ending mark in order not to confuse it with a punctuation mark. 
\begin{itemize}
\item[1)] We start with a single node labeled 20. 
\item[2)] Since $QB(20)=\{(10,10),(12,8)\}$, this node can split into the cherries $(10,10)$ and $(12,8)$.
\item[3.1)] Since $QB(10)=\{(5,5),(6,4)\}$, the different ways of splitting the leaves of the tree $(10,10)$ produce the trees $((5,5),(5,5))$, $((5,5),(6,4))$, and $((6,4),(6,4))$.
Now, since $QB(5)=\{(3,2)\}$, $QB(6)=\{(3,3),(4,2)\}$, and $QB(3)=\{(2,1)\}$, and 1, 2,  and 4 are powers of 2,
we have the following derivations from these trees through all possible combinations of splitting the leaves in the trees:
$$
\begin{array}{l}
((5,5),(5,5))  \Rightarrow (((3,2),(3,2)),((3,2),(3,2)))\\ \quad \Rightarrow ((((2,1),2),((2,1),2)),(((2,1),2),((2,1),2)))\\ \quad \Rightarrow 
(((((\cdot,\cdot),\cdot),(\cdot,\cdot)),(((\cdot,\cdot),\cdot),(\cdot,\cdot))),((((\cdot,\cdot),\cdot),(\cdot,\cdot)),(((\cdot,\cdot),\cdot),(\cdot,\cdot))))\\
((5,5),(6,4)) \Rightarrow (((3,2),(3,2)),((3,3),4))\\ \quad \Rightarrow  ((((2,1),2),((2,1),2)),(((2,1),(2,1)),4))\\ \quad\Rightarrow  (((((\cdot,\cdot),\cdot),(\cdot,\cdot)),(((\cdot,\cdot),\cdot),(\cdot,\cdot))),((((\cdot,\cdot),\cdot),((\cdot,\cdot),\cdot)),((\cdot,\cdot),(\cdot,\cdot)))\\
((5,5),(6,4)) \Rightarrow (((3,2),(3,2)),((4,2),4))\\ \quad \Rightarrow  ((((2,1),2),((2,1),2)),((4,2),4))\\ \quad\Rightarrow  (((((\cdot,\cdot),\cdot),(\cdot,\cdot)),(((\cdot,\cdot),\cdot),(\cdot,\cdot))),((((\cdot,\cdot),(\cdot,\cdot)),(\cdot,\cdot)),((\cdot,\cdot),(\cdot,\cdot))))
\end{array}
$$
$$
\begin{array}{l}
((6,4),(6,4)) \Rightarrow (((3,3),4),((3,3),4))\\ \quad \Rightarrow  ((((2,1),(2,1)),4),(((2,1),(2,1)),4))\\ \quad \Rightarrow  (((((\cdot,\cdot),\cdot),((\cdot,\cdot),\cdot)),((\cdot,\cdot),(\cdot,\cdot))),((((\cdot,\cdot),\cdot),((\cdot,\cdot),\cdot)),((\cdot,\cdot),(\cdot,\cdot))))\\
((6,4),(6,4)) \Rightarrow (((3,3),4),((4,2),4))\\ \quad \Rightarrow  ((((2,1),(2,1)),4),((4,2),4))\\ \quad  \Rightarrow  (((((\cdot,\cdot),\cdot),((\cdot,\cdot),\cdot)),((\cdot,\cdot),(\cdot,\cdot))),((((\cdot,\cdot),(\cdot,\cdot)),(\cdot,\cdot)),((\cdot,\cdot),(\cdot,\cdot))))\\
((6,4),(6,4))  \Rightarrow (((4,2),4),((4,2),4)) \\ \quad \Rightarrow  (((((\cdot,\cdot),(\cdot,\cdot)),(\cdot,\cdot)),((\cdot,\cdot),(\cdot,\cdot))),((((\cdot,\cdot),(\cdot,\cdot)),(\cdot,\cdot)),((\cdot,\cdot),(\cdot,\cdot))))
\end{array}
$$
\item[3.2)] Since $QB(12)=\{(6,6),(8,4)\}$ and 8 is a power of 2, the tree $(12,8)$ gives rise to the trees
$((6,6),8)$ and  $((8,4),8)$, and then, using $QB(6)=\{(3,3),(4,2)\}$ and $QB(3)=\{(2,1)\}$,
$$
\begin{array}{l}
((6,6),8)\Rightarrow (((3,3),(3,3)),8) \Rightarrow ((((2,1),(2,1)),((2,1),(2,1))),8) \\ \quad \Rightarrow (((((\cdot,\cdot),\cdot),((\cdot,\cdot),\cdot)),(((\cdot,\cdot),\cdot),((\cdot,\cdot),\cdot))),(((\cdot,\cdot),(\cdot,\cdot)),((\cdot,\cdot),(\cdot,\cdot))))\\
((6,6),8)  \Rightarrow (((3,3),(4,2)),8)\Rightarrow ((((2,1),(2,1)),(4,2)),8) \\ \quad \Rightarrow (((((\cdot,\cdot),\cdot),((\cdot,\cdot),\cdot)),(((\cdot,\cdot),(\cdot,\cdot)),(\cdot,\cdot))),(((\cdot,\cdot),(\cdot,\cdot)),((\cdot,\cdot),(\cdot,\cdot))))\\
((6,6),8)  \Rightarrow (((4,2),(4,2)),8)  \\ \quad\Rightarrow (((((\cdot,\cdot),(\cdot,\cdot)),(\cdot,\cdot)),(((\cdot,\cdot),(\cdot,\cdot)),(\cdot,\cdot))),(((\cdot,\cdot),(\cdot,\cdot)),((\cdot,\cdot),(\cdot,\cdot))))\\
((8,4),8)\\ \quad \Rightarrow (((((\cdot,\cdot),(\cdot,\cdot)),((\cdot,\cdot),(\cdot,\cdot))),((\cdot,\cdot),(\cdot,\cdot))),(((\cdot,\cdot),(\cdot,\cdot)),((\cdot,\cdot),(\cdot,\cdot))))
\end{array}
$$
\end{itemize}
So, there are 10 different minimal Colless trees in $\TT_{20}$. We depict them in Figure \ref{fig:1020} below.
\end{example}

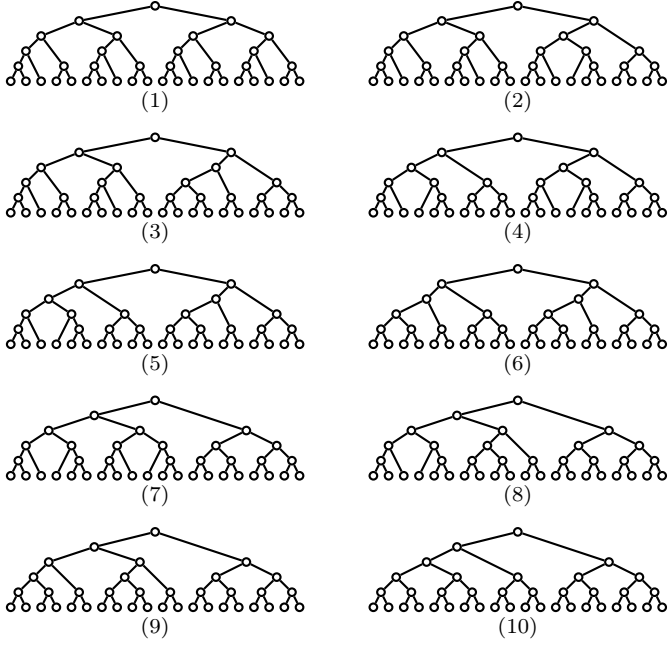
\begin{figure}[htb]
\begin{center}
\begin{tabular}{ccc}
\begin{tikzpicture}[thick,>=stealth,scale=0.2]
\draw(1,0) node[treppp] (l1) {}; 
\draw(2,0) node[treppp] (l2) {}; 
\draw(3,0) node[treppp] (l3) {}; 
\draw(4,0) node[treppp] (l4) {}; 
\draw(5,0) node[treppp] (l5) {}; 
\draw(6,0) node[treppp] (l6) {}; 
\draw(7,0) node[treppp] (l7) {}; 
\draw(8,0) node[treppp] (l8) {}; 
\draw(9,0) node[treppp] (l9) {}; 
\draw(10,0) node[treppp] (l10) {}; 
\draw(11,0) node[treppp] (l11) {}; 
\draw(12,0) node[treppp] (l12) {}; 
\draw(13,0) node[treppp] (l13) {}; 
\draw(14,0) node[treppp] (l14) {}; 
\draw(15,0) node[treppp] (l15) {}; 
\draw(16,0) node[treppp] (l16) {}; 
\draw(17,0) node[treppp] (l17) {}; 
\draw(18,0) node[treppp] (l18) {}; 
\draw(19,0) node[treppp] (l19) {}; 
\draw(20,0) node[treppp] (l20) {}; 
\draw(1.5,1) node[treppp] (v1) {}; 
\draw(4.5,1) node[treppp] (v2) {}; 
\draw(6.5,1) node[treppp] (v3) {}; 
\draw(9.5,1) node[treppp] (v4) {}; 
\draw(11.5,1) node[treppp] (v5) {}; 
\draw(14.5,1) node[treppp] (v6) {}; 
\draw(16.5,1) node[treppp] (v7) {}; 
\draw(19.5,1) node[treppp] (v8) {}; 
\draw(2,2) node[treppp] (w1) {}; 
\draw(7,2) node[treppp] (w2) {}; 
\draw(12,2) node[treppp] (w3) {}; 
\draw(17,2) node[treppp] (w4) {}; 
\draw(3,3) node[treppp] (x1) {}; 
\draw(8,3) node[treppp] (x2) {}; 
\draw(13,3) node[treppp] (x3) {}; 
\draw(18,3) node[treppp] (x4) {}; 
\draw(5.5,4) node[treppp] (y1) {}; 
\draw(15.5,4) node[treppp] (y2) {}; 
\draw(10.5,5) node[treppp] (r) {}; 
\draw (r)--(y1);
\draw (r)--(y2);
\draw (y1)--(x1);
\draw (y1)--(x2);
\draw (y2)--(x3);
\draw (y2)--(x4);
\draw (x1)--(w1);
\draw (x2)--(w2);
\draw (x3)--(w3);
\draw (x4)--(w4);
\draw (x1)--(v2);
\draw (x2)--(v4);
\draw (x3)--(v6);
\draw (x4)--(v8);
\draw (w1)--(v1);
\draw (w1)--(l3);
\draw (w2)--(v3);
\draw (w2)--(l8);
\draw (w3)--(v5);
\draw (w3)--(l13);
\draw (w4)--(v7);
\draw (w4)--(l18);
\draw (v1)--(l1);
\draw (v1)--(l2);
\draw (v2)--(l4);
\draw (v2)--(l5);
\draw (v3)--(l6);
\draw (v3)--(l7);
\draw (v4)--(l9);
\draw (v4)--(l10);
\draw (v5)--(l11);
\draw (v5)--(l12);
\draw (v6)--(l14);
\draw (v6)--(l15);
\draw (v7)--(l16);
\draw (v7)--(l17);
\draw (v8)--(l19);
\draw (v8)--(l20);
\end{tikzpicture}
&\quad&
\begin{tikzpicture}[thick,>=stealth,scale=0.2]
\draw(1,0) node[treppp] (l1) {}; 
\draw(2,0) node[treppp] (l2) {}; 
\draw(3,0) node[treppp] (l3) {}; 
\draw(4,0) node[treppp] (l4) {}; 
\draw(5,0) node[treppp] (l5) {}; 
\draw(6,0) node[treppp] (l6) {}; 
\draw(7,0) node[treppp] (l7) {}; 
\draw(8,0) node[treppp] (l8) {}; 
\draw(9,0) node[treppp] (l9) {}; 
\draw(10,0) node[treppp] (l10) {}; 
\draw(11,0) node[treppp] (l11) {}; 
\draw(12,0) node[treppp] (l12) {}; 
\draw(13,0) node[treppp] (l13) {}; 
\draw(14,0) node[treppp] (l14) {}; 
\draw(15,0) node[treppp] (l15) {}; 
\draw(16,0) node[treppp] (l16) {}; 
\draw(17,0) node[treppp] (l17) {}; 
\draw(18,0) node[treppp] (l18) {}; 
\draw(19,0) node[treppp] (l19) {}; 
\draw(20,0) node[treppp] (l20) {}; 
\draw(1.5,1) node[treppp] (v1) {}; 
\draw(4.5,1) node[treppp] (v2) {}; 
\draw(6.5,1) node[treppp] (v3) {}; 
\draw(9.5,1) node[treppp] (v4) {}; 
\draw(11.5,1) node[treppp] (v5) {}; 
\draw(15.5,1) node[treppp] (v6) {}; 
\draw(17.5,1) node[treppp] (v7) {}; 
\draw(19.5,1) node[treppp] (v8) {}; 
\draw(2,2) node[treppp] (w1) {}; 
\draw(7,2) node[treppp] (w2) {}; 
\draw(12,2) node[treppp] (w3) {}; 
\draw(15,2) node[treppp] (w4) {}; 
\draw(3,3) node[treppp] (x1) {}; 
\draw(8,3) node[treppp] (x2) {}; 
\draw(13.5,3) node[treppp] (x3) {}; 
\draw(18.5,2) node[treppp] (x4) {}; 
\draw(5.5,4) node[treppp] (y1) {}; 
\draw(15.5,4) node[treppp] (y2) {}; 
\draw(10.5,5) node[treppp] (r) {}; 
\draw (r)--(y1);
\draw (r)--(y2);
\draw (y1)--(x1);
\draw (y1)--(x2);
\draw (y2)--(x3);
\draw (y2)--(x4);
\draw (x1)--(w1);
\draw (x2)--(w2);
\draw (x3)--(w3);
\draw (x3)--(w4);
\draw (x1)--(v2);
\draw (x2)--(v4);
\draw (x4)--(v7);
\draw (x4)--(v8);
\draw (w1)--(v1);
\draw (w1)--(l3);
\draw (w2)--(v3);
\draw (w2)--(l8);
\draw (w3)--(v5);
\draw (w3)--(l13);
\draw (w4)--(v6);
\draw (w4)--(l14);
\draw (v1)--(l1);
\draw (v1)--(l2);
\draw (v2)--(l4);
\draw (v2)--(l5);
\draw (v3)--(l6);
\draw (v3)--(l7);
\draw (v4)--(l9);
\draw (v4)--(l10);
\draw (v5)--(l11);
\draw (v5)--(l12);
\draw (v6)--(l15);
\draw (v6)--(l16);
\draw (v7)--(l18);
\draw (v7)--(l17);
\draw (v8)--(l19);
\draw (v8)--(l20);
\end{tikzpicture}
\\[-0.5ex]
 (1) & & (2)\\[2ex]
\begin{tikzpicture}[thick,>=stealth,scale=0.2]
\draw(1,0) node[treppp] (l1) {}; 
\draw(2,0) node[treppp] (l2) {}; 
\draw(3,0) node[treppp] (l3) {}; 
\draw(4,0) node[treppp] (l4) {}; 
\draw(5,0) node[treppp] (l5) {}; 
\draw(6,0) node[treppp] (l6) {}; 
\draw(7,0) node[treppp] (l7) {}; 
\draw(8,0) node[treppp] (l8) {}; 
\draw(9,0) node[treppp] (l9) {}; 
\draw(10,0) node[treppp] (l10) {}; 
\draw(11,0) node[treppp] (l11) {}; 
\draw(12,0) node[treppp] (l12) {}; 
\draw(13,0) node[treppp] (l13) {}; 
\draw(14,0) node[treppp] (l14) {}; 
\draw(15,0) node[treppp] (l15) {}; 
\draw(16,0) node[treppp] (l16) {}; 
\draw(17,0) node[treppp] (l17) {}; 
\draw(18,0) node[treppp] (l18) {}; 
\draw(19,0) node[treppp] (l19) {}; 
\draw(20,0) node[treppp] (l20) {}; 
\draw(1.5,1) node[treppp] (v1) {}; 
\draw(4.5,1) node[treppp] (v2) {}; 
\draw(6.5,1) node[treppp] (v3) {}; 
\draw(9.5,1) node[treppp] (v4) {}; 
\draw(11.5,1) node[treppp] (v5) {}; 
\draw(13.5,1) node[treppp] (v6) {}; 
\draw(15.5,1) node[treppp] (v7) {}; 
\draw(17.5,1) node[treppp] (v8) {}; 
\draw(19.5,1) node[treppp] (v9) {}; 
\draw(2,2) node[treppp] (w1) {}; 
\draw(7,2) node[treppp] (w2) {}; 
\draw(12.5,2) node[treppp] (w3) {}; 
\draw(18.5,2) node[treppp] (w4) {}; 
\draw(3,3) node[treppp] (x1) {}; 
\draw(8,3) node[treppp] (x2) {}; 
\draw(14.5,3) node[treppp] (x3) {}; 
\draw(5.5,4) node[treppp] (y1) {}; 
\draw(15.5,4) node[treppp] (y2) {}; 
\draw(10.5,5) node[treppp] (r) {}; 
\draw (r)--(y1);
\draw (r)--(y2);
\draw (y1)--(x1);
\draw (y1)--(x2);
\draw (y2)--(x3);
\draw (y2)--(w4);
\draw (x1)--(w1);
\draw (x2)--(w2);
\draw (x3)--(w3);
\draw (x3)--(v7);
\draw (x1)--(v2);
\draw (x2)--(v4);
\draw (w1)--(v1);
\draw (w1)--(l3);
\draw (w2)--(v3);
\draw (w2)--(l8);
\draw (w3)--(v5);
\draw (w3)--(v6);
\draw (w4)--(v8);
\draw (w4)--(v9);
\draw (v1)--(l1);
\draw (v1)--(l2);
\draw (v2)--(l4);
\draw (v2)--(l5);
\draw (v3)--(l6);
\draw (v3)--(l7);
\draw (v4)--(l9);
\draw (v4)--(l10);
\draw (v5)--(l11);
\draw (v5)--(l12);
\draw (v6)--(l13);
\draw (v6)--(l14);
\draw (v7)--(l15);
\draw (v7)--(l16);
\draw (v8)--(l17);
\draw (v8)--(l18);
\draw (v9)--(l19);
\draw (v9)--(l20);
\end{tikzpicture}
&\quad&
\begin{tikzpicture}[thick,>=stealth,scale=0.2]
\draw(1,0) node[treppp] (l1) {}; 
\draw(2,0) node[treppp] (l2) {}; 
\draw(3,0) node[treppp] (l3) {}; 
\draw(4,0) node[treppp] (l4) {}; 
\draw(5,0) node[treppp] (l5) {}; 
\draw(6,0) node[treppp] (l6) {}; 
\draw(7,0) node[treppp] (l7) {}; 
\draw(8,0) node[treppp] (l8) {}; 
\draw(9,0) node[treppp] (l9) {}; 
\draw(10,0) node[treppp] (l10) {}; 
\draw(11,0) node[treppp] (l11) {}; 
\draw(12,0) node[treppp] (l12) {}; 
\draw(13,0) node[treppp] (l13) {}; 
\draw(14,0) node[treppp] (l14) {}; 
\draw(15,0) node[treppp] (l15) {}; 
\draw(16,0) node[treppp] (l16) {}; 
\draw(17,0) node[treppp] (l17) {}; 
\draw(18,0) node[treppp] (l18) {}; 
\draw(19,0) node[treppp] (l19) {}; 
\draw(20,0) node[treppp] (l20) {}; 
\draw(1.5,1) node[treppp] (v1) {}; 
\draw(5.5,1) node[treppp] (v2) {}; 
\draw(7.5,1) node[treppp] (v3) {}; 
\draw(9.5,1) node[treppp] (v4) {}; 
\draw(11.5,1) node[treppp] (v5) {}; 
\draw(15.5,1) node[treppp] (v6) {}; 
\draw(17.5,1) node[treppp] (v7) {}; 
\draw(19.5,1) node[treppp] (v8) {}; 
\draw(2,2) node[treppp] (w1) {}; 
\draw(5,2) node[treppp] (w2) {}; 
\draw(12,2) node[treppp] (w3) {}; 
\draw(15,2) node[treppp] (w4) {}; 
\draw(3.5,3) node[treppp] (x1) {}; 
\draw(8.5,2) node[treppp] (x2) {}; 
\draw(13.5,3) node[treppp] (x3) {}; 
\draw(18.5,2) node[treppp] (x4) {}; 
\draw(5.5,4) node[treppp] (y1) {}; 
\draw(15.5,4) node[treppp] (y2) {}; 
\draw(10.5,5) node[treppp] (r) {}; 
\draw (r)--(y1);
\draw (r)--(y2);
\draw (y1)--(x1);
\draw (y1)--(x2);
\draw (y2)--(x3);
\draw (y2)--(x4);
\draw (x1)--(w1);
\draw (x1)--(w2);
\draw (x2)--(v3);
\draw (x2)--(v4);
\draw (x3)--(w3);
\draw (x3)--(w4);
\draw (x4)--(v7);
\draw (x4)--(v8);
\draw (w1)--(v1);
\draw (w1)--(l3);
\draw (w2)--(v2);
\draw (w2)--(l4);
\draw (w3)--(v5);
\draw (w3)--(l13);
\draw (w4)--(v6);
\draw (w4)--(l14);
\draw (v1)--(l1);
\draw (v1)--(l2);
\draw (v2)--(l5);
\draw (v2)--(l6);
\draw (v3)--(l7);
\draw (v3)--(l8);
\draw (v4)--(l9);
\draw (v4)--(l10);
\draw (v5)--(l11);
\draw (v5)--(l12);
\draw (v6)--(l15);
\draw (v6)--(l16);
\draw (v7)--(l18);
\draw (v7)--(l17);
\draw (v8)--(l19);
\draw (v8)--(l20);
\end{tikzpicture}
\\[-0.5ex]
 (3) & & (4)\\[2ex]
\begin{tikzpicture}[thick,>=stealth,scale=0.2]
\draw(1,0) node[treppp] (l1) {}; 
\draw(2,0) node[treppp] (l2) {}; 
\draw(3,0) node[treppp] (l3) {}; 
\draw(4,0) node[treppp] (l4) {}; 
\draw(5,0) node[treppp] (l5) {}; 
\draw(6,0) node[treppp] (l6) {}; 
\draw(7,0) node[treppp] (l7) {}; 
\draw(8,0) node[treppp] (l8) {}; 
\draw(9,0) node[treppp] (l9) {}; 
\draw(10,0) node[treppp] (l10) {}; 
\draw(11,0) node[treppp] (l11) {}; 
\draw(12,0) node[treppp] (l12) {}; 
\draw(13,0) node[treppp] (l13) {}; 
\draw(14,0) node[treppp] (l14) {}; 
\draw(15,0) node[treppp] (l15) {}; 
\draw(16,0) node[treppp] (l16) {}; 
\draw(17,0) node[treppp] (l17) {}; 
\draw(18,0) node[treppp] (l18) {}; 
\draw(19,0) node[treppp] (l19) {}; 
\draw(20,0) node[treppp] (l20) {}; 
\draw(1.5,1) node[treppp] (v1) {}; 
\draw(5.5,1) node[treppp] (v2) {}; 
\draw(7.5,1) node[treppp] (v3) {}; 
\draw(9.5,1) node[treppp] (v4) {}; 
\draw(11.5,1) node[treppp] (v5) {}; 
\draw(13.5,1) node[treppp] (v6) {}; 
\draw(15.5,1) node[treppp] (v7) {}; 
\draw(17.5,1) node[treppp] (v8) {}; 
\draw(19.5,1) node[treppp] (v9) {}; 
\draw(2,2) node[treppp] (w1) {}; 
\draw(5,2) node[treppp] (w2) {}; 
\draw(12.5,2) node[treppp] (w3) {}; 
\draw(18.5,2) node[treppp] (w4) {}; 
\draw(3.5,3) node[treppp] (x1) {}; 
\draw(8.5,2) node[treppp] (x2) {}; 
\draw(14.5,3) node[treppp] (x3) {}; 
\draw(5.5,4) node[treppp] (y1) {}; 
\draw(15.5,4) node[treppp] (y2) {}; 
\draw(10.5,5) node[treppp] (r) {}; 
\draw (r)--(y1);
\draw (r)--(y2);
\draw (y1)--(x1);
\draw (y1)--(x2);
\draw (y2)--(x3);
\draw (y2)--(w4);
\draw (x1)--(w1);
\draw (x1)--(w2);
\draw (x2)--(v3);
\draw (x2)--(v4);
\draw (x3)--(w3);
\draw (x3)--(v7);
\draw (w1)--(v1);
\draw (w1)--(l3);
\draw (w2)--(v2);
\draw (w2)--(l4);
\draw (w3)--(v5);
\draw (w3)--(v6);
\draw (w4)--(v8);
\draw (w4)--(v9);
\draw (v1)--(l1);
\draw (v1)--(l2);
\draw (v2)--(l5);
\draw (v2)--(l6);
\draw (v3)--(l7);
\draw (v3)--(l8);
\draw (v4)--(l9);
\draw (v4)--(l10);
\draw (v5)--(l11);
\draw (v5)--(l12);
\draw (v6)--(l13);
\draw (v6)--(l14);
\draw (v7)--(l15);
\draw (v7)--(l16);
\draw (v8)--(l17);
\draw (v8)--(l18);
\draw (v9)--(l19);
\draw (v9)--(l20);
\end{tikzpicture}
&\quad&
\begin{tikzpicture}[thick,>=stealth,scale=0.2]
\draw(1,0) node[treppp] (l1) {}; 
\draw(2,0) node[treppp] (l2) {}; 
\draw(3,0) node[treppp] (l3) {}; 
\draw(4,0) node[treppp] (l4) {}; 
\draw(5,0) node[treppp] (l5) {}; 
\draw(6,0) node[treppp] (l6) {}; 
\draw(7,0) node[treppp] (l7) {}; 
\draw(8,0) node[treppp] (l8) {}; 
\draw(9,0) node[treppp] (l9) {}; 
\draw(10,0) node[treppp] (l10) {}; 
\draw(11,0) node[treppp] (l11) {}; 
\draw(12,0) node[treppp] (l12) {}; 
\draw(13,0) node[treppp] (l13) {}; 
\draw(14,0) node[treppp] (l14) {}; 
\draw(15,0) node[treppp] (l15) {}; 
\draw(16,0) node[treppp] (l16) {}; 
\draw(17,0) node[treppp] (l17) {}; 
\draw(18,0) node[treppp] (l18) {}; 
\draw(19,0) node[treppp] (l19) {}; 
\draw(20,0) node[treppp] (l20) {}; 
\draw(1.5,1) node[treppp] (v1) {}; 
\draw(3.5,1) node[treppp] (v2) {}; 
\draw(5.5,1) node[treppp] (v3) {}; 
\draw(7.5,1) node[treppp] (v4) {}; 
\draw(9.5,1) node[treppp] (v51) {}; 
\draw(11.5,1) node[treppp] (v5) {}; 
\draw(13.5,1) node[treppp] (v6) {}; 
\draw(15.5,1) node[treppp] (v7) {}; 
\draw(17.5,1) node[treppp] (v8) {}; 
\draw(19.5,1) node[treppp] (v9) {}; 
\draw(2.5,2) node[treppp] (w1) {}; 
\draw(8.5,2) node[treppp] (w2) {}; 
\draw(12.5,2) node[treppp] (w3) {}; 
\draw(18.5,2) node[treppp] (w4) {}; 
\draw(4.5,3) node[treppp] (x1) {}; 
\draw(14.5,3) node[treppp] (x3) {}; 
\draw(5.5,4) node[treppp] (y1) {}; 
\draw(15.5,4) node[treppp] (y2) {}; 
\draw(10.5,5) node[treppp] (r) {}; 
\draw (r)--(y1);
\draw (r)--(y2);
\draw (y1)--(x1);
\draw (y1)--(w2);
\draw (y2)--(x3);
\draw (y2)--(w4);
\draw (x1)--(w1);
\draw (x1)--(v3);
\draw (x3)--(w3);
\draw (x3)--(v7);
\draw (w1)--(v1);
\draw (w1)--(v2);
\draw (w2)--(v4);
\draw (w2)--(v51);
\draw (w3)--(v5);
\draw (w3)--(v6);
\draw (w4)--(v8);
\draw (w4)--(v9);
\draw (v1)--(l1);
\draw (v1)--(l2);
\draw (v2)--(l3);
\draw (v2)--(l4);
\draw (v3)--(l5);
\draw (v3)--(l6);
\draw (v4)--(l7);
\draw (v4)--(l8);
\draw (v51)--(l9);
\draw (v51)--(l10);
\draw (v5)--(l11);
\draw (v5)--(l12);
\draw (v6)--(l13);
\draw (v6)--(l14);
\draw (v7)--(l15);
\draw (v7)--(l16);
\draw (v8)--(l17);
\draw (v8)--(l18);
\draw (v9)--(l19);
\draw (v9)--(l20);
\end{tikzpicture}\\[-0.5ex]
 (5) & & (6)\\[2ex]
\begin{tikzpicture}[thick,>=stealth,scale=0.2]
\draw(1,0) node[treppp] (l1) {}; 
\draw(2,0) node[treppp] (l2) {}; 
\draw(3,0) node[treppp] (l3) {}; 
\draw(4,0) node[treppp] (l4) {}; 
\draw(5,0) node[treppp] (l5) {}; 
\draw(6,0) node[treppp] (l6) {}; 
\draw(7,0) node[treppp] (l7) {}; 
\draw(8,0) node[treppp] (l8) {}; 
\draw(9,0) node[treppp] (l9) {}; 
\draw(10,0) node[treppp] (l10) {}; 
\draw(11,0) node[treppp] (l11) {}; 
\draw(12,0) node[treppp] (l12) {}; 
\draw(13,0) node[treppp] (l13) {}; 
\draw(14,0) node[treppp] (l14) {}; 
\draw(15,0) node[treppp] (l15) {}; 
\draw(16,0) node[treppp] (l16) {}; 
\draw(17,0) node[treppp] (l17) {}; 
\draw(18,0) node[treppp] (l18) {}; 
\draw(19,0) node[treppp] (l19) {}; 
\draw(20,0) node[treppp] (l20) {}; 
\draw(1.5,1) node[treppp] (v1) {}; 
\draw(5.5,1) node[treppp] (v2) {}; 
\draw(7.5,1) node[treppp] (v3) {}; 
\draw(11.5,1) node[treppp] (v4) {}; 
\draw(13.5,1) node[treppp] (v6) {}; 
\draw(15.5,1) node[treppp] (v7) {}; 
\draw(17.5,1) node[treppp] (v8) {}; 
\draw(19.5,1) node[treppp] (v9) {}; 
\draw(2,2) node[treppp] (w1) {}; 
\draw(5,2) node[treppp] (w2) {}; 
\draw(8,2) node[treppp] (w11) {}; 
\draw(11,2) node[treppp] (w21) {}; 
\draw(14.5,2) node[treppp] (w3) {}; 
\draw(18.5,2) node[treppp] (w4) {}; 
\draw(3.5,3) node[treppp] (x1) {}; 
\draw(9.5,3) node[treppp] (x2) {}; 
\draw(6.5,4) node[treppp] (y1) {}; 
\draw(16.5,3) node[treppp] (y2) {}; 
\draw(10.5,5) node[treppp] (r) {}; 
\draw (r)--(y1);
\draw (r)--(y2);
\draw (y1)--(x1);
\draw (y1)--(x2);
\draw (y2)--(w3);
\draw (y2)--(w4);
\draw (x1)--(w1);
\draw (x1)--(w2);
\draw (x2)--(w11);
\draw (x2)--(w21);
\draw (w1)--(v1);
\draw (w1)--(l3);
\draw (w2)--(v2);
\draw (w2)--(l4);
\draw (w11)--(v3);
\draw (w11)--(l9);
\draw (w21)--(v4);
\draw (w21)--(l10);
\draw (w3)--(v6);
\draw (w3)--(v7);
\draw (w4)--(v8);
\draw (w4)--(v9);
\draw (v1)--(l1);
\draw (v1)--(l2);
\draw (v2)--(l5);
\draw (v2)--(l6);
\draw (v3)--(l7);
\draw (v3)--(l8);
\draw (v5)--(l11);
\draw (v5)--(l12);
\draw (v6)--(l13);
\draw (v6)--(l14);
\draw (v7)--(l15);
\draw (v7)--(l16);
\draw (v8)--(l17);
\draw (v8)--(l18);
\draw (v9)--(l19);
\draw (v9)--(l20);
\end{tikzpicture}
&\quad&
\begin{tikzpicture}[thick,>=stealth,scale=0.2]
\draw(1,0) node[treppp] (l1) {}; 
\draw(2,0) node[treppp] (l2) {}; 
\draw(3,0) node[treppp] (l3) {}; 
\draw(4,0) node[treppp] (l4) {}; 
\draw(5,0) node[treppp] (l5) {}; 
\draw(6,0) node[treppp] (l6) {}; 
\draw(7,0) node[treppp] (l7) {}; 
\draw(8,0) node[treppp] (l8) {}; 
\draw(9,0) node[treppp] (l9) {}; 
\draw(10,0) node[treppp] (l10) {}; 
\draw(11,0) node[treppp] (l11) {}; 
\draw(12,0) node[treppp] (l12) {}; 
\draw(13,0) node[treppp] (l13) {}; 
\draw(14,0) node[treppp] (l14) {}; 
\draw(15,0) node[treppp] (l15) {}; 
\draw(16,0) node[treppp] (l16) {}; 
\draw(17,0) node[treppp] (l17) {}; 
\draw(18,0) node[treppp] (l18) {}; 
\draw(19,0) node[treppp] (l19) {}; 
\draw(20,0) node[treppp] (l20) {}; 
\draw(1.5,1) node[treppp] (v1) {}; 
\draw(5.5,1) node[treppp] (v2) {}; 
\draw(7.5,1) node[treppp] (v3) {}; 
\draw(9.5,1) node[treppp] (v4) {}; 
\draw(11.5,1) node[treppp] (v5) {}; 
\draw(13.5,1) node[treppp] (v6) {}; 
\draw(15.5,1) node[treppp] (v7) {}; 
\draw(17.5,1) node[treppp] (v8) {}; 
\draw(19.5,1) node[treppp] (v9) {}; 
\draw(2,2) node[treppp] (w1) {}; 
\draw(5,2) node[treppp] (w2) {}; 
\draw(8.5,2) node[treppp] (w11) {}; 
\draw(14.5,2) node[treppp] (w3) {}; 
\draw(18.5,2) node[treppp] (w4) {}; 
\draw(3.5,3) node[treppp] (x1) {}; 
\draw(9.5,3) node[treppp] (x2) {}; 
\draw(6.5,4) node[treppp] (y1) {}; 
\draw(16.5,3) node[treppp] (y2) {}; 
\draw(10.5,5) node[treppp] (r) {}; 
\draw (r)--(y1);
\draw (r)--(y2);
\draw (y1)--(x1);
\draw (y1)--(x2);
\draw (y2)--(w3);
\draw (y2)--(w4);
\draw (x1)--(w1);
\draw (x1)--(w2);
\draw (x2)--(w11);
\draw (x2)--(v5);
\draw (w1)--(v1);
\draw (w1)--(l3);
\draw (w2)--(v2);
\draw (w2)--(l4);
\draw (w11)--(v3);
\draw (w11)--(v4);
\draw (w3)--(v6);
\draw (w3)--(v7);
\draw (w4)--(v8);
\draw (w4)--(v9);
\draw (v1)--(l1);
\draw (v1)--(l2);
\draw (v2)--(l5);
\draw (v2)--(l6);
\draw (v3)--(l7);
\draw (v3)--(l8);
\draw (v4)--(l9);
\draw (v4)--(l10);
\draw (v5)--(l11);
\draw (v5)--(l12);
\draw (v6)--(l13);
\draw (v6)--(l14);
\draw (v7)--(l15);
\draw (v7)--(l16);
\draw (v8)--(l17);
\draw (v8)--(l18);
\draw (v9)--(l19);
\draw (v9)--(l20);
\end{tikzpicture}\\[-0.5ex]
 (7) & & (8)\\[2ex]
\begin{tikzpicture}[thick,>=stealth,scale=0.2]
\draw(1,0) node[treppp] (l1) {}; 
\draw(2,0) node[treppp] (l2) {}; 
\draw(3,0) node[treppp] (l3) {}; 
\draw(4,0) node[treppp] (l4) {}; 
\draw(5,0) node[treppp] (l5) {}; 
\draw(6,0) node[treppp] (l6) {}; 
\draw(7,0) node[treppp] (l7) {}; 
\draw(8,0) node[treppp] (l8) {}; 
\draw(9,0) node[treppp] (l9) {}; 
\draw(10,0) node[treppp] (l10) {}; 
\draw(11,0) node[treppp] (l11) {}; 
\draw(12,0) node[treppp] (l12) {}; 
\draw(13,0) node[treppp] (l13) {}; 
\draw(14,0) node[treppp] (l14) {}; 
\draw(15,0) node[treppp] (l15) {}; 
\draw(16,0) node[treppp] (l16) {}; 
\draw(17,0) node[treppp] (l17) {}; 
\draw(18,0) node[treppp] (l18) {}; 
\draw(19,0) node[treppp] (l19) {}; 
\draw(20,0) node[treppp] (l20) {}; 
\draw(1.5,1) node[treppp] (v1) {}; 
\draw(3.5,1) node[treppp] (v2) {}; 
\draw(5.5,1) node[treppp] (v3) {}; 
\draw(7.5,1) node[treppp] (v4) {}; 
\draw(9.5,1) node[treppp] (v45) {}; 
\draw(11.5,1) node[treppp] (v5) {}; 
\draw(13.5,1) node[treppp] (v6) {}; 
\draw(15.5,1) node[treppp] (v7) {}; 
\draw(17.5,1) node[treppp] (v8) {}; 
\draw(19.5,1) node[treppp] (v9) {}; 
\draw(2.5,2) node[treppp] (w1) {}; 
\draw(8.5,2) node[treppp] (w11) {}; 
\draw(14.5,2) node[treppp] (w3) {}; 
\draw(18.5,2) node[treppp] (w4) {}; 
\draw(3.5,3) node[treppp] (x1) {}; 
\draw(9.5,3) node[treppp] (x2) {}; 
\draw(6.5,4) node[treppp] (y1) {}; 
\draw(16.5,3) node[treppp] (y2) {}; 
\draw(10.5,5) node[treppp] (r) {}; 
\draw (r)--(y1);
\draw (r)--(y2);
\draw (y1)--(x1);
\draw (y1)--(x2);
\draw (y2)--(w3);
\draw (y2)--(w4);
\draw (x1)--(w1);
\draw (x1)--(v3);
\draw (x2)--(w11);
\draw (x2)--(v5);
\draw (w1)--(v1);
\draw (w1)--(v2);
\draw (w11)--(v4);
\draw (w11)--(v45);
\draw (w3)--(v6);
\draw (w3)--(v7);
\draw (w4)--(v8);
\draw (w4)--(v9);
\draw (v1)--(l1);
\draw (v1)--(l2);
\draw (v2)--(l3);
\draw (v2)--(l4);
\draw (v3)--(l5);
\draw (v3)--(l6);
\draw (v4)--(l7);
\draw (v4)--(l8);
\draw (v45)--(l9);
\draw (v45)--(l10);
\draw (v5)--(l11);
\draw (v5)--(l12);
\draw (v6)--(l13);
\draw (v6)--(l14);
\draw (v7)--(l15);
\draw (v7)--(l16);
\draw (v8)--(l17);
\draw (v8)--(l18);
\draw (v9)--(l19);
\draw (v9)--(l20);
\end{tikzpicture}
&\quad&
\begin{tikzpicture}[thick,>=stealth,scale=0.2]
\draw(1,0) node[treppp] (l1) {}; 
\draw(2,0) node[treppp] (l2) {}; 
\draw(3,0) node[treppp] (l3) {}; 
\draw(4,0) node[treppp] (l4) {}; 
\draw(5,0) node[treppp] (l5) {}; 
\draw(6,0) node[treppp] (l6) {}; 
\draw(7,0) node[treppp] (l7) {}; 
\draw(8,0) node[treppp] (l8) {}; 
\draw(9,0) node[treppp] (l9) {}; 
\draw(10,0) node[treppp] (l10) {}; 
\draw(11,0) node[treppp] (l11) {}; 
\draw(12,0) node[treppp] (l12) {}; 
\draw(13,0) node[treppp] (l13) {}; 
\draw(14,0) node[treppp] (l14) {}; 
\draw(15,0) node[treppp] (l15) {}; 
\draw(16,0) node[treppp] (l16) {}; 
\draw(17,0) node[treppp] (l17) {}; 
\draw(18,0) node[treppp] (l18) {}; 
\draw(19,0) node[treppp] (l19) {}; 
\draw(20,0) node[treppp] (l20) {}; 
\draw(1.5,1) node[treppp] (v1) {}; 
\draw(3.5,1) node[treppp] (v2) {}; 
\draw(5.5,1) node[treppp] (v3) {}; 
\draw(7.5,1) node[treppp] (v4) {}; 
\draw(9.5,1) node[treppp] (v45) {}; 
\draw(11.5,1) node[treppp] (v5) {}; 
\draw(13.5,1) node[treppp] (v6) {}; 
\draw(15.5,1) node[treppp] (v7) {}; 
\draw(17.5,1) node[treppp] (v8) {}; 
\draw(19.5,1) node[treppp] (v9) {}; 
\draw(2.5,2) node[treppp] (w1) {}; 
\draw(6.5,2) node[treppp] (w2) {}; 
\draw(10.5,2) node[treppp] (w11) {}; 
\draw(14.5,2) node[treppp] (w3) {}; 
\draw(18.5,2) node[treppp] (w4) {}; 
\draw(4.5,3) node[treppp] (x1) {}; 
\draw(6.5,4) node[treppp] (y1) {}; 
\draw(16.5,3) node[treppp] (y2) {}; 
\draw(10.5,5) node[treppp] (r) {}; 
\draw (r)--(y1);
\draw (r)--(y2);
\draw (y1)--(x1);
\draw (y1)--(w11);
\draw (y2)--(w3);
\draw (y2)--(w4);
\draw (x1)--(w1);
\draw (x1)--(w2);
\draw (w1)--(v1);
\draw (w1)--(v2);
\draw (w2)--(v3);
\draw (w2)--(v4);
\draw (w11)--(v45);
\draw (w11)--(v5);
\draw (w3)--(v6);
\draw (w3)--(v7);
\draw (w4)--(v8);
\draw (w4)--(v9);
\draw (v1)--(l1);
\draw (v1)--(l2);
\draw (v2)--(l3);
\draw (v2)--(l4);
\draw (v3)--(l5);
\draw (v3)--(l6);
\draw (v4)--(l7);
\draw (v4)--(l8);
\draw (v45)--(l9);
\draw (v45)--(l10);
\draw (v5)--(l11);
\draw (v5)--(l12);
\draw (v6)--(l13);
\draw (v6)--(l14);
\draw (v7)--(l15);
\draw (v7)--(l16);
\draw (v8)--(l17);
\draw (v8)--(l18);
\draw (v9)--(l19);
\draw (v9)--(l20);
\end{tikzpicture}\\[-0.5ex]
 (9) & & (10)
 \end{tabular}
\end{center}
\caption{\label{fig:1020} The 10 trees in $\TT_{20}$ with minimum Colless index, 8. They are enumerated in the same order as they have been produced in Example \ref{ex:2}.}
\end{figure}

We have implemented Algorithm MinColless, with step 8 efficiently carried out by means of Proposition \ref{cor:QB}, in a Python script that generates, for every $n$, the Newick description of all minimal Colless trees in $\TT_n$. It is available at the GitHub repository \url{https://github.com/biocom-uib/Colless}. As a proof of concept, we have computed for every $n$ from 1 to  128 all such minimal Colless trees in $\TT_n$. Figure \ref{Fig_NumberCollessMinima} shows their number for every $n$. These numbers are in agreement with those provided by the recurrence established in Proposition \ref{Thm_number_Colless} in the next subsection.

\begin{figure}[htb]
	\centering
	\includegraphics[width=\linewidth]{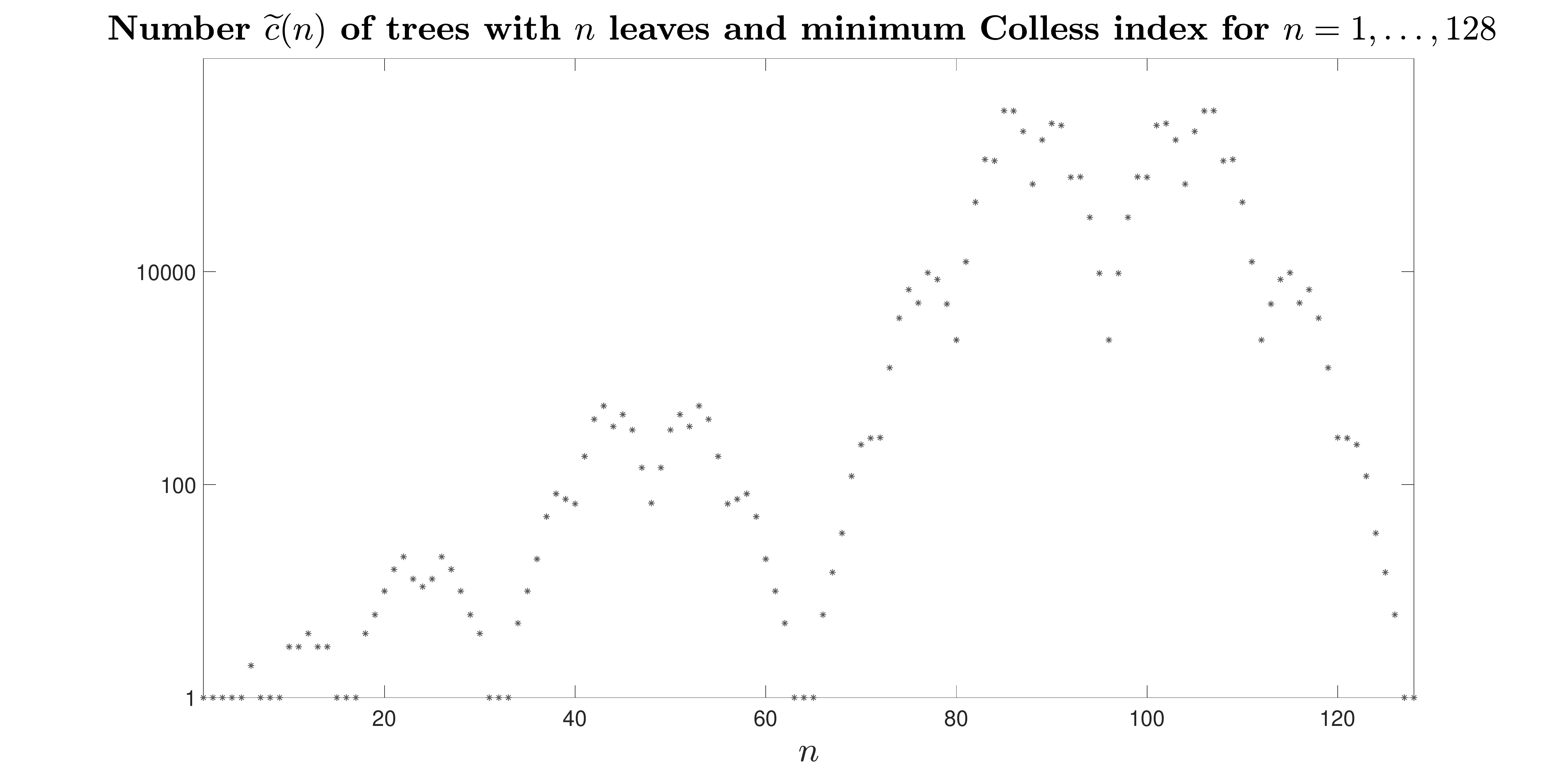}
	\caption{Plot of $\widetilde{c}(n)$  for $n=1, \ldots, 128$.}
	\label{Fig_NumberCollessMinima}
\end{figure}

\subsection{Counting minimal Colless trees}
Let $\widetilde{\mathcal{MC}}_n$ denote the set of all minimal Colless trees in $\TT_n$  and   $\widetilde{c}(n)\coloneqq\big|\widetilde{\mathcal{MC}}_n\big |$ its cardinality. 
To simplify the notations, set
$$
\widetilde{QB}(n)\coloneqq\{(n_a,n_b)\in QB(n)\mid n_a>n_b\}.
$$
We have the following recursive formula for $\widetilde{c}(n)$:

\begin{proposition} \label{Thm_number_Colless}
The sequence $\widetilde{c}(n)$ satisfies that  $\widetilde{c}(1) = 1$ and, for every $n\geq 2$,
$$
\widetilde{c}(n) =\hspace*{-2ex} \sum_{(n_a,n_b)\in \widetilde{QB}(n)}\hspace*{-4ex}   \widetilde{c}(n_a)\cdot \widetilde{c}(n_b) +\binom{\widetilde{c}(n/2)+1}{2}\cdot \delta_{\mathit{even}}(n)
$$
where $\delta_{\mathit{even}}(n)=1$ if $n$ is even and 0 otherwise.
\end{proposition}

\begin{proof}
By Lemma \ref{max_subtrees} and Proposition \ref{lem:charmin1}, $T=(T_a,T_b)\in \widetilde{\mathcal{MC}}_n$ if, and only if, 
$(n_a,n_b)\in QB(n)$, $T_a\in \widetilde{\mathcal{MC}}_{n_a}$ and $T_b\in \widetilde{\mathcal{MC}}_{n_b}$.  
The correctness of the formula in the statement stems then from the following  facts:
\begin{itemize}
\item If $n$ is odd, $\widetilde{\mathcal{MC}}_n$ is in bijection with the set
$$
X_n=\big\{(n_a,n_b,T_a,T_b)\mid  (n_a,n_b)\in \widetilde{QB}(n), T_a\in \widetilde{\mathcal{MC}}_{n_a}, T_b\in \widetilde{\mathcal{MC}}_{n_b}\big\},
$$
through the relation 
$$
T=(T_a,T_b)\in \widetilde{\mathcal{MC}}_n\Longleftrightarrow (n_a,n_b,T_a,T_b)\in X_n.
$$

\item If $n$ is even, $\widetilde{\mathcal{MC}}_n$ is in bijection with the set
$$
X_n \sqcup \big\{\{T_a,T_b\}\mid T_a,T_b\in \widetilde{\mathcal{MC}}_{n/2}, T_a\neq T_b\big\}
 \sqcup
\big\{T_a\mid T_a\in  \widetilde{\mathcal{MC}}_{n/2}\big\}
$$
through the relation 
$$
T=(T_a,T_b)\in \widetilde{\mathcal{MC}}_n\Longleftrightarrow \begin{cases}
n_a>n_b\mbox{ and }(n_a,n_b,T_a,T_b)\in X_n,\mbox{ or}\\
n_a=n_b, T_a\neq T_b, \mbox{ and }T_a,T_b\in \widetilde{\mathcal{MC}}_{n/2},\mbox{ or}\\
n_a=n_b, T_a=T_b\in \widetilde{\mathcal{MC}}_{n/2}\end{cases}
$$

\item The cardinality of $X_n$ is $\sum\limits_{(n_a,n_b)\in \widetilde{QB}(n)}\hspace*{-4ex}   \widetilde{c}(n_a)\cdot \widetilde{c}(n_b) $ and the cardinality of
$$
 \big\{\{T_a,T_b\}\mid T_a,T_b\in \widetilde{\mathcal{MC}}_{n/2}, T_a\neq T_b\big\}
\cup
\big\{T_a\mid T_a\in  \widetilde{\mathcal{MC}}_{n/2}\big\}
$$
is $\binom{\widetilde{c}(n/2)+1}{2}$. \qed
\end{itemize}
\end{proof}

The sequence $\widetilde{c}(n)$ seems to be new in the literature, and it has been added to the \textsl{Online Encyclopedia of Integer Sequences} \citep{OEIS} as sequence A307689. It would definitely be of interest to find an explicit formula for $\widetilde{c}(n)$ and to analyze  the fractal structure  suggested by Figure \ref{Fig_NumberCollessMinima}, which continues for larger values of $n$ and  seems also related to the Blancmange curve (compare Figure \ref{Fig_NumberCollessMinima} with Figure \ref{Fig_MinimumColless}).

{
\begin{remark}\label{rem:widehat}
Recall that, as we mentioned in the Introduction, even though the Colless index is mainly used to study the shape of \emph{phylogenetic trees} (i.e. of leaf-labeled trees where the leaf labels may for example correspond to some extant taxa or any other Operational Taxonomic Units) in the present manuscript we deal with \emph{unlabeled trees}. This decision was simply due to the fact that the Colless index only depends on the topology of the tree and not on the actual taxa labeling its leaves, and therefore, in particular, the fact that a phylogenetic tree achieves or does not achieve the minimum Colless index for its number of leaves does not depend on its actual labels. When counting minimal Colless trees, however, it might be of interest not only to count the number of minimal Colless tree topologies, but also to count the number of minimal Colless phylogenetic trees on a given set of $n$ taxa. 

Now, using some combinatorial arguments (in particular the famous Burnside's lemma) it can be shown that, for any given bifurcating tree $T \in  \TT_n$, there are $n!/2^{s(T)}$ many \emph{phylogenetic trees on $n$ leaves} ---that is, phylogenetic trees  with their leaves bijectively labeled by $\{1,\ldots,n\}$--- of this shape, where $s(T)$ denotes the number of symmetry vertices of $T$; see, for instance,  Corollary 2.4.3 in \citep{Semple2003}. 
Let $\widehat{c}(n)$ denote the number of phylogenetic trees on $n$ leaves that have minimum Colless index. Then, we can formally calculate this number as the sum, over all unlabeled minimal Colless trees $T \in \widetilde{\mathcal{MC}}_n$, of the number of phylogenetic trees on $n$ leaves that have shape $T$:
$$ 
\widehat{c}(n) = \sum_{T \in \widetilde{\mathcal{MC}}_n} \frac{n!}{2^{s(T)}}.
$$
Unfortunately, we have not been able to find even a recurrence for this sequence. We shall return to it in Remark \ref{rem:maxGFB}.
\end{remark}}

\subsection{{Greedy from the bottom trees: another particular} family of minimal Colless trees} \label{sec:GFB}

As we have seen in Theorem \ref{thm:minC}, the maximally balanced trees $T_n^{\mathit{mb}}$ have minimum Colless index. 
These trees are obtained through the recursive strategy suggested by Corollary \ref{colless_minimum}: given a number $n$ of leaves, we split $n$ into $n_a = \lceil n/2 \rceil$ and $n_b = \lfloor n/2 \rfloor$ and we produce a tree  $T=(T_a, T_b)$ with $T_a\in \TT_{n_a}$ and $T_b\in \TT_{n_b}$ constructed recursively through the same procedure. This strategy could be understood to be \enquote{greedy from the top} because, starting at the root and going towards the leaves, we bipartition the leaf set of each rooted subtree into two sets so that the difference of their cardinalities is minimized.  

There is another strategy for building minimal Colless trees, which we call \enquote{greedy from the bottom}, where instead of minimally splitting the sets of leaves, one minimally joins rooted subtrees by pending them from a common parent of their roots, as in the coalescent process \citep{kingman82}. More specifically, these trees are constructed by means of the following algorithm:\medskip

\begin{algorithm}[H] \label{alg_gfb}
\SetAlgoLined
$n \leftarrow$ number of taxa\;
$\treeset \leftarrow n$  trees consisting of one node each\;
$min \leftarrow 1$\  \texttt{//least number of leaves of all trees in treeset}\;
\While{ $\vert \treeset \vert >1$}{
	$u \leftarrow$ tree from $\treeset$ with $min$ leaves\;
	$\treeset = \treeset \setminus \{u\}$\;
	$min \leftarrow$ minimal number of leaves of all trees in $\treeset$\;
	$v \leftarrow$ tree from $\treeset$ with $min$ leaves\;
	$\treeset = \treeset \setminus \{v\}$\;
	$newtree \leftarrow$ tree consisting of new root $\rho_{uv}$ and maximal pending subtrees $u$ and $v$\;
	$\treeset \leftarrow \treeset \cup \{newtree\}$\;
	$min \leftarrow$ minimal number of leaves of all trees in $\treeset$\;
}
$finaltree \leftarrow \treeset[1]$\ \texttt{//the only remaining element of treeset}\;
\Return $finaltree$\;
\caption{Greedy from the bottom}
\label{gfb}
\end{algorithm}\medskip

We shall call henceforth {any bifurcating tree with $n$ leaves that results from Algorithm \ref{gfb} \emph{greedy from the bottom}, or simply \emph{GFB}}, and we shall denote it by $T^{\mathit{gfb}}_n$. This notation leads to no ambiguity, because of the following lemma. 

\begin{lemma}\label{cor:GFBunique}
For every $n\geq 1$, there exists only one GFB tree with $n$ leaves (up to isomorphisms).
\end{lemma}

\begin{proof}
When $n=1$, Algorithm \ref{alg_gfb} skips the \textbf{while} loop and it returns the only tree in $\TT_1$. Assume now that $n\geq 2$. With the notations of Algorithm \ref{alg_gfb}, let us denote by  $\treeset_k$, for $k=1,\ldots,n-1$, the content of the auxiliary tree multiset  $\treeset$ after the $k$-th iteration of the \textbf{while} loop. 
We shall prove by induction on $k$ that, for every two applications of Algorithm \ref{gfb} with input $n$ (whose \emph{treesets} will be distinguished henceforth with superscripts $(1)$ and $(2)$):
\begin{enumerate}[(a)] 
\item We have the equality of tree multisets
$\treeset_k^{(1)}={\treeset}_k^{(2)}$, which means that these two multisets of trees have the same elements with the same multiplicities;  and 
\item  For every $2\leq m\leq n$, all trees with $m$ leaves created in the first $k$  iterations of the loop in both applications of the algorithm have the same shape. 
\end{enumerate}
This will imply that when, after $n-1$ iterations of the loop, both multisets $\mathit{treeset}^{(i)}_{n-1}$, $i=1,2$, consist of a single tree with $n$ leaves, these two trees are the same.

The base case $k=1$ is obvious, because $\mathit{treeset}_1$ always consists of a cherry and $n-2$ isolated nodes.
Assume now that the statement is true for the $(k-1)$-th iteration, and in particular that, immediately before the $k$-th iteration, $\treeset_{k-1}^{(1)}={\treeset}_{k-1}^{(2)}$ (by (a)) and this multiset contains trees of only one shape for each present number of leaves (by (b)). This implies that the minimal number of leaves of a tree in $\treeset_{k-1}^{(1)}$ and ${\treeset}_{k-1}^{(2)}$ is the same, let us call it $m_1$, and that all trees with $m_1$ leaves in both $\mathit{treeset}$ have the same shape. Moreover,  if we remove one tree with $m_1$ leaves  from each $\treeset$ (which will be the same tree ---up to isomorphisms--- in both applications of the algorithm), the resulting multisets are equal again, and therefore the minimal number of leaves of a tree in each one of them is again the same, let us call it $m_2$, and all trees with $m_2$ leaves in both multisets are equal. Then, in the $k$-th iteration of the loop in each application of the algorithm, we remove from the corresponding $\mathit{treeset}$ the same tree with $m_1$ leaves and the same tree with $m_2$ leaves and we add the same tree with $m_1+m_2$ leaves, obtained by pending the removed trees to a common root. This proves that $\treeset_{k}^{(1)}={\treeset}_{k}^{(2)}$, i.e. assertion (a). 

To prove that (b) also holds, it remains to check that if some $\treeset_{j}^{(1)}$ with $j\leq k-1$ already contained some tree $T'$ with $m_1+m_2$ leaves, then it has the same shape as the new one. 
Assume that such a tree $T'$ with $m_1+m_2$ leaves has been created in the $j$-th iteration of the loop. Let $m_1'$ and $m_2'$, with $m_1'\leq m_2'$, be the numbers of leaves of the maximal pending subtrees of $T'$. By construction, this means that the minimal number of leaves of any tree in the multiset $\treeset_{j-1}^{(1)}$ was $m_1'$, and the second minimal number of leaves
was $m_2'$. Now, remember that, in each iteration of the loop, two trees are removed from the $\treeset$ and replaced by a tree with number of leaves the sum of the numbers of leaves of the removed trees. This clearly implies that the minimal and second minimal numbers of leaves of members of the $\treeset$ cannot decrease in any such iteration.
Therefore, $m_1'\leq m_1$, because if $m_1< m_1'$, then $\treeset_{j-1}^{(1)}$ cannot contain any tree with $m_1$ leaves (as $m_1'$ is the minimal number of leaves of a member of $\treeset_{j-1}^{(1)}$) and such a tree cannot be added in further iterations of the loop, but there is at least one such tree in $\treeset_{k-1}^{(1)}$. Since $m_1+m_2=m_1'+m_2'$, if $m_1'< m_1$ then $m_2'> m_2$, but a similar argument shows that this inequality is in contradiction with the fact that $m_2'$ is the smallest number of leaves of a tree in $\treeset_{j-1}^{(1)}$ after removing a tree with $m_1'$ leaves.
Therefore, $m_1'= m_1$ and hence $m_2'= m_2$, too. But then, by (b) in  the induction hypothesis, the trees with $m_1$ and $m_2$ leaves combined in the $j$-th iteration of the first application of the algorithm have the same shape as the trees with $m_1$ and $m_2$ leaves combined in the $k$-th iteration, and therefore the tree with $m_1+m_2$ leaves that already existed in $\treeset_{j}^{(1)}$ has the same shape as the one added in the $k$-th iteration. This completes the proof of the inductive step. \qed
\end{proof}

Note that Algorithm \ref{alg_gfb} greedily clusters trees of minimal numbers of leaves starting with single nodes and proceeding until only one tree is left, which is the reason we call the resulting trees ``greedy from the bottom.'' Our main goal in this subsection is to prove that they are also minimal Colless and, in general, different from the maximally balanced trees with the same number of leaves  (cf. Figure \ref{Fig_ExampleGFB}). 

Next result easily implies  that any rooted subtree of a  GFB tree is also a GFB tree, by induction on the depth of the subtree's root.

\begin{lemma} \label{gfb_subtrees}
If $T=(T_a,T_b)$ is a GFB tree, then $T_a$ and $T_b$ are also GFB trees.
\end{lemma}

\begin{proof}
Let $T = (T_a,T_b)$ be a GFB tree and let $n_a$ and $n_b$ denote the numbers of leaves of $T_a$ and $T_b$, respectively.
This entails that Algorithm \ref{alg_gfb} induces a bipartition of the $n$ leaves into two disjoint sets of sizes $n_a$ and $n_b$, respectively, in the sense  that all iterations  of the \textbf{while} loop except for the very last one  combine  pairs of subtrees with both sets of leaves contained  either in $V_L(T_a)$ or in $V_L(T_b)$.

Now, when in an iteration of the algorithm a pair of subtrees of $T_a$ is combined, it is because their numbers of leaves are the two smallest ones in the global $\treeset$, and hence also in the submultiset of $\treeset$ consisting only of trees with leaves in $V_L(T_a)$. This shows that $T_a$ is obtained through the application of Algorithm \ref{alg_gfb} to $n_a$ leaves, i.e. $T_a=T^{\mathit{gfb}}_{n_a}$, and by symmetry $T_b=T^{\mathit{gfb}}_{n_b}$. \qed
\end{proof}

The next proposition characterizes the pairs of numbers of leaves of the maximal pending subtrees of a GFB tree. Besides allowing the construction of GFB trees through an alternative top-to-bottom procedure, by  splitting clusters into subclusters of suitable sizes, this characterization easily entails that the GFB trees almost never are maximally balanced, and moreover it will allow us to use Proposition \ref{lem:charmin1} to prove that the GFB trees are minimal Colless (see Theorem \ref{GFB_is_minimal} below). 

\begin{proposition} \label{GFB_Decomposition}
Let $T_n^{\mathit{gfb}}=(T_a,T_b)$ be a GFB tree with $n\geq 2$, $T_a\in \TT_{n_a}$, $T_b\in \TT_{n_b}$ and $n_a\geq n_b$. Let $n=2^m+p$ with $m=\lfloor \log_2(n)\rfloor$ and $0\leq p<2^m$. Then, we have:
	\begin{enumerate}[\rm (i)]
	\item If $0\leq p\leq 2^{m-1}$, then $n_a = 2^{m-1}+p$, $n_b = 2^{m-1}$ and $T_b$ is fully symmetric. 
	\item If $2^{m-1}\leq p<2^m$, $n_a = 2^{m}$, $n_b=p$ and $T_a$ is fully symmetric.
	\end{enumerate}
\end{proposition}

Since the proof of this proposition is quite long, we postpone it until Appendix A.2 at the end of the manuscript. 

\begin{remark}
We want to point out here that for the proof of Proposition \ref{GFB_Decomposition} provided in Appendix A.2, we derive a technical lemma (Lemma \ref{subtree_in_common}) stating that if $n\geq 3$ is any odd number of leaves, then the GFB trees $T_{n-1}^{\mathit{gfb}}$, $T_{n}^{\mathit{gfb}}$, and $T_{n+1}^{\mathit{gfb}}$ have a maximal pending subtree in common, which is moreover fully symmetric. Using that the maximal pending subtrees of a GFB tree are again GFB (Lemma \ref{gfb_subtrees}), their explicit numbers of leaves provided by Proposition  \ref{GFB_Decomposition}, and the next proposition, which clearly implies that the GFB trees with numbers of leaves that are powers of 2 must be fully symmetric, the thesis of Lemma \ref{subtree_in_common}  is easily extended to the GFB trees  $T_{n-1}^{\mathit{gfb}}$, $T_{n}^{\mathit{gfb}}$, and $T_{n+1}^{\mathit{gfb}}$ for any number of leaves $n$ that is not of the form $3\times 2^m$.
\end{remark}

Now, as we announced, we  use Proposition  \ref{GFB_Decomposition} to prove  that the GFB trees always have minimum Colless index:

\begin{proposition} \label{GFB_is_minimal}
Let $T_n^{\mathit{gfb}}$ be the GFB tree with $n$ leaves. Then, $\mathcal{C}(T_n^{\mathit{gfb}}) = c_n$.
\end{proposition}

\begin{proof}
We prove that $T_n^{\mathit{gfb}}$ is Colless minimal by induction on the number of leaves $n$. 
The base case $n = 1$ is obvious, because there is only one tree in $\TT_1$.
Assume now that every GFB tree with at most $n-1$ leaves is Colless minimal and consider the tree $T_n^{\mathit{gfb}}$.   
By Lemma \ref{gfb_subtrees}, if $T_n^{\mathit{gfb}}=(T_a,T_b)$, then $T_a$ and $T_b$ are GFB trees and then, by the induction hypothesis, they  are Colless minimal and in particular $\mathcal{C}(T_a)=c_{n_a}$ and $\mathcal{C}(T_b)=c_{n_b}$. Let us write $n$ as $2^m+p$  with $m=\lfloor \log_2(n)\rfloor$ and $0\leq p<2^m$, and  consider its binary expansion
$n = \sum_{j=1}^\ell 2^{m_j}$ with $m_1>\cdots>m_\ell$, so that $m_1=m$ and $p = \sum_{j=2}^\ell 2^{m_j}$ is the binary expansion of $p$ if $p>0$.
Now:
\begin{enumerate}[(i)]
\item If $p=0$, then, by Proposition \ref{GFB_Decomposition}, $n_a =n_b = 2^{m-1}$, and then, by Lemma \ref{colless_sum}
and the induction hypothesis,
$$
\mathcal{C}(T^{\mathit{gfb}}_{n}) = \mathcal{C}(T^{\mathit{gfb}}_{n_a})+\mathcal{C}(T^{\mathit{gfb}}_{n_b})+n_a-n_b
=c_{n_a}+c_{n_b}+n_a-n_b=0=c_n.
$$

\item If $1\leq p<2^{m-1}$, then, by Proposition \ref{GFB_Decomposition}, $n_a = 2^{m-1} + p$ and $n_b = 2^{m-1}$. In this case, $m_2<m-1=m_1-1$ and thus $n_a = 2^{m_1 - 1} + \sum_{j=2}^\ell 2^{m_j}$ is the binary expansion of $n_a$. So, by Theorem \ref{thm_binaryExpansion} and the induction hypothesis, $\mathcal{C}(T^{\mathit{gfb}}_{n_b})  =  c_{n_b}=0$ and 
$$
\mathcal{C}(T^{\mathit{gfb}}_{n_a})  =  c_{n_a}= \sum_{j=2}^\ell 2^{m_j}(m_1 - 1 - m_j - 2(j-2))
$$
and then,  by Lemma \ref{colless_sum},
\begin{align*}
\mathcal{C}(T^{\mathit{gfb}}_{n}) & = \mathcal{C}(T^{\mathit{gfb}}_{n_a})+\mathcal{C}(T^{\mathit{gfb}}_{n_b})+n_a-n_b\\
&=
\sum_{j=2}^\ell 2^{m_j}(m_1 - 1 - m_j - 2(j-2)) + \sum_{j=2}^\ell 2^{m_j}\\
& =  \sum_{j=2}^\ell 2^{m_j}(m_1 - m_j - 2(j-2))=c_n.
\end{align*}

\item If $p= 2^{m-1}$, so that $n=2^m+2^{m-1}$ is the binary expansion of $n$, then, by Proposition \ref{GFB_Decomposition}, $n_a = 2^{m}$ and $n_b=2^{m-1}$. In this case, by the induction hypothesis, $\mathcal{C}(T_a)=c_{n_a}=0$ and $\mathcal{C}(T_b)=c_{n_b}=0$, and then,  by Lemma~\ref{colless_sum},
$$
\mathcal{C}(T^{\mathit{gfb}}_{n})  = \mathcal{C}(T^{\mathit{gfb}}_{n_a})+\mathcal{C}(T^{\mathit{gfb}}_{n_b})+n_a-n_b=2^{m-1}=c_n
$$
by Theorem \ref{thm_binaryExpansion}.

\item Finally, assume  that $p > 2^{m-1}$, so that its binary expansion is $p= 2^{m-1} + 2^{m_3}+\cdots+2^{m_\ell}$, and in particular $m_2=m-1$. In this case, $n_a = 2^m$ and $n_b = p$, so that $n_a-n_b=2^{m}-p=2^{m-1}-(2^{m_3}+\cdots+2^{m_\ell})$, and, by the induction hypothesis, $\mathcal{C}(T^{\mathit{gfb}}_{n_a})  =  c_{n_a}=0$ and
\begin{align*}
C(T^{\mathit{gfb}}_{n_b}) & =  c_{n_b}= \sum_{j=3}^\ell 2^{m_j}(m_2 - m_j - 2(j-1-2))   \\
			  	 & =  \sum_{j=3}^\ell 2^{m_j}(m-1 - m_j - 2(j-2)+2)\\
			  	 & =  \sum_{j=3}^\ell 2^{m_j}(m - m_j - 2(j-2)) +  \sum_{j=3}^\ell 2^{m_j}
\end{align*}
Then,  by Lemma \ref{colless_sum},
\begin{align*}
\mathcal{C}(T^{\mathit{gfb}}_{n}) & = \mathcal{C}(T^{\mathit{gfb}}_{n_a})+\mathcal{C}(T^{\mathit{gfb}}_{n_b})+n_a-n_b\\
&=\sum_{j=3}^\ell 2^{m_j}(m_1 - m_j - 2(j-2)) +  \sum_{j=3}^\ell 2^{m_j}+2^{m-1}-\sum_{j=3}^\ell 2^{m_j}\\
&=\sum_{j=3}^\ell 2^{m_j}(m_1 - m_j - 2(j-2)) +2^{m_2}\\
& =  \sum_{j=2}^\ell 2^{m_j}(m_1 - m_j - 2(j-2))=c_n
\end{align*}
(in the third and fourth equalities we use that $m=m_1$ and  $m_2=m-1=m_1-1$) as we wanted to show.\qed
\end{enumerate}
\end{proof}

So,  for any given number $n$ of leaves, both the maximally balanced trees and the GFB trees have minimum Colless index.
Moreover, while the balance value of the root of $T_{n}^{\mathit{mb}}$ is by definition at most 1, Proposition \ref{GFB_Decomposition} implies that if  $n=2^m+p$ with $m=\lfloor \log_2(n)\rfloor$, the balance value of the root of $T_{n}^{\mathit{gfb}}$ is $\min\{p,2^m-p\}$ and therefore $T_{n}^{\mathit{mb}}\neq T_{n}^{\mathit{gfb}}$ if $p\neq 0,1,2^{m}-1$.
On the other hand, we already know  (cf. Corollary \ref{min_colless}) that if $n=2^m$, then there is only one minimal Colless tree with $n$ leaves and therefore in this case $T_{n}^{\mathit{mb}}= T_{n}^{\mathit{gfb}}$, and it is straightforward to prove by induction on $m$, using Proposition \ref{Thm_number_Colless}  and the fact that $QB(2^m-1)=\{(2^{m-1},2^{m-1}-1)\}$ and $QB(2^m+1)=\{(2^{m-1}+1,2^{m-1})\}$, that if $n$ has the form $2^m\pm 1$,  then there is only one minimal Colless tree in $\TT_n$, too.  In summary, this proves the following result.

\begin{corollary}\label{cor:morethan1}
For every $n\geq 1$,  if $n\notin \{2^m-1,2^m,2^m+1\}$ for any $m\in \NN_{\geq 1}$, then $T_{n}^{\mathit{mb}}\neq T_{n}^{\mathit{gfb}}$, while if $n\in \{2^m-1,2^m,2^m+1\}$ for some $m\in \NN_{\geq 1}$, then
there is only one minimal Colless tree in $\TT_n$.
\end{corollary}

The next result entails that the GFB trees can also be built through a top-down strategy as follows: we start with a cluster of $n$ leaves, and build a hierarchical clustering by splitting clusters into pairs of subclusters of suitable cardinalities.

\begin{corollary}\label{cor:tdGFB}
For every $T\in \TT_n$, $T=T^{\mathit{gfb}}_n$ if, and only if, for every $v\in \mathring{V}(T)$, if  we write $\kappa_T(v)=2^k+s$ with $k=\lfloor \log_2(\kappa_T(v))\rfloor$ and $0\leq s<2^k$, then the numbers of descendant leaves of the children of $v$ are, respectively, $2^{k-1}+s$ and $2^{k-1}$, if $0\leq s\leq 2^{k-1}$, or $2^k$ and $s$,  if $2^{k-1}\leq s<2^k$. 
\end{corollary}

\begin{proof}
The ``only if'' implication is a direct consequence of Proposition \ref{GFB_Decomposition} and the fact that, as as a consequence of Lemma \ref{gfb_subtrees}, any rooted subtree of a GFB tree is again GFB. We prove now the ``if'' implication by induction on $n$. The base case when $n=1$ is obvious, because there is only one tree with 1 leaf. Assume now that this implication is true for every $1\leq n'<n$, and let $T\in \TT_n$ be such that, for every $v\in \mathring{V}(T)$, if  we write $\kappa_T(v)=2^k+s$ with $k=\lfloor \log_2(\kappa_T(v))\rfloor$ and $0\leq s<2^k$, then the numbers of descendant leaves of the children of $v$ are, respectively, $2^{k-1}+s$ and $2^{k-1}$, if $0\leq s\leq 2^{k-1}$, or $2^k$ and $s$,  if $2^{k-1}\leq s<2^k$. Consider the decomposition $T=(T_a,T_b)$ of $T$ into its two maximal pending subtrees, with $T_a\in \TT_{n_a}$ and $T_b\in \TT_{n_b}$, $n_a\geq n_b$. Then, on the one hand, the internal nodes of both $T_a$ and $T_b$ satisfy the aforementioned property on the numbers of descendant leaves of their children, which implies by the induction hypothesis that $T_a=T^{\mathit{gfb}}_{n_a}$ and $T_b=T^{\mathit{gfb}}_{n_b}$. And, on the other hand, by hypothesis $n_a$ and $n_b$ satisfy that if we write $n=2^m+p$, with $m=\lfloor \log_2(n)\rfloor$ and $0\leq p<2^m$, then $n_a=2^{m-1}+p$ and $n_b=2^{m-1}$, if $0\leq p\leq 2^{m-1}$, or $n_a=2^m$ and $n_b=p$,  if $2^{m-1}\leq p<2^m$.  But then, by Proposition \ref{GFB_Decomposition} and Lemma \ref{gfb_subtrees}, the decomposition of $T^{\mathit{gfb}}_n$ into its maximal pending subtrees is $(T^{\mathit{gfb}}_{n_a},T^{\mathit{gfb}}_{n_b})$ with $n_a$ and $n_b$ precisely given by these formulas. This implies that $T=T^{\mathit{gfb}}_n$.
\qed
\end{proof}

The maximally balanced trees and the GFB trees turn out to be extremal among the minimal Colless trees in the sense that no minimal Colless tree  can have a smaller difference between the number of leaves of its maximal pending subtrees than the maximally balanced tree or a larger difference between these numbers than the GFB tree. The assertion on the maximally balanced trees being obvious, because that difference is the least possible one (0 or 1, depending on whether the number of leaves is even or odd, respectively), we must prove the assertion on the GFB trees.

\begin{proposition} \label{gfb_extreme}
Let $T^{\mathit{gfb}}_n = (T^{\mathit{gfb}}_{n^{\mathit{gfb}}_a},T^{\mathit{gfb}}_{n^{\mathit{gfb}}_b})$ be the decomposition of a GFB tree with $n$ leaves into its maximal pending subtrees. If 
${T}=({T}_a, {T}_b)$, with ${T}_a\in \TT_{{n}_a}$ and ${T}_b\in \TT_{{n}_b}$, is another minimal Colless tree with $n$ leaves, then ${n}_a - {n}_b\leq n^{\mathit{gfb}}_a-n^{\mathit{gfb}}_b$.
\end{proposition}

\begin{proof} 
Write $n$ as $n = 2^{m} + p$ with $m=\lfloor \log_2(n)\rfloor$ and $0\leq p<2^m-1$. We know from Proposition \ref{GFB_Decomposition} that if $0\leq p\leq 2^{m-1}$, then $(n^{\mathit{gfb}}_a,n^{\mathit{gfb}}_b)=(2^{m-1}+p,2^{m-1})$ and hence $n^{\mathit{gfb}}_a-n^{\mathit{gfb}}_b=p$, and 
if $2^{m-1}\leq p< 2^m$, then $(n^{\mathit{gfb}}_a,n^{\mathit{gfb}}_b)=(2^{m},p)$ and hence $n^{\mathit{gfb}}_a-n^{\mathit{gfb}}_b=2^m-p$.
Moreover, if $p\in\{0,1,2^m-1\}$, we know from Corollary \ref{cor:morethan1} that there is only one minimal Colless tree in $\TT_n$, and therefore we can assume henceforth that $2\leq p\leq 2^m-2$. 

Now, if ${T}=({T}_a, {T}_b)\in \widetilde{\mathcal{MC}}_n$, then, by Proposition \ref{lem:charmin1},
$({n}_a, {n}_b)\in QB(n)$. Therefore, it is enough to prove that if $(n_a,n_b)\in QB(n)$, then $n_a-n_b\leq \min\{p,2^m-p\}$. We shall do it using the explicit description of $QB(n)$ given in Proposition \ref{cor:QB}. So, let 
$2^k$ be the largest power of 2 that divides $n$, which is also the largest power of 2 that divides $p$, and let
$2^{m_1}+\cdots+2^{m_\ell}$, with $m_1=m-k>\cdots>m_\ell=1$ be the binary expansion of $n_0=n/2^k$, so that
$p=2^k(2^{m_2}+\cdots+2^{m_\ell})$.

Then, using the same notations as in Proposition \ref{cor:QB}:

\begin{itemize}
\item[(a)] Since $n$ is not a power of 2, this case cannot happen.
\item[(b.1)] If $(n_a, n_b)$ has the form
$$
\Big(2^k\Big(\sum_{i=1}^{\ell-1} 2^{m_i-1}+1\Big), 2^k\sum_{i=1}^{\ell-1} 2^{m_i-1}\Big),
$$
then 
$$
n_a - n_b=2^k\leq \min\{p, 2^m-p\}
$$
because $2^k$ divides both $p$ and $2^m-p$.

\item[(b.2)] If $(n_a, n_b)$ has the form
$$
\Big(2^{k}\Big(\sum_{i=1}^{j-1} 2^{m_i-1}+2^{m_j}\Big),n-2^{k}\Big(\sum_{i=1}^{j-1} 2^{m_i-1}+2^{m_j}\Big)\Big),
$$
for some $j=2,\ldots, \ell-1$ such that $m_j>m_{j+1}+1$, then
$$
n_a - n_b   = 2\cdot 2^k\Big(\sum_{i=1}^{j-1} 2^{m_i-1}+2^{m_j}\Big) - n=2^k\Big(\sum_{i=1}^{j-1} 2^{m_i}+2^{m_j+1}\Big) - n
$$
and this is smaller or equal than $\min\{p, 2^m-p\}$
because, on the one hand,
\begin{align*}
& 2^k\Big(\sum_{i=1}^{j-1} 2^{m_i}+2^{m_j+1}\Big) - n =2^k\Big(\sum_{i=1}^{j} 2^{m_i}+2^{m_j}\Big)-n\\
& \leq  
2^k\Big(\sum_{i=1}^{\ell-1} 2^{m_i}+2^{m_j}\Big)-n <
2^k (n_0 + 2^{m_j}) - n\\
&  = 2^k \cdot 2^{m_j} 
\leq 2^k\sum_{i=2}^{\ell} 2^{m_i}=p
\end{align*}
and, on the other hand,
\begin{align*}
 &2^k\Big(\sum_{i=1}^{j-1} 2^{m_i}+2^{m_j+1}\Big) - n  \leq 2^k\Big(\sum_{i=1}^{j-1} 2^{m_i}+2^{m_{j-1}}\Big) - n\\
 &\quad  \leq 2^k\Big(\sum_{i=m_{j-1}}^{m_1} 2^{i}+2^{m_{j-1}}\Big)-n= 2^k\cdot 2^{m_1+1}-n=2^{m+1}-n=2^m-p.
\end{align*}

\item[(b.3)] If $(n_a, n_b)$ has the form
$$
\Big(n-2^{k}\sum_{i=1}^{j-1} 2^{m_i-1},2^{k}\sum_{i=1}^{j-1} 2^{m_i-1}\Big)
$$
for some $j=2,\ldots,\ell-1$ such that $m_j<m_{j-1}-1$, then
$$
n_a - n_b=n-2\cdot 2^{k}\sum_{i=1}^{j-1} 2^{m_i-1}=n-2^{k}\sum_{i=1}^{j-1} 2^{m_i}\leq\min\{p, 2^m-p\}
$$
because, on the one hand
$$
n-2^{k}\sum_{i=1}^{j-1} 2^{m_i}\leq n-2^{k}\cdot 2^{m_1}=n-2^m=p,
$$
and, on the other hand,
\begin{align*}
&  n - 2^k \sum_{i=1}^{j-1} 2^{m_i}  =     2^k \sum_{i=1}^{\ell} 2^{m_i} - 2^k \sum_{i=1}^{j-1} 2^{m_i} = 2^k \sum_{i=j}^{\ell} 2^{m_i}\\
 &\quad  \leq 2^k \sum_{i=0}^{m_j} 2^{i}=2^k(2^{m_j+1}-1)<   2^k\Big(2^{m_1}-\sum_{i=2}^\ell 2^{m_i}\Big)= 2^{m} - p,
\end{align*}
where the last inequality holds because $m_j<m_j+1<m_{j-1}$ implies that
$$
2^{m_j+1}+\sum_{i=2}^\ell 2^{m_i}\leq \sum_{s=0}^{m_2} 2^s=2^{m_2+1}-1<2^{m_1}.
$$

\item[(b.4)] If $(n_a , n_b )=(n/2,n/2)$, then $n_a - n_b = 0 < \min\{p,2^m-p\}$. \qed
\end{itemize}
\end{proof}

We now immediately have:

\begin{corollary} \label{leaf_partioning}
Let $T^{\mathit{gfb}}_n$ be the GFB tree with $n$ leaves and $n_a^{\mathit{gfb}}\geq n_b^{\mathit{gfb}}$ the numbers of leaves of its maximal pending subtrees. Then, for every $T\in \widetilde{\mathcal{MC}}_n$, if $n_a\geq n_b$ are the numbers of leaves of its maximal pending subtrees,
$$
n_b^{\mathit{gfb}} \leq n_b \leq \lfloor n/2\rfloor\leq \lceil n/2\rceil  \leq n_a \leq n_a^{\mathit{gfb}}.
$$
\end{corollary}

\begin{proof}
Assume that $n_a \leq \lceil n/2\rceil-1$. Since $n_a+n_b=n$, this would imply that $n_b\geq \lfloor n/2\rfloor+1\geq \lceil n/2\rceil$ and it would contradict the assumption that $n_a \geq n_b$. Thus, $n_a \geq\lceil n/2\rceil$. A similar argument shows that $n_b \leq \lfloor n/2\rfloor$. 

Assume now that $n_a > n_a^{\mathit{gfb}}$. Then, since $n_a+n_b=n_a^{\mathit{gfb}}+n_b^{\mathit{gfb}}$, this would imply
that $n_b<n_b^{\mathit{gfb}}$ and hence that $n_a-n_b>n_a^{\mathit{gfb}}-n_b^{\mathit{gfb}}$, which contradicts Proposition \ref{gfb_extreme}. A similar argument shows that $n_b^{\mathit{gfb}} \leq n_b$.  \qed
\end{proof}

\begin{remark}
Since any rooted subtree of a  minimal Colless tree (respectively, of a maximally balanced tree or a GFB tree) is again minimal Colless  (respectively, maximally balanced  or GFB), the last corollary applies not only to the numbers of leaves of the maximal pending subtrees of a minimal Colless tree, but also to the numbers of descendant leaves of the children of any internal node $v$ in minimal Colless trees, relative to the number of descendant leaves of $v$.
\end{remark}

{We want to point out here an interesting consequence of the last corollary: the GFB tree with $n$ leaves has the largest number of symmetry vertices, and hence also of automorphisms, among all minimal Colless trees with $n$ leaves. So, when the GFB tree with $n$ leaves is not maximally balanced, it is ``more symmetrical'' (in terms of the number of automorphisms) than the maximally balanced tree with $n$ leaves.

\begin{proposition}\label{prop:GFBmax}
For every $n\geq 1$, let $n=\sum_{i=1}^\ell 2^{m_i}$, with $\ell\geq 1$ and $m_1>\cdots> m_\ell$, be its binary expansion.
\begin{enumerate}[(a)]
\item $s(T_n^\mathit{gfb})=n-1-(m_1-m_\ell)$.
\item For every  $T\in \widetilde{\mathcal{MC}}_n$, if $T\neq T_n^\mathit{gfb}$, then $s(T)< s(T_n^\mathit{gfb})$.
\end{enumerate}
\end{proposition}

We postpone the proof of this proposition until Appendix A.3 at the end of the paper.

\begin{remark}
\label{rem:maxGFB}
Last proposition also has a consequence on the number $\widehat{c}(n)$ of minimal Colless phylogenetic trees on $n$ leaves. As we saw in Remark \ref{rem:widehat}, 
$$
\widehat{c}(n) = \sum_{T \in \widetilde{\mathcal{MC}}_n} \frac{n!}{2^{s(T)}}.
$$
So, by the last proposition (and using its same notations)  and the fact that for $n \in \{2^{m_1}-1, 2^{m_1}, 2^{m_1}+1\}$ there is only one minimal Colless tree in $\mathcal{T}_n$ (cf. Corollary \ref{cor:morethan1}), we have that:
\begin{itemize}
\item If $n=2^{m_1}$, $\widehat{c}(n)=n!/2^{n-1}$.
\item If $n=2^{m_1}\pm 1$, $\widehat{c}(n)=n!/2^{n-1-m_1}$.
\item For all other values of $n$, 
$\widehat{c}(n)>n!\cdot \widetilde{c}(n)/2^{n-1-(m_1-m_\ell)}$.
\end{itemize} 
\end{remark}}

\subsection{The  minimal Colless trees have also minimum Sackin index}
Finally, we shortly focus on another popular index of tree balance, namely the so-called  \emph{Sackin index} \citep{Sackin1972,Shao:90}. Recall that   the Sackin index of a (not necessarily bifurcating) rooted tree is defined as the sum of the depths of its leaves:
$$
\mathcal{S}(T)=\sum_{x\in V_L(T)} \delta_T(x).
$$
Equivalently \citep{Blum2005}, it is equal to the sum of the numbers of descendant leaves of the internal nodes of $T$:
$$
\mathcal{S}(T)=\sum_{v\in \mathring{V}(T)} \kappa_T(v).
$$
The bifurcating trees with $n$ leaves that achieve the maximum Sackin index are exactly the caterpillars \citep{Fischer2018,Shao:90}. As to those  achieving its minimum value, they have been recently characterized by \citet{Fischer2018} and in particular they include the fully symmetric trees (cf. Theorem 5 therein). We shall generalize this result by showing that they actually include all minimal Colless trees.
We shall use from Fischer's paper the following result (cf. Corollary 4 therein):

\begin{lemma} \label{cor_mareike}
Let $T= (T_a, T_b)$ be a  bifurcating tree with $n \in \mathbb{N}_{\geq 2}$ leaves and let  $k_n= \lceil \log_2 (n) \rceil$. Then, $T$ has minimum Sackin index if, and only if, $T_a$ and $T_b$ have minimal Sackin index and $n_a - n_b \leq \min \{n-2^{k_n-1}, 2^{k_n}-n\}$.
\end{lemma}

Based on this lemma we can prove the following statement.

\begin{proposition} \label{Colless_implies_Sackin}
{For every $n \geq 1$, if $T$ is a bifurcating tree with $n$ leaves that has minimum Colless index, $T$ has also minimum Sackin index in $\TT_n$.}
\end{proposition}

\begin{proof} 
We show the statement by induction on $n$.
For $n=1$, it is, as always, obvious because there is only one tree in $\TT_1$.
Assume now that the claim holds for every $1\leq n'<n$  and let $T=(T_a, T_b)\in \TT_n$ be a minimal Colless tree with $n$ leaves, with $T_a\in \TT_{n_a}$ and $T_b\in \TT_{n_b}$. 
Write $n=2^m+p$, with $m=\lfloor\log_2(n)\rfloor$ and $0\leq p<2^m$. If $p=0$, there is only one minimal Colless tree, which is fully symmetric and therefore it has minimum Sackin index. So, we assume that $p>0$, in which case $k_n= \lceil \log_2 (n) \rceil=m+1$. By Lemma \ref{max_subtrees}, both $T_a$ and $T_b$ are minimal Colless trees and therefore, by the induction hypothesis, they have minimum Sackin index. Thus, by Lemma \ref{cor_mareike}, to prove that $T$ has minimum Colless index it is enough to prove that 
$$
n_a - n_b \leq \min \{n-2^{k_n-1}, 2^{k_n}-n\}=\{n-2^m,2^{m+1}-n\}=\{p,2^{m}-p\}.
$$
But this has already been proved in the proof of  Proposition \ref{gfb_extreme}.\qed
\end{proof}

The converse implication is not true. For example, the tree $T_2$ depicted in Figure \ref{Fig_Sackin_not_Colless} has minimum Sackin index, but it does not have minimum Colless index. {Note that this entails that, as far as classifying trees to be most balanced is concerned, the Colless index is a finer index than Sackin's is.}

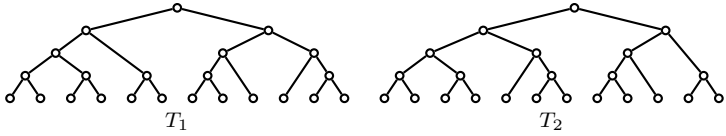
\begin{figure}[htbp]
\begin{center}
\begin{tikzpicture}[thick,>=stealth,scale=0.2]
\draw(0,0) node [treppp] (1) {}; 
\draw(2,0) node [treppp] (2) {}; 
\draw(4,0) node [treppp] (3) {}; 
\draw(6,0) node [treppp] (4) {}; 
\draw(8,0) node [treppp] (5) {};  
\draw(10,0) node [treppp] (6) {}; 
\draw(12,0) node [treppp] (7) {}; 
\draw(14,0) node [treppp] (8) {}; 
\draw(16,0) node [treppp] (9) {}; 
\draw(18,0) node [treppp] (10) {};  
\draw(20,0) node [treppp] (11) {}; 
\draw(22,0) node [treppp] (12) {}; 
\draw(1,1.5) node[treppp] (a) {};
\draw(5,1.5) node[treppp] (b) {};
\draw(9,1.5) node[treppp] (c) {};
\draw(3,3) node[treppp] (x) {};
\draw(5,4.5) node[treppp] (h) {};
\draw(13,1.5) node[treppp] (d) {};
\draw(21,1.5) node[treppp] (f) {};
\draw(14,3) node[treppp] (e) {};
\draw(20,3) node[treppp] (g) {};
\draw(17,4.5) node[treppp] (i) {};
\draw(11,6) node[treppp] (r) {};
\draw  (a)--(1);
\draw  (a)--(2);
\draw  (b)--(3);
\draw  (b)--(4);
\draw  (c)--(5);
\draw  (c)--(6);
\draw  (x)--(a);
\draw  (x)--(b);
\draw  (h)--(x);
\draw  (h)--(c);
\draw  (f)--(11);
\draw  (f)--(12);
\draw  (d)--(7);
\draw  (d)--(8);
\draw  (e)--(d);
\draw  (e)--(9);
\draw  (g)--(f);
\draw  (g)--(10);
\draw  (i)--(e);
\draw  (i)--(g);
\draw  (r)--(h);
\draw  (r)--(i);
\draw(11,-1.5) node {\footnotesize $T_1$};
\end{tikzpicture}
\quad
\begin{tikzpicture}[thick,>=stealth,scale=0.2]
\draw(0,0) node [treppp] (1) {}; 
\draw(2,0) node [treppp] (2) {}; 
\draw(4,0) node [treppp] (3) {}; 
\draw(6,0) node [treppp] (4) {}; 
\draw(8,0) node [treppp] (5) {};  
\draw(10,0) node [treppp] (6) {}; 
\draw(12,0) node [treppp] (7) {}; 
\draw(14,0) node [treppp] (8) {}; 
\draw(16,0) node [treppp] (9) {}; 
\draw(18,0) node [treppp] (10) {};  
\draw(20,0) node [treppp] (11) {}; 
\draw(22,0) node [treppp] (12) {}; 
\draw(1,1.5) node[treppp] (a) {};
\draw(5,1.5) node[treppp] (b) {};
\draw(3,3) node[treppp] (x) {};
\draw(11,1.5) node[treppp] (c) {};
\draw(10,3) node[treppp] (d) {};
\draw(6.5,4.5) node[treppp] (h) {};
\draw(15,1.5) node[treppp] (f) {};
\draw(16,3) node[treppp] (g) {};
\draw(21,1.5) node[treppp] (j) {};
\draw(18.5,4.5) node[treppp] (i) {};
\draw(12.5,6) node[treppp] (r) {};

\draw  (a)--(1);
\draw  (a)--(2);
\draw  (b)--(3);
\draw  (b)--(4);
\draw  (x)--(a);
\draw  (x)--(b);
\draw  (c)--(6);
\draw  (c)--(7);
\draw  (d)--(5);
\draw  (d)--(c);
\draw  (h)--(d);
\draw  (h)--(x);
\draw  (f)--(8);
\draw  (f)--(9);
\draw  (g)--(10);
\draw  (g)--(f);
\draw  (j)--(11);
\draw  (j)--(12);
\draw  (i)--(j);
\draw  (i)--(g);
\draw  (r)--(h);
\draw  (r)--(i);
\draw(11,-1.5) node {\footnotesize $T_2$};
\end{tikzpicture}
\end{center}
\caption{Trees $T_1$ and $T_2$ with 12 leaves. We have $\mathcal{C}(T_1)=4=c_{12}$ and $\mathcal{C}(T_2) = 6$. Thus, $T_1$ has minimum Colless index, while $T_2$ does not. Note, however, that their Sackin indices are $\mathcal{S}(T_1) = \mathcal{S}(T_2) = 44$, which can be shown to be minimal (cf. Theorem 3 in \citet{Fischer2018}).}
\label{Fig_Sackin_not_Colless}
\end{figure}

\section{Discussion}

{The balance of a phylogenetic tree is informally defined as the tendency of its internal nodes to split their sets of descendant leaves among their children nodes into clades of similar sizes. This property is independent  of the actual taxa labeling the leaves of the phylogenetic tree, and therefore it is a property of its shape, i.e., of the (unlabeled) tree underlying it. The Colless index $\mathcal{C}(T)$ of a rooted bifurcating phylogenetic tree $T$ directly quantifies the balance of $T$ by adding up 
the local imbalances of its internal nodes $v$,  measured as the absolute value of the difference of the numbers of descendant leaves of the children of $v$.  Introduced by \citet{Colless1982}, this index is not only one of the oldest balance indices for bifurcating phylogenetic trees, but probably the most widely used: 
for instance, its citations according to Google Scholar double those of the second most  widely used such index,  the Sackin index (240  \textsl{vs} 118 citations; data retrieved on December 15, 2019). But, despite its popularity, neither its minimum value for any given number of leaves nor the trees where this minimum value is achieved were known so far. This paper fills this gap in the literature, with two main contributions. 

First, we have established both a recursive and two different closed expressions for the minimum value $c_n$ of the Colless index on the space $\TT_n$ of bifurcating trees with $n$ leaves. Knowing this minimum value, as well as its maximum value, which is reached at the caterpillars and is equal to $\binom{n-1}{2}$, allows one to normalize the Colless index so that its range becomes the unit interval $[0,1]$, by means of the usual affine transformation
$$
\widetilde{C}(T)=\frac{\mathcal{C}(T)-c_n}{\binom{n-1}{2}-c_n}.
$$
This normalized index  then allows for the comparison of the balance of trees with different numbers of leaves, which cannot be done directly with the unnormalized Colless index $\mathcal{C}$, because its value tends to grow with $n$. It should be mentioned that another popular normalization, or rather \emph{standardization}, strategy for balance indices,  relative to a family of probability distributions on the spaces of phylogenetic trees with $n$ leaves, does not need the knowledge of the extremal values of the index. It consists  in subtracting the expected value of the index and dividing by its standard deviation; such a normalization of the Colless index for several probability distributions  is available for instance in the R package \texttt{apTreeshape} \citep{aptreeshape}. In this way, size effects are reduced when comparing trees with different numbers of leaves, but the resulting indices take values in intervals that still grow with $n$.

Our expressions for $c_n$ have been obtained by first proving that the \emph{maximally balanced} trees are \emph{minimal Colless  trees}, i.e. they have minimum Colless index  for their  number of leaves. This result is not surprising, because, in words of \citet{Shao:90}, they are considered  to be the  ``most balanced'' bifurcating trees. But it turns out that for almost all values of  $n$ there are minimal Colless trees  that are not maximally balanced. More precisely, if $n$ differs at least 2 from any power of 2, then there are minimal Colless trees with $n$ leaves that are not maximally balanced. So, our second main contribution has been a structural characterization of the minimal Colless trees, an efficient algorithm to produce all of them for any number $n$ of leaves, and a recurrence to compute the number $\widetilde{c}(n)$ of different minimal Colless trees with $n$ leaves for every $n$. 
Moreover, we have described a second family of minimal Colless trees, that we have called \emph{greedy from the bottom}, \emph{GFB}, with a member in every space $\TT_n$. These GFB trees are different from the maximally balanced trees for all numbers of leaves for which there exist at least two different minimal Colless trees,  and they turn out to be the most symmetrical (i.e., those with the maximum number of automorphisms) minimal Colless trees. We have not been able to characterize the minimal Colless trees with the least number of automorphisms, or even to find a formula for this least number for each number $n$ of leaves. It would be natural to conjecture that, since the maximally balanced trees and the GFB trees are extremal among all minimal Colless trees in a very specific sense (cf. Corollary \ref{leaf_partioning}) and the GFB trees have the largest number of automorphisms, the maximally balanced trees would have the least number of automorphisms, which would be $2^{n-1-c_n}$ by Corollary \ref{cor:autom}. Although this is true for many values of $n$, it is false in general. The first counterexample appears with $n=24$: see Figure \ref{fig:ceswaswrong}. The tree depicted in this figure is obtained by replacing in the maximally balanced tree $T_{24}^\mathit{mb}$ a maximally balanced rooted subtree with 6 leaves by a GFB tree with 6 leaves, reducing in this way in 1 the number of symmetry vertices in $T_{24}^\mathit{mb}$. So, we leave as an open problem to characterize the minimal Colless trees with the least number of symmetry vertices. 

\begin{figure}[htb]
\centering
\begin{tikzpicture}[thick,>=stealth,scale=0.4]
\draw(1,0) node[trepp] (l1) {}; 
\draw(2,0) node[trepp] (l2) {}; 
\draw(3,0) node[trepp] (l3) {}; 
\draw(4,0) node[trepp] (l4) {}; 
\draw(5,0) node[trepp] (l5) {}; 
\draw(6,0) node[trepp] (l6) {}; 
\draw(7,0) node[trepp] (l7) {}; 
\draw(8,0) node[trepp] (l8) {}; 
\draw(9,0) node[trepp] (l9) {}; 
\draw(10,0) node[trepp] (l10) {}; 
\draw(11,0) node[trepp] (l11) {}; 
\draw(12,0) node[trepp] (l12) {}; 
\draw(13,0) node[trepp] (l13) {}; 
\draw(14,0) node[trepp] (l14) {}; 
\draw(15,0) node[trepp] (l15) {}; 
\draw(16,0) node[trepp] (l16) {}; 
\draw(17,0) node[trepp] (l17) {}; 
\draw(18,0) node[trepp] (l18) {}; 
\draw(19,0) node[trepp] (l19) {}; 
\draw(20,0) node[trepp] (l20) {}; 
\draw(21,0) node[trepp] (l21) {}; 
\draw(22,0) node[trepp] (l22) {}; 
\draw(23,0) node[trepp] (l23) {}; 
\draw(24,0) node[trepp] (l24) {}; 
\draw(1.5,1) node[trepp] (v1) {}; 
\draw(5.5,1) node[trepp] (v2) {}; 
\draw(7.5,1) node[trepp] (v3) {}; 
\draw(11.5,1) node[trepp] (v4) {}; 
\draw(13.5,1) node[trepp] (v5) {}; 
\draw(17.5,1) node[trepp] (v6) {}; 
\draw(19.5,1) node[trepp] (v7) {}; 
\draw(21.5,1) node[trepp] (v8) {}; 
\draw(23.5,1) node[trepp] (v9) {}; 
\draw(2.5,2) node[trepp] (w1) {}; 
\draw(4.5,2) node[trepp] (w2) {}; 
\draw(8.5,2) node[trepp] (w3) {}; 
\draw(10.5,2) node[trepp] (w4) {}; 
\draw(14.5,2) node[trepp] (w5) {}; 
\draw(16.5,2) node[trepp] (w6) {}; 
\draw(20.5,2) node[trepp] (w7) {}; 
\draw(3.5,3) node[trepp] (x1) {}; 
\draw(9.5,3) node[trepp] (x2) {}; 
\draw(15.5,3) node[trepp] (x3) {}; 
\draw(21.5,3) node[trepp] (x4) {}; 
\draw(6.5,4) node[trepp] (y1) {}; 
\draw(18.5,4) node[trepp] (y2) {}; 
\draw(12.5,5) node[trepp] (r) {}; 
\draw (r)--(y1);
\draw (r)--(y2);
\draw (y1)--(x1);
\draw (y1)--(x2);
\draw (y2)--(x3);
\draw (y2)--(x4);
\draw (x1)--(w1);
\draw (x1)--(w2);
\draw (x2)--(w3);
\draw (x2)--(w4);
\draw (x3)--(w5);
\draw (x3)--(w6);
\draw (x4)--(w7);
\draw (x4)--(v9);
\draw (w1)--(v1);
\draw (w1)--(l3);
\draw (w2)--(l4);
\draw (w2)--(v2);
\draw (w3)--(v3);
\draw (w3)--(l9);
\draw (w4)--(l10);
\draw (w4)--(v4);
\draw (w5)--(v5);
\draw (w5)--(l15);
\draw (w6)--(l16);
\draw (w6)--(v6);
\draw (w7)--(v7);
\draw (w7)--(v8);
\draw (v1)--(l1);
\draw (v1)--(l2);
\draw (v2)--(l5);
\draw (v2)--(l6);
\draw (v3)--(l7);
\draw (v3)--(l8);
\draw (v4)--(l11);
\draw (v4)--(l12);
\draw (v5)--(l13);
\draw (v5)--(l14);
\draw (v6)--(l17);
\draw (v6)--(l18);
\draw (v7)--(l19);
\draw (v7)--(l20);
\draw (v8)--(l21);
\draw (v8)--(l22);
\draw (v9)--(l23);
\draw (v9)--(l24);
\end{tikzpicture}
\caption{\label{fig:ceswaswrong} A maximally balanced tree with $n=24$ leaves and $14<s(T_{24}^\mathit{mb})=24-1-c_{24}=15$ symmetry vertices.}
\end{figure}
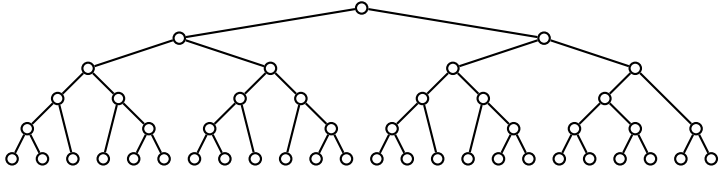

We would like to point out that one of our expressions for $c_n$ entails a fractal structure for the graph of $(n,c_n)$ related to the fractal Blancmange curve (cf. Figure \ref{Fig_MinimumColless}). It turns out that a similar fractal structure seems to appear also in the graph of $(n,\widetilde{c}(n))$ (cf. Figure \ref{Fig_NumberCollessMinima}). Unfortunately, we have not been able so far to find an explicit formula for  $\widetilde{c}(n)$, and it would definitely be of interest to find such a formula and to analyze whether this seemingly fractal structure is real or not and its possible relationship with that of the sequence $c_n$. 

We have concluded by showing that every Colless minimal tree also has minimum Sackin index, while the converse is not true. This implies that the Sackin index classifies more trees as \enquote{most balanced} than the Colless index. The Colless index, on the other hand, considers more trees as \enquote{most balanced} than for example the  \emph{total cophenetic index}, for which the minimum value on each $\TT_n$ is uniquely achieved by the maximally balanced tree (\citet{Mir2013}). These differences in performance could be due to the differences in the ranges of values of these indices. The total cophenetic index has a range of values with lower limit in $O(n^2)$ and upper limit $\binom{n}{3}$, and so its width  grows in $O(n^3)$. As for the other two indices, the range of values of the Colless index has a lower limit below $\min\{n/2,2^{\lceil\log_2(n)\rceil}/3\}$ (see Corollary \ref{min_colless_properties})  and upper limit $\binom{n-1}{2}$, while the range of values of the Sackin index goes from at least $n\lfloor\log_2(n)\rfloor$ to $\binom{n+1}{2}-1$: thus, although both widths grow in $O(n^2)$, the Sackin index has a narrower range.  

One might  also wonder about the distribution of minimal Colless trees among reconstructed phylogenetic trees. To answer this question, we report the results of an experiment in which we have looked for minimal Colless trees in the TreeBase database \citep{Piel,Vos}. To do that, we have downloaded all rooted bifurcating species trees in it with more than 3 leaves (data retrieved in December 12, 2019) and after removing those that had format
problems that prevented  parsing them, we obtained 5,617 trees with at least 4 leaves. Of them, only 24 were  minimal Colless trees. Among these  minimal Colless trees, 15 have numbers of leaves for which there exists only one minimal Colless tree shape. The other 9  minimal Colless trees can be  classified into a set of 4 trees with $n\leq 11$ leaves, all of them  GFB, and a set of 5 trees with $n\geq 12$ leaves, all of them maximally balanced. These tree shapes are available at the GitHub repository associated to this paper (\url{https://github.com/biocom-uib/Colless}).

For instance, TreeBase contains 52 rooted bifurcating trees with 4 leaves, of which 6 are fully symmetric and the remaining 46 are caterpillars (these are the only two possibilities in $\TT_4$). The 95\% Clopper-Pearson confidence interval for the probability of a phylogenetic tree with 4 leaves to be fully symmetric computed with these data goes from 0.0435 to 0.2344. Now, the probability of a phylogenetic tree with 4 leaves to be fully symmetric under Aldous' $\beta$-model for bifurcating phylogenetic trees \citep{Ald1} is
$$
P_\beta(T_4^\mathit{mb})=\frac{3\beta+6}{7\beta+18}
$$
and under Ford's $\alpha$-model \citep{Ford} this probability is
$$
P_\alpha(T_4^\mathit{mb})=\frac{1-\alpha}{3-\alpha}
$$
(for detailed computations of these probabilities, see  Lemmas 4 and 5 in \citep{CMR}, respectively).  This produces the 95\% confidence interval $(-1.94, -1.31)$ for the parameter $\beta$ in Aldous' model and the 95\% confidence interval $(0.39,0.91)$ for the parameter $\alpha$ in Ford's model. Since the Yule-Harding model \citep[p. 43]{Steel} corresponds to $\beta=0$ and  $\alpha=0$ and the uniform model \citep[p. 50]{Steel} corresponds to $\beta=-1.5$ and $\alpha=0.5$, we conclude that the data currently contained in TreeBase are inconsistent with the Yule model on $\TT_4$ and consistent with the uniform model to the 95\% level of confidence.

As another example, consider the case $n=6$. TreeBase contains 43 phylogenetic trees with 6 leaves: 3 of them are GFB trees and none of them is maximally balanced. Now, under Ford's $\alpha$-model, a  phylogenetic tree with 6 leaves is maximally balanced or GFB, respectively, with probabilities 
$$
P_\alpha(T_6^\mathit{mb})=\frac{(1-\alpha)^2(8-\alpha)}{(3-\alpha)(4-\alpha)(5-\alpha)},\quad 
P_\alpha(T_6^\mathit{gfb})=\frac{2(1-\alpha)(2-\alpha)}{(4-\alpha)(5-\alpha)}
$$
(see Figs.~28 and 29 in \citep{Ford}). Under Aldous' $\beta$-model, these probabilities are
\begin{align*}
P_\beta(T_6^\mathit{mb}) & =
\frac{10 (\beta+3)(\beta+2)}{31 \beta^2 + 194\beta + 300}
\\
P_\beta(T_6^\mathit{gfb}) &  = \frac{45 (\beta+4)(\beta+2)^2}{(31 \beta^2 + 194\beta + 300)(7\beta+18)}
\end{align*}
(see Appendix A.4 for the detailed computation of these probabilities).
From these formulas  it is easy to check that
\begin{align*}
& P_\alpha(T_6^\mathit{gfb})<P_\alpha(T_6^\mathit{mb}) \qquad \mbox{for every $\alpha\in [0,1]$}  \\
& P_\beta(T_6^\mathit{gfb})<P_\beta(T_6^\mathit{mb})  \qquad \mbox{for every $\beta\in (-2,\infty)$} 
\end{align*}
So, although for every possible value of the parameters $\alpha$ or $\beta$ the maximally balanced tree is more probable than the GFB tree in $\TT_6$, it seems that the phylogenetic reconstruction methods reverse this preference (although the difference is not statistically significant: $p=0.125$  for the bilateral binomial exact test).
}

\begin{acknowledgement}
Tom\'as M. Coronado and Francesc Rossell\'o thank the Spanish Ministry of Economy and Competitiveness and the European Regional Development Fund for partial support for this research through projects DPI2015-67082-P and PGC2018-096956-B-C43 (MINECO/FEDER). Moreover, Mareike Fischer thanks the joint research project \textit{DIG-IT!} supported by the European Social Fund (ESF), reference: ESF/14-BM-A55-0017/19, and the Ministry of Education, Science and Culture of Mecklenburg-Vorpommern, Germany.
Additionally, Lina Herbst thanks the state Mecklenburg-Western Pomerania for a Landesgraduierten-Studentship and Kristina Wicke thanks the German Academic Scholarship Foundation for a studentship.
Moreover, we thank the anonymous reviewers and the editors for their helpful comments on an earlier version of this manuscript.
\end{acknowledgement}

\bibliographystyle{spbasic}    

\begin{thebibliography}{32}
\providecommand{\natexlab}[1]{#1}
\providecommand{\url}[1]{{#1}}
\providecommand{\urlprefix}{URL }
\expandafter\ifx\csname urlstyle\endcsname\relax
  \providecommand{\doi}[1]{DOI~\discretionary{}{}{}#1}\else
  \providecommand{\doi}{DOI~\discretionary{}{}{}\begingroup
  \urlstyle{rm}\Url}\fi
\providecommand{\eprint}[2][]{\url{#2}}

\bibitem[{Agapow and Purvis(2002)}]{Agapow02}
Agapow P, Purvis A (2002) Power of eight tree shape statistics to detect nonrandom diversification: A comparison by simulation of two models of cladogenesis. Systematic Biology 51:866--872.

\bibitem[{Aldous(1996)}]{Ald1}
Aldous D (1996) Probability distributions on cladograms. In: Aldous D, Pemantle R (eds) Random Discrete
Structures. The IMA Volumes in Mathematics and its Applications, vol 76. Springer, New York, pp 1--18.

\bibitem[{Aldous(2001)}]{Aldous01}
Aldous D (2001)  Stochastic models and descriptive statistics for phylogenetic trees, from Yule to today. Statistical Science 16: 23--34.

  
\bibitem[{Allaart and Kawamura(2012)}]{Allaart2012}
Allaart PC, Kawamura K (2012) The Takagi function: a survey. Real
  Analysis Exchange 37:1--54.

\bibitem[{Avino et~al(2018)Avino, Garway, et al}]{Avino18}
Avino M, Garway TN, et al  (2018) Tree
  shape-based approaches for the comparative study of cophylogeny.
  bioRxiv \doi{10.1101/388116}.
  

  

\bibitem[{Blum and Fran{\c{c}}ois(2005)}]{Blum2005}
Blum MG, Fran{\c{c}}ois O (2005) On statistical tests of phylogenetic tree
  imbalance: The {S}ackin and other indices revisited. Mathematical Biosciences
  195:141--153.
  
\bibitem[{Blum and Fran{\c{c}}ois(2006)}]{Blum2006}
Blum MG, Fran{\c{c}}ois O (2006) Which random processes describe the tree
of life? A large-scale study of phylogenetic tree imbalance. Systematic Biology
55:685--691.


\bibitem[{Blum et~al(2006)Blum, François, and Janson}]{Blum2006a}
Blum MGB, François O, Janson S (2006) The mean, variance and limiting
  distribution of two statistics sensitive to phylogenetic tree balance. Annals of
  Applied Probability 16:2195--2214.
  
\bibitem[{Bortolussi et al(2005)}]{aptreeshape}
Bortolussi N, Durand E, Blum M,  Fran\c cois O (2005) apTreeshape: statistical analysis of phylogenetic tree shape. Bioinformatics, 22:363--364.
  
\bibitem[{Brower and Rindal(2013)}]{BrowerRindal13}
  Brower AVZ, Rindal E (2013) Reality check: A reply to Smith. Cladistics 29:464--465.



\bibitem[{Cardona et~al(2013)Cardona, Mir, and Rossell{\'o}}]{CMR2013}
Cardona G, Mir A, Rossell{\'o} F (2013) Exact formulas for the variance of
  several balance indices under the Yule model. Journal of Mathematical Biology
  67:1833--1846.

\bibitem[{Chalmandrier et~al(2018)Chalmandrier, Albouy, Descombes, Sandel,
  Faurby, Svenning, Zimmermann, and Pellissier}]{chalmandrier2018}
Chalmandrier L, Albouy C, et al (2018) Comparing spatial diversification and
  meta-population models in the Indo-Australian Archipelago. Royal Society Open
  Science 5:171366.
  
 
\bibitem[{Colless(1982)}]{Colless1982}
Colless D (1982) Review of \enquote{Phylogenetics: the theory and practice of
  phylogenetic systematics}. Systematic Zoology 31:100--104.
  
 \bibitem[{Colless(1995)}]{Colless1995}
Colless D (1995) Relative symmetry of cladograms and phenograms: An experimental study. Systematic Biology, 44:102--108.


\bibitem[{{Coronado} et~al(2019){Coronado}, {Mir}, {Rossell{\'o}}, and
  {Valiente}}]{CMR}
{Coronado} TM, {Mir} A, {Rossell{\'o}} F, {Valiente} G (2019) {A balance index
  for phylogenetic trees based on rooted quartets}. Journal of Mathematical 
  Biology 79:1105--1148.
  
 \bibitem[{Cunha and Giribet(2019)}]{Cunha2019}
Cunha T,   Giribet G (2019) A congruent topology for deep gastropod relationships. Proceedings of the Royal Society B, 286:20182776.

  
\bibitem[{Drummond et~al(2006)Drummond, Ho, Phillips, and Rambaut}]{Drummond}  
Drummond AJ, Ho SYW, Phillips MJ, Rambaut A (2006) Relaxed Phylogenetics and Dating with Confidence. PLoS Biology 4:e88. 

\bibitem[{Duchene et~al(2018)Duchene, Bouckaert, Duchene, Stadler, and
  Drummond}]{duchene2018}
Duchene S, Bouckaert R, Duchene DA, Stadler T, Drummond AJ (2018) Phylodynamic
  model adequacy using posterior predictive simulations. Systematic Biology
  68:358--364.
  
 \bibitem[{Farris and K\"allersj\"o(1998)}]{Farris98}
Farris J,  K\"allersj\"o M (1998)
 Asymmetry and explanations. Cladistics, 14:159--166.


\bibitem[{Felsenstein(2004)}]{fel:04}
Felsenstein J (2004) Inferring Phylogenies. Oxford University Press.

\bibitem[{Fischer(2018)}]{Fischer2018}
Fischer M (2018) Extremal values of the {S}ackin balance index for rooted
  binary trees. arXiv preprint  arXiv:1801.10418v3.

\bibitem[{{Fischer} and {Liebscher}(2015)}]{Fischer2015}
{Fischer} M, {Liebscher} V (2015) {On the Balance of Unrooted Trees}. arXiv preprint  arXiv:1510.07882.

\bibitem[{Ford(2005)}]{Ford}
Ford DJ (2005) Probabilities on cladograms: introduction to the alpha model.
  PhD thesis, Stanford University. arXiv preprint arXiv:math/0511246.
  
\bibitem[{Fusco and Cronk(1995)}]{Fusco95}
Fusco G,  Cronk QC (1995) A new method for evaluating the shape of large phylogenies. Journal of Theoretical Biology, 175:235--243.

  
\bibitem[{Futuyma(1999)}]{Futuyma} Futuyma DJ ed. (1999) Evolution, Science and Society: Evolutionary biology and the National Research Agenda. The State University of New Jersey.


\bibitem[{Goloboff et~al(2017)Goloboff, Arias, and Szumik}]{Goloboff17}
Goloboff PA, Arias JS, Szumik CA (2017) Comparing tree shapes: beyond symmetry.
  Zoologica Scripta 46:637--648.
  
\bibitem[{{Guyer} and {Slowinski}(1993)}]{Guyer93}  Guyer C, Slowinski J (1993) Adaptive radiation and the topology of large phylogenies. Evolution 47:253--263.

\bibitem[{Hayati, Shadgar and  Chindelevitch(2019)}]{Hayati19}
Hayati M, Shadgar B,  Chindelevitch L (2019). A new resolution function to evaluate tree shape statistics. PloS One 14:e0224197.

\bibitem[{Heard(1992)}]{Heard1992}
Heard SB (1992) Patterns in tree balance among cladistic, phenetic, and
  randomly generated phylogenetic trees. Evolution 46:1818--1826.
 
\bibitem[{Hillis et al(1992)}]{Hillis92}
Hillis D, Bull J, White M et al (1992). Experimental phylogenetics: Generation of a known phylogeny. Science, 255:589--592.

 
 \bibitem[{Holton et al(2014)Holton, Wilkinson, and Pisani}]{Holton2014}
Holton T, Wilkinson M, Pisani D (2014) The shape of modern tree reconstruction methods. Systematic biology, 63:436--441.

\bibitem[{Kayondo et al(2019)}]{Kayondo}
Kayondo H, Mwalili S, Mango J (2019). Inferring Multi-Type Birth-Death Parameters for a Structured Host Population with Application to HIV Epidemic in Africa. Computational Molecular Bioscience, 9:108--131.

\bibitem[{Kingman(1982)}]{kingman82}
Kingman JFC (1982) The coalescent. Stochastic processes and their applications
  13:235--248.

\bibitem[{Kirkpatrick and Slatkin(1993)}]{KiSl:93}
Kirkpatrick M, Slatkin M (1993) Searching for evolutionary patterns in the
  shape of a phylogenetic tree. Evolution 47:1171--1181.
  
\bibitem[{Kubo and Iwasa(1995)}]{Kubo95}
Kubo T, Iwasa Y (1995) Inferring the rates of branching and extinction from molecular phylogenies. Evolution  49:694-704

\bibitem[{Matsen(2006)}]{Matsen06}
Matsen F (2006) A geometric approach to tree shape statistics.
Systematic Biology 55:652--61.
  
\bibitem[{McKenzie and Steel(2000)}]{cherries}
McKenzie A, Steel M (2000) Distributions of cherries for two models of trees.
  Mathematical Biosciences 164:81--92.

\bibitem[{Metzig et~al(2019)Metzig, Ratmann, Bezemer, and Colijn}]{metzig2019}
Metzig C, Ratmann O, Bezemer D, Colijn C (2019) Phylogenies from dynamic
  networks. PLoS Computational Biology 15:e1006761.

\bibitem[{Mir et~al(2013)Mir, Roselló, and Rotger}]{Mir2013}
Mir A, Roselló F, Rotger L (2013) A new balance index for phylogenetic trees.
  Mathematical Biosciences 241:125--136.
  
\bibitem[{Mir et~al(2018)Mir, Rotger, and Rossell{\'{o}}}]{Mir2018}
Mir A, Rotger L, Rossell{\'{o}} F (2018) Sound {C}olless-like balance indices
  for multifurcating trees. PLoS {ONE} 13:e0203401.
  
\bibitem[{Mooers and Heard(1997)}]{Mooers1997}
Mooers AO, Heard SB (1997) Inferring evolutionary process from
  phylogenetic tree shape. The Quarterly Review of Biology 72:31--54.
  
\bibitem[{Nelson and Holmes(2007)}]{Nelson}
Nelson MI, Holmes EC (2007) The evolution of epidemic influenza. Nature Reviews
  Genetics 8:196--205.
  
\bibitem[{Piel et al(2009)}]{Piel}
Piel WH, Chan L, Dominus MJ et al (2009). TreeBASE v.2: A Database of Phylogenetic Knowledge. In: e-BioSphere 2009.

\bibitem[{Vos et al(2012)}]{Vos}
Vos RA, Balhoff JP, Caravas JA et al (2012). NeXML: Rich, extensible, and verifiable representation of comparative data and metadata. Systematic Biology 61:675--689.

  
 \bibitem[{Poon(2015)}]{Poon2015}
Poon AF (2015) Phylodynamic inference with kernel ABC and its application to HIV epidemiology. Molecular Biology and Evolution, 32:2483--2495.

\bibitem[{Purvis(1996)}]{Purvis1996}
Purvis A (1996) Using interspecies phylogenies to test macroevolutionary hypotheses. In: New Uses for New Phylogenies, Oxford University Press, 153--168.


  
\bibitem[{Purvis et~al(2011)}]{Purvis11}
Purvis A, Fritz S, Rodr\'\i guez J, Harvey P, Grenyer R (2011) The shape
of mammalian phylogeny: Patterns, processes and scales. Philosophical Transactions of The Royal Society B  366:2462--2477.

\bibitem[{Purvis et al(2002)}]{PurvisKA02}
Purvis A, Katzourakis A,  Agapow P-M (2002) Evaluating phylogenetic tree shape: Two modifications to Fusco \& Cronk's method. Journal of Theoretical Biology, 214:99--103.

\bibitem[{Rindal and Brower(2011)}]{RindalBrower11}
Rindal E, Brower AVZ (2011) Do model‐based phylogenetic analyses perform better than parsimony? A test with empirical data. Cladistics 27:331--334.


\bibitem[{Rogers(1993)}]{Rogers1993}
Rogers JS (1993) Response of {C}olless's tree imbalance to number of
  terminal taxa. Systematic Biology 42:102.

\bibitem[{Sackin(1972)}]{Sackin1972}
Sackin MJ (1972) \enquote{Good} and \enquote{bad} phenograms. Systematic
  Zoology 21:225--226.
  
\bibitem[{Saulnier, Alizon, and Gascuel(2016)}]{Saulnier16}
Saulnier E, Alizon S, Gascuel O (2016) Assessing the accuracy of Approximate Bayesian Computation approaches to infer epidemiological parameters from phylogenies. bioRxiv, 050211 \url{https://doi.org/10.1101/050211}.
  
\bibitem[{Savage(1983)}]{Savage}
Savage HM (1983) The shape of evolution: Systematic tree topology. Biological Journal of the Linnean Society, 20:225--244.

\bibitem[{Shao and Sokal(1990)}]{Shao:90}
Shao K, Sokal R (1990) Tree balance. Systematic Zoology 39:266--276.

\bibitem[{Semple and Steel(2003)}]{Semple2003}
Semple C, Steel M (2003) Phylogenetics. Oxford University Press, Oxford, 2003.

\bibitem[{Sloane(1964)}]{OEIS}
Sloane NJA (1964) The {O}n-{L}ine {E}ncyclopedia of {I}nteger {S}equences (OEIS).
\url{http://oeis.org}. Last accessed, July 8, 2019.

\bibitem[{Slowinski(1990)}]{Slowinski90} Slowinski J (1990) Probabilities of $n$-trees under two models: A
demonstration that asymmetrical interior nodes are not improbable. Systematic Zoology 39:89–94.

\bibitem[{Sober(1993)}]{Sober93}
Sober E (1993) Experimental tests of phylogenetic inference methods. Systematic biology, 42:85--89.




\bibitem[{Stam(2002)}]{Stam02}
Stam E (2002) Does imbalance in phylogenies reflect only bias? Evolution
  56:1292--1295.
  
\bibitem[{Steel(2016)}]{Steel}
Steel M (2016). Phylogeny: Discrete and random processes in evolution. SIAM.

\bibitem[{Stich and Manrubia(2009)}]{Stich09}
Stich M, Manrubia SC (2009) Topological properties of phylogenetic trees in
  evolutionary models. The European Physical Journal B 70:583--592.
  
\bibitem[{Takagi(1901)}]{Takagi1901}
Takagi T (1901) A simple example of  continuous function without derivative.
  Tokyo Sugaku-Butsurigakkwai Hokoku 1:F176--F177.

 \bibitem[{Verboom et al(2019)}]{Verboom2019}
Verboom G, Boucher F, Ackerly D \textsl{et al} (2019) Species Selection Regime and Phylogenetic Tree Shape. Systematic Biology, in press \url{https://doi.org/10.1093/sysbio/syz076}

\bibitem[{Willis and Yule(1922)}]{Yule}
Willis JC, Yule GU (1922) Some statistics of evolution and geographical
  distribution in plants and animals, and their significance. Nature
  109:177--179.

\bibitem[{Wu and Choi(2015)}]{Wu15}  
  Wu T, Choi K (2015) On joint subtree distributions under two
evolutionary models. Theoretical Population Biology 108:13--23.

\end{thebibliography}

\section*{Appendices}

\subsection*{A.1 Proof of Proposition \ref{prop:eqC}}

{Recall that, for every $n \in \mathbb{N}_{\geq 2}$,
$$
QB(n)\coloneqq\big\{ (n_a,n_b)\in \NN^2\mid  n_a\geq n_b\geq 1,\ n_a+n_b=n,
 c_{n_a}+c_{n_b}+n_a-n_b=c_n \big\}.
$$}

{We shall} establish the following result.

\setcounter{proposition}{1}
\begin{proposition}
For every $n\geq 2$ and for every $n_a,n_b\in \NN_{\geq 1}$ such that $n_a\geq n_b$ and $n_a+n_b=n$:
\begin{enumerate}[(1)]

\item If $n_a=n_b=n/2$, then $(n_a,n_b)\in QB(n)$ always.

\item If $n_a>n_b$, then $(n_a,n_b)\in QB(n)$ if, and only if, one of the following three conditions is satisfied:
\begin{itemize}
\item There exist $k\in \NN$ and $p\in \NN_{\geq 1}$ such that $n=2^k(2p+1)$, $n_a=2^k(p+1)$ and $n_b=2^kp$.

\item There exist $k\in \NN$, $l\in \NN_{\geq 2}$, $p\in \NN_{\geq 1}$, and $t\in \NN$,  $0\leq t<2^{l-2}$, such that $n=2^k(2^l(2p+1)+2t+1)$, $n_a=2^{k+l}(p+1)$, and $n_b=2^{k}(2^lp+2t+1)$. 

\item There exist $k\in \NN$, $l\in \NN_{\geq 2}$, $p\in \NN_{\geq 1}$, and $t\in \NN$,  $0\leq t<2^{l-2}$, such that $n=2^k(2^l(2p+1)-(2t+1))$, $n_a=2^k(2^l(p+1)-(2t+1))$, and $n_b=2^{k+l}p$.
\end{itemize}
\end{enumerate}
\end{proposition}

The proof of this proposition relies on several auxiliary lemmas. In order to simplify the language in their statements and proofs, throughout this section we systematically assume, without any further notice, that the symbols $j$, $k$, $m$, $n$, $p$, $s$, $t$, and $x$, possibly with subscripts or superscripts, always represent natural numbers.

\begin{lemma}\label{lem:eqeven}
Let $s=2^ks_0$ with $k\geq 1$ and $s_0\geq 1$. Then, for every $m\geq 1$, $(m+s,m)\in QB(2m+s)$  if, and only if,  $m=2^km_0$, for some $m_0\geq 1$ such that $(m_0+s_0,m_0)\in QB(2m_0+s_0)$.
\end{lemma}

\begin{proof}
We prove the equivalence in the statement by induction on the exponent $k\geq 1$.
Recall that,  by Remark \ref{rem:postkey}.(b), if $s\geq 1$ is even and $c_{m+s}+c_m+s=c_{2m+s}$, then $m$ must be even, too.  Therefore, if $s=2t_0$, then $m=2m_1$ for some $m_1\geq 1$, and then, since
$$c_{2m_1+2t_0}+c_{2m_1}+2t_0=2(c_{m_1+t_0}+c_{m_1}+t_0)$$ and 
$c_{4m_1+2t_0}=2c_{2m_1+t_0}$,
the equality $c_{m+s}+c_m+s=c_{2m+s}$ is equivalent to the equality $c_{m_1+t_0}+c_{m_1}+t_0=c_{2m_1+t_0}$. 
This proves the equivalence in the statement when $k=1$. 

Now, assume that this equivalence is true  for the exponent $k-1$, and let $s=2^ks_0$. Then, by the case $k=1$, $c_{m+s}+c_m+s=c_{2m+s}$ if, and only if, 
$m=2m_1$ for some $m_1\geq 1$ such that $$c_{m_1+2^{k-1}s_0}+c_{m_1}+2^{k-1}s_0=c_{2m_1+2^{k-1}s_0},$$ and, by the induction hypothesis, this last equality holds if, and only if, 
$m_1=2^{k-1}m_0$ for some $m_0\geq 1$ such that $c_{m_0+s_0}+c_{m_0}+s_0=c_{2m_0+s_0}$. Combining both equivalences we obtain the equivalence in the statement, thus proving the inductive step. \qed
\end{proof}


\begin{lemma}\label{lem:eqodd2}
Let  $s=2^{j+1}-(2t+1)$ be an odd integer, with $j=\floor*{\log_2(s)}$ and  $0\leq t<2^{j-1}$. Then, for every $m\geq 1$, 
$(2m+s,2m)\in QB(4m+s)$ if, and only if, $m=2^{j}p$ for some $p\geq 1$.
\end{lemma}

\begin{proof}
We prove the equivalence in the statement by induction on $s$. When $s=1=2^{1}-1$, so that $j=t=0$, the equivalence says that 
$$
c_{2m+1}+c_{2m}+1=c_{4m+1}
$$  
for every $m\geq 1$, which is true by Corollary \ref{colless_minimum}.

Assume now that the equivalence is true for every odd natural number $s'<s$ and for every $m$, and let us prove it for $s=2^{j+1}-(2t+1)$ with $0\leq t<2^{j-1}$. We have that
$$
\begin{array}{l}
c_{2m+2^{j+1}-2t-1}+c_{2m}+2^{j+1}-2t-1\\
\qquad =\big(c_{m+2^j-t}+c_m+2^j-t\big)+\big(c_{m+2^j-t-1}+c_m+2^j-t-1\big)+1\\
c_{4m+2^{j+1}-2t-1}=c_{2m+2^j-t}+c_{2m+2^j-t-1}+1
 \end{array}
 $$
 and since, by Eqn.~\eqref{cn_leq},
 $c_{m+2^j-t}+c_m+2^j-t\geq c_{2m+2^j-t}$ and $c_{m+2^j-t-1}+c_m+2^j-t-1\geq c_{2m+2^j-t-1}$,
 we have that $c_{2m+s}+c_{2m}+s=c_{4m+s}$ if, and only if, the following two identities are satisfied:
\begin{align}
& c_{m+2^j-t}+c_m+2^j-t= c_{2m+2^j-t} \label{eq:eqodd2-1}\\
& c_{m+2^j-t-1}+c_m+2^j-t-1= c_{2m+2^j-t-1}  \label{eq:eqodd2-2}
\end{align}
So, we must prove that Eqns.~(\ref{eq:eqodd2-1}) and (\ref{eq:eqodd2-2}) hold if, and only if, $m=2^{j}p$ for some $p\geq 1$.
We  distinguish two subcases, depending on the parity of $t$:
\begin{itemize}

\item If $t=2x$ for some $0\leq x<2^{j-2}$, then Eqn.~(\ref{eq:eqodd2-1}) and Lemma  \ref{lem:eqeven} imply that $m$ is even, say $m=2m_0$, and then  (\ref{eq:eqodd2-2}) says
\begin{equation}
c_{2m_0+2^j-2x-1}+c_{2m_0}+2^j-2x-1= c_{4m_0+2^j-2x-1},
\label{eq:eqodd2-3}
\end{equation}
which, by induction,  is equivalent to $m_0=2^{j-1}p$ for some $p\geq 1$, i.e. to $m=2^{j}p$ for some $p\geq 1$.
So, to complete the proof of the desired equivalence, it remains to prove that if $m=2^{j}p$, then
Eqn.~(\ref{eq:eqodd2-1}) holds. If $t=0$, this equality says
$$
c_{2^{j}p+2^j}+c_{2^{j}p}+2^j=c_{2^{j+1}p+2^j}
$$
and it is a direct consequence  of Lemma   \ref{lem:eqeven} and   Corollary \ref{colless_minimum}. So, assume that $t>0$ and write it as $t=2^i(2x_0+1)$ with $1\leq i<j-1$ and $x_0<2^{j-i-2}$. Then
$$
\begin{array}{l}
c_{m+2^j-t}+c_m+2^j-t  \\
\qquad = c_{2^{j}p+2^j-2^i(2x_0+1)} +c_{2^{j}p}+2^j-2^i(2x_0+1)\\
\qquad = 2^i\big(c_{2^{j-i}p+2^{j-i}-2x_0-1}+c_{2^{j-i}p}+2^{j-i}-2x_0-1\big)\\
\qquad = 2^ic_{2^{j-i+1}p+2^{j-i}-2x_0-1} \mbox{ (by the induction hypothesis)}\\
\qquad = c_{2^{j+1}p+2^{j}-2^i(2x_0+1)}=c_{2m+2^j-t}.
\end{array}
$$

\item If $t=2x+1$ for some $0\leq x< 2^{j-2}$, then Eqn.~\eqref{eq:eqodd2-2} and Lemma  \ref{lem:eqeven} imply that $m$ is even, say $m=2m_0$, and then it is Eqn.~\eqref{eq:eqodd2-1} which  becomes Eqn.~\eqref{eq:eqodd2-3} above, which, in turn, by induction
is equivalent to $m_0=2^{j-1}p$ for some $p\geq 1$, that is, to $m=2^{j}p$ for some $p\geq 1$. 
Thus, to complete the proof of the desired equivalence, it remains to prove that if $m=2^{j}p$, then
(\ref{eq:eqodd2-2}) holds. Now:
$$
\begin{array}{l}
c_{m+2^j-t-1}+c_m+2^j-t-1\\ \qquad  = c_{2^{j}p+2^j-2x-2}+c_{2^{j}p}+2^j-2x-2\\
\qquad =2\big(c_{2^{j-1}p+2^{j-1}-x-1}+c_{2^{j-1}p}+2^{j-1}-x-1\big)
\end{array}
$$
If $x$ is even, say $x=2x_0$, then, since $x_0<2^{j-3}$, the induction hypothesis implies that
$$
\begin{array}{l}
2\big(c_{2^{j-1}p+2^{j-1}-x-1}+c_{2^{j-1}p}+2^{j-1}-x-1\big)\\ 
\qquad =2 c_{2^{j}p+2^{j-1}-x-1} 
 = c_{2^{j+1}p+2^{j}-2x-2}= c_{2m+2^j-t-1}.
\end{array}
$$
And if $x$ is odd, write it as $x=2^i(2t_0+1)-1$ for some $1\leq i<j-1$ (and notice that $x<2^{j-2}$ implies $t_0<2^{j-i-3}$) and then
$$
\begin{array}{l}
2\big(c_{2^{j-1}p+2^{j-1}-x-1}+c_{2^{j-1}p}+2^{j-1}-x-1\big)\\
\qquad=
2\big(c_{2^{j-1}p+2^{j-1}-2^i(2t_0+1)}+c_{2^{j-1}p}+2^{j-1}-2^i(2t_0+1)\big)\\
\qquad =2\cdot 2^i\big(c_{2^{j-i-1}p+2^{j-i-1}-(2t_0+1)}+c_{2^{j-i-1}p}+2^{j-i-1}-(2t_0+1)\big)\\
\qquad = 2^{i+1} c_{2^{j-i}p+2^{j-i-1}-(2t_0+1)} \mbox{ (by the induction hypothesis)}\\
\qquad = c_{2^{j+1}p+2^{j}-2^{i+1}(2t_0+1)}=c_{2^{j+1}p+2^{j}-2x-2}\\
\qquad =c_{2m+2^j-t-1}
\end{array}
$$
This completes the proof of the desired equivalence when $t$ is odd.
\end{itemize}
So, the inductive step is true in all cases. \qed
\end{proof}

\begin{lemma}\label{lem:eqodd3}
Let  $s=2^{j+1}-(2t+1)$ be an odd integer, with $j=\floor*{\log_2(s)}$ and  $0\leq t<2^{j-1}$.  Then, for every $m\geq 0$, $(2m+1+s,2m+1)\in QB(4m+2+s)$ if, and only if, either $m=2^{j}p+t$ for some $p\geq 1$ or $s=1$ (i.e. $j=t=0$) and $m=0$.
\end{lemma}

\begin{proof}
We  also prove the equivalence in this statement by induction on $s$. When $s=1=2^{1}-1$, the equivalence says that
$c_{2m+2}+c_{2m+1}+1=c_{4m+3}$ for every $m\geq 0$, which is true by Corollary \ref{colless_minimum}.

Assume now that the equivalence is true for every odd natural number $1\leq s'<s$ and for every $m\geq 0$, and let us prove it for $s=2^{j+1}-(2t+1)\geq 3$ with $0\leq t<2^{j-1}$. In this case,  $m$ cannot be 0, because,  by Remark \ref{rem:postkey}.(a), $(s+1,1)\in QB(s+2)$ if, and only if, $s=1$. So, we can consider only the case $m\geq 1$. Then, we have that
$$
\begin{array}{l}
c_{2m+1+2^{j+1}-2t-1}+c_{2m+1}+2^{j+1}-2t-1\\
\qquad =\big(c_{m+2^j-t}+c_m+2^j-t\big)+\big(c_{m+2^j-t}+c_{m+1}+2^j-t-1\big)+1\\
c_{4m+2+2^{j+1}-2t-1}=c_{2m+2^j-t}+c_{2m+2^j-t+1}+1
 \end{array}
 $$
 and since, by Eqn.~\eqref{cn_leq},
 $c_{m+2^j-t}+c_m+2^j-t\geq c_{2m+2^j-t}$ and 
$c_{m+2^j-t}+c_{m+1}+2^j-t-1\geq c_{2m+2^j-t+1}$,
 we have that $c_{2m+1+s}+c_{2m+1}+s=c_{4m+2+s}$ if, and only if,
\begin{align}
& c_{m+2^j-t}+c_m+2^j-t= c_{2m+2^j-t} \label{eq:eqodd3-1}\\
& c_{m+2^j-t}+c_{m+1}+2^j-t-1=c_{2m+2^j-t+1} \label{eq:eqodd3-2}
\end{align}
So, we must prove that Eqns.~\eqref{eq:eqodd3-1} and (\ref{eq:eqodd3-2}) hold for $m\geq 1$ if, and only if, $m=2^{j}p+t$ for some $p\geq 1$. We distinguish again two subcases, depending on the parity of $t$:

\begin{itemize}

\item If $t=2x$ for some $0\leq x<2^{j-2}$, then Eqn.~\eqref{eq:eqodd3-1} and Lemma  \ref{lem:eqeven} imply that $m$ is even, say $m=2m_0$ with $m_0\geq 1$, and then Eqn.~\eqref{eq:eqodd3-2} can be written 
$$
c_{2m_0+1+2^j-2x-1}+c_{2m_0+1}+2^j-2x-1= c_{4m_0+2+2^j-2x-1}
$$ 
which, by induction,  is equivalent to $m_0=2^{j-1}p+x$ for some $p\geq 1$, that is, to $m=2m_0=2^{j}p+t$ for some $p\geq 1$. Hence, to complete the proof of the desired equivalence, it remains to check that if $m=2^{j}p+t$, then Eqn.~\eqref{eq:eqodd3-1} holds. Now, if $x=0$, so that $m=2^{j}p$,  Corollary \ref{colless_minimum} and Lemma   \ref{lem:eqeven} clearly imply Eqn.~\eqref{eq:eqodd3-1} (cf. the case when $t$ is even in the proof of Lemma \ref{lem:eqodd2}). So, assume that $x>0$ and write it as $x=2^i(2y_0+1)$ with $0\leq i< j-2$ and $y_0<2^{j-i-3}$. Then
$$
\begin{array}{l}
c_{m+2^j-t}+c_m+2^j-t\\
\quad =c_{2^{j}p+2x+2^j-2x}+c_{2^{j}p+2x}+2^j-2x\\
\quad= c_{2^{j}p+2^{i+1}(2y_0+1)+2^{j}-2^{i+1}(2y_0+1)}\\
\qquad\qquad +c_{2^{j}p+2^{i+1}(2y_0+1)}+2^{j}-2^{i+1}(2y_0+1)\\
\quad =2^{i+1}\big(c_{2^{j-i-1}p+2y_0+1+2^{j-i-1} -(2y_0+1)}\\
\qquad\qquad +c_{2^{j-i-1}p+2y_0+1}+2^{j-i-1}-(2y_0+1)\big)\\
\quad =2^{i+1} c_{2^{j-i}p+4y_0+2+2^{j-i-1}-(2y_0+1)}\quad
\mbox{(by the induction hypothesis)}\\
\quad =c_{2^{j+1}p+2^{j}+2^{i+1}(2y_0+1)}=c_{2^{j+1}p+2^{j}+2x}\\
\quad =c_{2m+2^j-t}
\end{array}
$$
as we wanted to prove.

\item If $t=2x+1$ for some $0\leq x<2^{j-2}$,  Eqn.~\eqref{eq:eqodd3-2} and Lemma  \ref{lem:eqeven} imply that $m+1$ is even, and then $m$ is odd, say $m=2m_0+1$ for some $m_0\geq 0$, and  Eqn.~\eqref{eq:eqodd3-1} can be written 
\begin{equation}
c_{2m_0+1+2^j-2x-1}+c_{2m_0+1}+2^j-2x-1= c_{4m_0+2+2^j-2x-1}.
\label{eq:encaraunaltra}
\end{equation}
Now, if $m_0=0$, Remark \ref{rem:postkey}.(a) implies that this equality holds if, and only if, $2^j-2x-1=1$ which, under the condition $0\leq x<2^{j-2}$, only happens when $j=1$ and $x=0$, but then $t=1=2^{j-1}$ against the assumption that  $t<2^{j-1}$. Therefore $m_0$ must be at least 1.

Then, by induction, Identity (\ref{eq:encaraunaltra})
 is equivalent to $m_0=2^{j-1}p+x$ for some $p\geq 1$, that is, to $m=2m_0+1=2^{j}p+2x+1=2^{j}p+t$ for some $p\geq 1$.  So, to complete the proof of the desired equivalence, it remains to check that if $m=2^{j}p+t$, then Eqn.~\eqref{eq:eqodd3-2} holds.  Now, in the current situation:
$$
\begin{array}{l}
c_{m+2^j-t}+c_{m+1}+2^j-t-1\\
\qquad =c_{2^{j}p+2x+1+2^j-2x-1}+c_{2^{j}p+2x+2}+2^j-2x-2\\
\qquad =c_{2^{j}p+2^j}+c_{2^{j}p+2x+2}+2^j-2x-2\\
\qquad = 2\big(c_{2^{j-1}p +2^{j-1}}+c_{2^{j-1}p+x+1}+2^{j-1}-x-1\big)\\
\qquad = 2\big(c_{(2^{j-1}p+x+1)+(2^{j-1}-x-1)}+c_{2^{j-1}p+x+1}+2^{j-1}-x-1\big)=(**)
\end{array}
$$

If $x$ is even, say $x=2x_0$ with $0\leq x_0<2^{j-3}$, then 
\begin{align*}
(**) & =2\big(c_{(2^{j-1}p+2x_0+1)+(2^{j-1}-2x_0-1)}+c_{2^{j-1}p+2x_0+1}+2^{j-1}-2x_0-1\big)\\ 
& = 2c_{2^{j}p+2(2x_0+1)+2^{j-1}-(2x_0+1)} \quad \mbox{ (by the induction hypothesis)}\\
&=c_{2^{j+1}p+2^{j}+4x_0+2} =c_{2m+2^j-t+1}.
\end{align*}

And if $x$ is odd, write it as $x=2^i(2t_0+1)-1$ with $1\leq i<j-1$ and $t_0<2^{j-i-3}$, and then
\begin{align*}
(**) & =2\big(c_{2^{j-1}p+2^i(2t_0+1)+2^{j-1}-2^i(2t_0+1)}\\ & \qquad +c_{2^{j-1}p+2^i(2t_0+1)}+2^{j-1}-2^i(2t_0+1)\big)\\
 & = 2^{i+1}\big(c_{2^{j-i-1}p+2t_0+1+2^{j-i-1}-(2t_0+1)}\\ & \qquad+c_{2^{j-i-1}p+2t_0+1}+2^{j-i-1}-(2t_0+1)\big)\\
& = 2^{i+1}c_{2^{j-i}p+4t_0+2+2^{j-i-1}-(2t_0+1)} \quad \mbox{ (by the induction hypothesis)}\\
& = c_{2^{j+1}p+2^{i+1}(2t_0+1)+2^{j}}=c_{2^{j+1}p+2x+2+2^{j}}\\ &=c_{2m+2^j-t+1}
\end{align*}
This completes the proof of the desired equivalence when $t$ is odd.  \qed
\end{itemize}
\end{proof}

We are now in a position to proceed with the proof of Proposition \ref{prop:eqC}. Assertion (1) in it is a direct consequence of Corollary \ref{colless_minimum}. So, assume $n_a>n_b$ and set $s=n_a-n_b$, so that $n_a=n_b+s$. Then:
\begin{enumerate}[(a)]
\item If $s=1$, then, by Lemma \ref{lem:eqodd3}, $c_{n_a}+c_{n_b}+n_a-n_b=c_{n_a+n_b}$ for every $n_b\geq 1$.

\item If $s>1$ is odd, write it as $s=2^{j+1}-(2t+1)$, with $j=\floor*{\log_2(s)}\geq 1$ and  $0\leq t<2^{j-1}$. Then, by Lemmas \ref{lem:eqodd2} and \ref{lem:eqodd3}, $c_{n_a}+c_{n_b}+n_a-n_b=c_{n_a+n_b}$ if, and only if, either
$n_b=2^{j+1}p$ or $n_b=2^{j+1}p+2t+1$, for some $p\geq 1$.

\item If $s\geq 2$ is even, write it as $s=2^ks_0$, with $k\geq 1$ the largest exponent of a power of 2 that divides $s$ and $s_0$  an odd integer, and write the latter as $s_0=2^{j+1}-(2t+1)$ with $j=\floor*{\log_2(s_0)}\geq 0$ and  $0\leq t<2^{j-1}$. Then, by Lemma \ref{lem:eqeven}, $c_{n_a}+c_{n_b}+n_a-n_b=c_{n_a+n_b}$ if, and only if,
$n_b=2^km$, for some $m\geq 1$ such that $c_{m+s_0}+c_m+s_0=c_{2m+s_0},$ and then:
\begin{itemize}
\item If $s_0=1$ (equivalently, if $j=0$), $c_{m+s_0}+c_m+s_0=c_{2m+s_0}$ for every $m\geq 1$ and therefore, in this case, $c_{n_a}+c_{n_b}+n_a-n_b=c_{n_a+n_b}$ for every $n_b=2^km$ with $m\geq 1$.
\item If $s_0>1$  (equivalently, if $j>0$), Lemmas \ref{lem:eqodd2} and \ref{lem:eqodd3} imply that $c_{m+s_0}+c_m+s_0=c_{2m+s_0}$ if, and only if, $m=2^{j+1}p$ or $m=2^{j+1}p+2t+1$, for some $p\geq 1$. Therefore, in this case, $c_{n_a}+c_{n_b}+n_a-n_b=c_{n_a+n_b}$ if, and only if, 
$n_b=2^{k+j+1}p$ or $n_b=2^k(2^{j+1}p+2t+1)$, for some $p\geq 1$.
\end{itemize}
\end{enumerate}

Combining the three cases, and taking $k=0$ in the odd $s$ case, we conclude that $$c_{n_a}+c_{n_b}+n_a-n_b=c_{n_a+n_b}$$ if, and only if, writing $n_a-n_b=2^k(2^{j+1}-(2t+1))$ (for some $k\geq 0$, $j\geq 0$, and $0\leq t<2^{j-1}$),
\begin{itemize} 
\item If $j=0$, then $n_b=2^kp$ for some $p\geq 1$, in which case $n_a=2^k(p+1)$ and $n=2^{k}(2p+1)$.

\item If $j>0$, then there exists some $p\geq 1$ for which one of the following conditions holds:
\begin{itemize}
\item $n_b=2^{k+j+1}p$,  in which case $n_a=2^k(2^{j+1}(p+1)-(2t+1))$ and $n=2^k(2^{j+1}(2p+1)-(2t+1))$.

\item  $n_b=2^k(2^{j+1}p+2t+1)$, $n_a=2^{k+j+1}(p+1)$ and $n=2^k(2^{j+1}(2p+1)+ 2t+1)$.
\end{itemize}
\end{itemize}
This is equivalent to the expressions for $n_a$ and $n_b$ in option (2) in the statement (replacing $j+1$ with $j>0$ by $l\geq 2$). 

This completes the proof of Proposition \ref{prop:eqC}.

\subsection*{A.2 Proof of Proposition \ref{GFB_Decomposition}}

This appendix is devoted to establish the following result.

\setcounter{proposition}{4}
\begin{proposition} 
Let $T_n^{\mathit{gfb}}=(T_a,T_b)$ be a GFB tree with $n\geq 2$, $T_a\in \TT_{n_a}$, $T_b\in \TT_{n_b}$ and $n_a\geq n_b$. Let $n=2^m+p$ with $m=\lfloor \log_2(n)\rfloor$ and $0\leq p<2^m$. Then, we have:
	\begin{enumerate}[\rm (i)]
	\item If $0\leq p\leq 2^{m-1}$, then $n_a = 2^{m-1}+p$, $n_b = 2^{m-1}$ and $T_b$ is fully symmetric. 
	\item If $2^{m-1}\leq p<2^m$, $n_a = 2^{m}$, $n_b=p$ and $T_a$ is fully symmetric.
	\end{enumerate}
\end{proposition}

The proof of this proposition  requires of the following lemma. The   idea guiding its proof  is illustrated in Figure \ref{Fig_treeset}.

\begin{lemma} \label{subtree_in_common}
Let $n\geq 3$ be an odd natural number. Then, $T_n^{\mathit{gfb}}$ shares a maximal pending subtree with $T_{n-1}^{\mathit{gfb}}$ and a maximal pending subtree with $T_{n+1}^{\mathit{gfb}}$.
\end{lemma}

\begin{proof}
Since $n\geq 3$ is odd, the first $(n-1)/{2}$ iterations of the loop in Algorithm \ref{alg_gfb} result in ${(n-1)}/{2}$ cherries and a single node, which in the $(n+1)/{2}$-th iteration is clustered with a cherry to form a tree with 3 leaves. From this moment on, as the algorithm continues clustering trees, in each $i$-th iteration there will be one, and only one, tree $T_i^{odd}$ with an odd number $s(i)$ of leaves. Note now that, on the one hand, this unique tree with $s(i)$ leaves is treated by the algorithm like a tree with $s(i)-1$ leaves, except that it is clustered as late as possible, i.e. when all other trees in $\treeset$ with $s(i)-1$ leaves (if there are any) have already been clustered. On the other hand, however, this tree is also treated by the algorithm like a tree with $s(i)+1$ leaves,  except that it is clustered as early as possible, i.e. before any other elements in $\treeset$ with $s(i)+1$ leaves (if there are any) get clustered. So, to summarize, after the first $i\geq (n+1)/{2}$ iterations of the loop, $\treeset$ contains a unique tree $T_i^{odd}$ with an odd number $s(i)$ of leaves, which at the same time
\begin{enumerate}
\item [(i)] is treated like a tree with $s(i)-1$ leaves, but is clustered as late as possible;
\item [(ii)] is treated like a tree with $s(i)+1$ leaves, but is clustered as soon as possible.
\end{enumerate}

Now, first consider Algorithm \ref{gfb} for $n-1$, which is an even number.
After the first $(n-3)/{2}$ iterations of the loop, $\treeset$ contains $(n-3)/{2}$ trees with 2 leaves and two trees with 1 leaf, which are clustered last to form the last cherry. We keep tracking one leaf $u$ of this cherry throughout the algorithm. The algorithm at this stage contains only cherries, which are all isomorphic, so without loss of generality, we may assume that $u$ is contained in the one that gets clustered with another tree last, i.e. after all other cherries have been clustered. We continue like this, always assuming without loss of generality (when there is more than one tree in $\treeset$ of the same size as the tree that contains $u$) that the tree containing $u$ is in the last one to be clustered. By (i), this means that if we replace $u$ in $T_{n-1}^{\mathit{gfb}}$ by a cherry, we derive $T_{n}^{\mathit{gfb}}$. This is due to the fact that in the analogous step where $\treeset$ for $n-1$ only contains cherries, $\treeset$ for $n$ will contain only cherries and a tree containing three leaves. This triplet will subsequently act like a cherry, but like the one that happens to be clustered last. So,  we identify the cherry {in this} triplet to {$u$} to see the correspondence between $T_{n-1}^{\mathit{gfb}}$ and $T_{n}^{\mathit{gfb}}$. Note that this directly implies that $T_{n-1}^{\mathit{gfb}}$ and $T_{n}^{\mathit{gfb}}$ share a common maximal pending subtree ---namely the one that does {\em not} contain $u$.

Note that by (ii), an analogous procedure for $n+1$ leads to $T_{n+1}^{\mathit{gfb}}$ and $T_{n}^{\mathit{gfb}}$ sharing a common maximal pending subtree. In this case, we track a cherry in $T_{n+1}^{\mathit{gfb}}$, namely the one that happens to be clustered first, and replace it by a single leaf to see the correspondence between $T_{n+1}^{\mathit{gfb}}$ and $T_{n}^{\mathit{gfb}}$. This completes the proof. \qed
\end{proof}

\begin{figure}
	\centering
	\includegraphics[width=\linewidth]{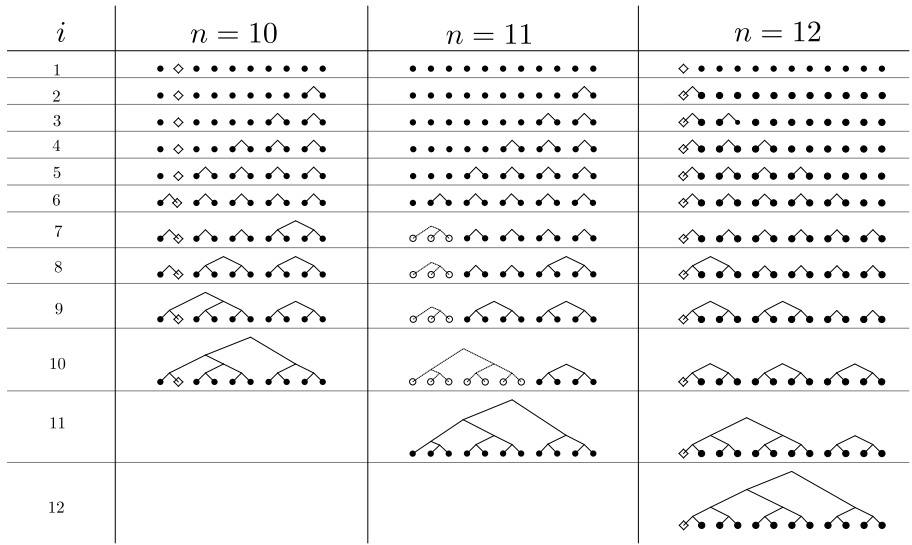}
	\caption{Content of $\treeset$ before the $i^{\rm th}$ iteration of the loop in Algorithm \ref{alg_gfb} for $n=10, n=11$ and $n=12$. In case of $n=11$, the tree {with white leaves} for $i=7, \ldots, 10$, depicts the \emph{unique} tree in $\treeset$ with an odd number of leaves. For $n=10$, the leaf depicted as a diamond represents leaf $u$ used in the proof of Lemma \ref{subtree_in_common}. Note that the tree containing this leaf is always clustered as late as possible. In case of $n=12$, the leaf depicted as a diamond again represents leaf $u$ used in the proof of Lemma \ref{subtree_in_common}. In this case, the tree containing this leaf is always clustered as soon as possible. The last tree depicted in each column represents the GFB tree. Note that $T_n^{\mathit{gfb}}$ can be obtained from $T_{n-1}^{\mathit{gfb}}$ by replacing the leaf depicted as a diamond by a cherry. Moreover, $T_n^{\mathit{gfb}}$ can be obtained from $T_{n+1}^{\mathit{gfb}}$ by replacing the cherry containing the diamond leaf by a single leaf. }
	\label{Fig_treeset}
\end{figure}

We can proceed now to prove Proposition \ref{GFB_Decomposition}. Let $n=2^m+p$ with $m=\lfloor \log_2(n)\rfloor$ and $0\leq p<2^m$.  We shall prove by induction on $n$ that if $T_n^{\mathit{gfb}}=(T_a,T_b)$ is a GFB tree with $n\geq 2$, $T_a\in \TT_{n_a}$, $T_b\in \TT_{n_b}$ and $n_a\geq n_b$ then:
\begin{enumerate}[(i)]
\item If $0\leq p\leq 2^{m-1}$, $n_a = 2^{m-1}+p$ and $n_b = 2^{m-1}$ and then $T_b$ is fully symmetric.
\item If $2^{m-1}\leq p<2^m$, we have $n_a = 2^{m}$ and $n_b=p$ and then $T_a$ is fully symmetric.
\end{enumerate}
We want to point out that we understand that the conjunction of these two assertions in the case when both premises are satisfied, namely when $p=2^{m-1}$, says that $n_a =2^m$ and  $n_b = 2^{m-1}$ and then both $T_a$ (by (ii)) and $T_b$ (by (i)) are fully symmetric.

The base case for (i) is when $n=2$ and for (ii), when $n=3$. In both cases the assertions are obvious, because there is only one bifurcating tree with $2=2^1+0$ leaves (a cherry with $n_a=n_b=1=2^{0}$) and only one  bifurcating tree with $3=2^1+1$ leaves (a caterpillar with $n_a=2=2^1$ and $n_b=1$).

%
%
%

Now, let $n \geq 4$ and assume that (i) and (ii) hold for up to $n-1$ leaves. Let $T=(T_a,T_b)$ be a GFB tree with $n$ leaves, with $T_a\in \TT_{n_a}$, $T_b\in \TT_{n_b}$ and $n_a\geq n_b$ Recall that $T_a$ and $T_b$ are  again GFB trees by Lemma \ref{gfb_subtrees}. We distinguish two cases, depending on the parity of $n$:
\begin{itemize}
\item Assume that $n$ is  even, say $n=2n_0$ with $n_0\geq 2$. In this case, Algorithm \ref{alg_gfb} results in a tree $T^{\mathit{gfb}}_n$ with $n_0$ cherries (because in each of the first $n_0$ iterations of the loop a pair of nodes are merged into a cherry). We now consider the tree $T'$ with $n_0$ leaves that is obtained from $T^{\mathit{gfb}}_n$ by replacing all cherries by single leaves. Let $T'=(T_a',T_b')$ be the decomposition into maximal pending subtrees, with $T_a'\in \TT_{n_a'}$, $T_b'\in \TT_{n_b'}$ and $n_a'\geq n_b'$. By construction, $T_a$ and $T_b$ are obtained by replacing the leaves of $T_a'$ and $T_b'$ by cherries, and therefore, in particular,
$n_a=2n_a'$ and $n_b=2n_b'$.
Note now that, since $T^{\mathit{gfb}}_n$ is a GFB tree, so is $T'$ (because as soon as Algorithm \ref{alg_gfb} only has cherries to choose from, they are treated like leaves). Note also that, since $n$ is even, so is $p$, say $p=2p_0$, and $n_0=2^{m-1}+p_0$. Then we have that:
\begin{enumerate}[(i)]
\item If $0\leq p\leq 2^{m-1}$, then $0\leq p_0\leq 2^{m-1-1}$ and hence, by the induction hypothesis,
$n_a'=2^{m-2}+p_0$, $n_b'=2^{m-2}$, and $T_b'$ is fully symmetric, which implies that $n_a=2n_a'=2^{m-1}+2p_0=2^{m-1}+p$, $n_b=2n_b'=2^{m-1}$, and $T_b$ is fully symmetric, because it is obtained from the fully symmetric tree $T_b'$ by replacing all its leaves by cherries.

\item If $2^{m-1}\leq p< 2^{m}$, then $2^{m-1-1}\leq p_0\leq 2^{m-1}$ and hence, by the induction hypothesis,
$n_a'=2^{m-1}$, $n_b'=p_0$, and $T_a'$ is fully symmetric, which implies that $n_a=2n_a'=2^{m}$, $n_b=2n_b'=2p_0=p$, and, arguing as in (i), $T_a$ is fully symmetric.
\end{enumerate}

\item Assume that $n$ is odd, say $n=2n_0+1$ with $n_0\geq 2$. In this case both $n-1=2n_0$ and $n+1=2(n_0+1)$ are even.
Write $n=2^m+p$ and $p=2p_0+1$, so that $n_0=2^{m-1}+p_0$ with $0\leq p_0<2^{m-1}$.
Let $T^1 \coloneqq T_{n-1}^{\mathit{gfb}}$ and $T^2\coloneqq T_{n+1}^{\mathit{gfb}}$. The tree $T^1$ satisfies (i) and (ii) by the induction hypothesis, and it can be proved that $T^2$ also satisfies these assertions by arguing as in the previous case when $n$ is even (i.e., replacing the pending $n_0+1$ cherries in $T^2$ by single leaves, noticing that the resulting tree is GFB, applying the induction hypothesis to it and finally returning back to $T^2$ by replacing leaves by cherries).  Let $T^1=(T^1_a, T^1_b)$ ---with $T^1_a\in \TT_{n^1_a}$ and $T^1_b\in \TT_{n^1_b}$ and $n^1_a\geq n^1_b$--- and $T^2=(T^2_a, T^2_b)$ ---with $T^2_a\in \TT_{n^2_a}$ and $T^2_b\in \TT_{n^2_b}$ and $n^2_a\geq n^2_b$--- denote the  decompositions of $T^1$ and $T^2$ into maximal pending subtrees, respectively.
Note that, since $n$ is odd, $p\neq 0,2^{m-1}$. Now we have:
\begin{enumerate}[(i)]
\item If $0< p<2^{m-1}$, then $n-1=2^m+(p-1)$ with $0\leq p-1<2^{m-1}$ and $n+1=2^m+(p+1)$ with 
$0<p+1\leq 2^{m-1}$. Then, since $T^1$ and $T^2$ satisfy assertion (i),
$$
n^1_a=2^{m-1}+p-1,\ n^1_b=2^{m-1},\  n^2_a=2^{m-1}+p+1,\ n^2_b=2^{m-1}
$$
and both $T_b^1$ and $T_b^2$ are fully symmetric and hence (since they have the same numbers of leaves) $T_b^1=T_b^2$. 

Now, we know by Lemma \ref{subtree_in_common} that $T$ shares a maximal pending subtree with $T^1$ and a maximal pending subtree with $T^2$. Looking at the numbers of leaves of the 
maximal pending subtrees of $T^1$ and $T^2$, one easily deduces that the only possibility for this to happen is that $T$ shares with $T^1$ and $T^2$ the same maximal pending subtree:  the fully symmetric subtree $T_b^1=T_b^2$. (Indeed, since $T_a^1\neq T_a^2$, because they have different numbers of leaves, if $T$ did not share $T_b^1=T_b^2$ with both $T^1$ and $T^2$, then it would have a maximal pending subtree in common with $T^1$ and the other maximal pending subtree in common with $T^2$, but no combination of a maximal pending subtree of $T^1$ and a maximal pending subtree of $T^2$ yields a tree with $2^m+p$ leaves.)
\textsl{A fortiori}, one of the maximal pending subtrees of $T$ is a fully symmetric tree with $2^{m-1}$ leaves and the other must have, thus, the remaining $2^{m-1}+p$ leaves. This shows that $n_a=2^{m-1}+p$ and $n_b=2^{m-1}$ and $T_b$ is fully symmetric.

\item If $2^{m-1}<p\leq 2^{m}-3$ then $n-1=2^m+(p-1)$ with $2^{m-1}\leq p-1<2^{m}$ and $n+1=2^m+(p+1)$ with 
$2^{m-1}<p+1< 2^{m}$. Then, since $T^1$ and $T^2$ satisfy assertion (ii),
$$
n^1_a=2^{m},\ n^1_b=p-1,\  n^2_a=2^{m},\ n^2_b=p+1
$$
and both $T_a^1$ and $T_a^2$ are fully symmetric and hence (since they have the same numbers of leaves) 
$T_a^1=T_a^2$. Reasoning as in the previous case, we deduce that $T$ shares with  both $T^1$ and $T^2$ the fully symmetric maximal pending subtree $T_a^1=T_a^2$. In particular, one of its maximal pending subtrees has $2^{m}$ leaves (and it is fully symmetric) and the other must have, thus, the remaining $p$ leaves. This shows that $n_a=2^{m}$ and $n_b=p$ and $T_a$ is fully symmetric.

\item Consider finally the case when $p= 2^{m}-1> 2^{m-1}$. Then, $n-1=2^m+(p-1)$ with $2^{m-1}\leq p-1<2^{m}$ and $n+1=2^{m+1}$. In this case, since $T^1$ satisfies assertion (ii) and $T^2$ satisfies assertion (i),
$$
n^1_a=2^{m},\ n^1_b=2^{m}-2,\  n^2_a=2^{m},\ n^2_b=2^{m}
$$
and $T_a^1$, $T_b^1$ and $T_b^2$ are fully symmetric and hence (since they have the same numbers of leaves) 
$T_a^1=T_b^1=T_b^2$. Arguing as in the previous cases we conclude that $T$ has a maximal pending subtree with $2^m$ leaves that is fully symmetric and the other maximal pending subtree with the remaining $2^m-1$ leaves, and hence it satisfies assertion (ii). 
\end{enumerate}
\end{itemize}
This completes the proof. 

{\subsection*{A.3 Proof of Proposition \ref{prop:GFBmax}}

This appendix is devoted to establish the following result.

\setcounter{proposition}{7}
\begin{proposition}
For every $n\geq 1$, let $n=\sum_{i=1}^\ell 2^{m_i}$, with $\ell\geq 1$ and $m_1>\cdots> m_\ell$, be its binary expansion.
\begin{enumerate}[(a)]
\item $s(T_n^\mathit{gfb})=n-1-(m_1-m_\ell)$.
\item For every  $T\in \widetilde{\mathcal{MC}}_n$, if $T\neq T_n^\mathit{gfb}$, then $s(T)< s(T_n^\mathit{gfb})$.
\end{enumerate}
\end{proposition}

\begin{proof}
Note first of all that the number $s$ of symmetry vertices satisfies the following recurrence: if $T\in \TT_1$, then $s(T)=0$, and if $T=(T_a,T_b)\in \TT_n$ with $n\geq 2$, then 
\begin{align}\label{s_decomp}
s(T) = \begin{cases}
s(T_a) + s(T_b) + 1, &\text{ if $T_a$ and $T_b$ are isomorphic},\\
s(T_a)+ s(T_b) &\text{ otherwise}.
\end{cases}
\end{align}

We shall now prove (a) by induction on $n$. When $n=1=2^0$, the statement holds because $s(T_1^\mathit{gfb})=0 = 1-1-(0-0) = n-1-(m_1-m_\ell)$. 
More in general, the statement clearly holds whenever $n$ is a power of 2, say $n=2^{m_1}$, because in this case $\Tgfb$ is fully symmetric and therefore all its internal nodes are symmetry vertices, i.e.  $s(\Tgfb)=n-1=n-1-(m_1-m_1)$.

Now assume that the statement holds for every GFB tree with $n'$ leaves, with $n'<n$, and consider the tree $\Tgfb$.  By Lemma \ref{gfb_subtrees}, if $\Tgfb = (T_a, T_b)$, then $T_a$ and $T_b$ are GFB trees and, by the inductive hypothesis, the statement holds for $T_a=T_{n_a^\mathit{gfb}}^\mathit{gfb}$ and $T_b=T_{n_b^\mathit{gfb}}^\mathit{gfb}$.

Let us now write $n$ as $2^m+p$ with $m = \lfloor \log_2(n) \rfloor$ and $0 \leq p < 2^m$, and consider its binary expansion
$n = \sum_{j=1}^\ell 2^{m_j}$ with $m_1>\cdots>m_\ell$, so that $m_1=m$ and $p = \sum_{j=2}^\ell 2^{m_j}$ is the binary expansion of $p$ if $p>0$. Now, we distinguish four cases:

\begin{enumerate}
\item[(i)] If $p=0$, then $n$ is a power of 2, in which case we have already seen that the statement holds.

\item [(ii)] If $1\leq p<2^{m-1}$, then, by Proposition \ref{GFB_Decomposition}, $n_a^\mathit{gfb} = 2^{m-1} + p$ and $n_b^\mathit{gfb} = 2^{m-1}$ and $T_b$ is fully symmetric. In this case, $m_2<m-1=m_1-1$ and thus $n_a^\mathit{gfb} = 2^{m_1 - 1} + \sum_{j=2}^\ell 2^{m_j}$ is the binary expansion of $n_a^\mathit{gfb}$.
Then $s(T_b)=2^{m-1}-1$ and, by the induction hypothesis, 
$$
s(T_a)=2^{m-1} + p-1-(m_1-1-m_\ell)=2^{m-1} + p-(m_1-m_\ell)
$$ and hence
$$
s(\Tgfb) = s(T_a) + s(T_b)= 2^{m-1} + p-(m_1-m_\ell) + 2^{m-1}-1 = n- 1- (m_1 - m_\ell).
$$

\item [(iii)] If $p= 2^{m-1}$, so that $n=2^m+2^{m-1}$ is the binary expansion of $n$, then, by Proposition \ref{GFB_Decomposition}, $n_a^\mathit{gfb} = 2^{m}$ and $n_b^\mathit{gfb}=2^{m-1}$ and both $T_a$ and $T_b$ are fully symmetric. In this case, $s(T_a)=2^m-1$ and $s(T_b)=2^{m-1}-1$ and hence
\begin{align*}
s(\Tgfb) & = s(T_a) + s(T_b)= 2^m-1 + 2^{m-1}-1 \\
& = 2^m + 2^{m-1} - 1 - (m-(m-1)) = n-1-(m_1-m_\ell).
\end{align*}

\item [(iv)] Finally, assume that $p > 2^{m-1}$, so that its binary expansion is $p= 2^{m-1} + \sum_{i=3}^\ell 2^{m_i}$, and in particular $m_2=m-1=m_1-1$. In this case, by Proposition \ref{GFB_Decomposition}, $n_a^\mathit{gfb} = 2^m$, and $T_a$ is fully symmetric, and $n_b^\mathit{gfb} = p$. Then, $s(T_a)=2^m-1$ and, by the induction hypothesis,
$s(T_b)=p-1-(m_1-1-m_\ell)=p-(m_1-m_\ell)$ and hence
$$
s(T) = s(T_a) + s(T_b) = 2^{m}-1+p-(m_1-m_\ell)= n-1-(m_1 - m_\ell).
$$
\end{enumerate}
This completes the proof of (a).

As far as (b) goes, we also prove it by induction on $n$. The case  $n=1$ is obvious, since there is only one bifurcating tree in $\TT_1$. Let now $n\geq 2$ and assume that the statement is true for every number $n'$ of leaves smaller than $n$. Let $T = (T_a,T_b)$, with $T_a\in \TT_{n_a}$, $T_b\in \TT_{n_b}$, and $n_a\geq n_b$,
be a minimal Colless tree with $n$ leaves such that $s(T)$ is maximum in $\widetilde{\mathcal{MC}}_n$. We want to prove that $T=\Tgfb$.

By Lemma \ref{max_subtrees}, $T_a\in \widetilde{\mathcal{MC}}_{n_a}$ and $T_b\in \widetilde{\mathcal{MC}}_{n_b}$  and therefore, by the inductive hypothesis, $s(T_a)\leq s(T_{n_a}^\mathit{gfb})$ and $s(T_b)\leq s(T_{n_b}^\mathit{gfb})$. To prove that $T=\Tgfb=(T_{n_a^\mathit{gfb}}^\mathit{gfb},T_{n_b^\mathit{gfb}}^\mathit{gfb})$, it is enough to prove that
$n_a=n_a^\mathit{gfb}$ and $n_b=n_b^\mathit{gfb}$ (and, actually, it is enough to prove one of these equalities, because then the other will follow from $n_a+n_b=n=n_a^\mathit{gfb}+n_b^\mathit{gfb}$)
and that
$s(T_a)= s(T_{n_a}^\mathit{gfb})$ and $s(T_b)= s(T_{n_b}^\mathit{gfb})$ (because by the inductive hypothesis 
these equalities imply that $T_a=T_{n_a}^\mathit{gfb}$ and $T_b=T_{n_b}^\mathit{gfb}$). Let $n_a=\sum_{i=1}^{\ell_a} 2^{s_i}$ and 
$n_b=\sum_{i=1}^{\ell_b} 2^{t_i}$ be the binary decompositions of $n_a$ and $n_b$.

Now, two cases arise, depending on whether the root  of $T$ is a symmetry vertex or not. 
Let us assume first that it is a symmetry vertex, i.e, that $T_a=T_b$.  In this case, $n$ must be even and $n_a = n_b = n/2=\sum_{i=1}^{\ell} 2^{m_i-1}$. In particular $s_1=t_1=m_1-1$ and $s_{\ell_a}=t_{\ell_b}=m_{\ell}-1$.
Moreover, it must happen that $s(T_a)= s(T_{n/2}^\mathit{gfb})$, because if 
$s(T_a)< s(T_{n/2}^\mathit{gfb})$ and if we denote by $T'$ the tree $(T_{n/2}^\mathit{gfb},T_{n/2}^\mathit{gfb})$, then $T'\in \widetilde{\mathcal{MC}}_n$ by Proposition \ref{lem:charmin1} (recall that $(n/2,n/2)$ always belongs to $QB(n)$) and, by 
Eqn. (\ref{s_decomp}),
$$
s(T)=2s(T_a)+1<2s(T_{n/2}^\mathit{gfb})+1=s(T')
$$
against the assumption that $s(T)$ is maximum in $\widetilde{\mathcal{MC}}_n$. So, in this case it remains to prove that $n_a^\mathit{gfb}=n_b^\mathit{gfb}=n/2$.

Now, applying Eqn. (\ref{s_decomp}) and (a), we have that 
\begin{align*}
s(T) & =2s(T_a) + 1  =2s(T_{n/2}^\mathit{gfb}) +  1 = 2\Big(\frac{n}{2}-1-(s_1-s_{\ell_a})\Big)+1\\
& = n-1-2(m_1-1-m_\ell+1)=n-1-2(m_1-m_\ell)\\
&=s(T_{n}^\mathit{gfb})-(m_1-m_\ell).
\end{align*}
Thus, if $\ell>1$, then $s(T)<s(T_n^\mathit{gfb})$, against the assumption that $s(T)$ is maximum in $\widetilde{\mathcal{MC}}_n$. Therefore, $\ell=1$, i.e., $n=2^{m_1}$ and hence $n_a^\mathit{gfb}=n_b^\mathit{gfb}=n/2=n_a=n_b$, as we wanted to prove.

Let us assume now that the root of $T$ is not a symmetry vertex. Recall from Corollary \ref{leaf_partioning} that
$$
n_b^\mathit{gfb}\leq n_b\leq n_a\leq n_a^\mathit{gfb}.
$$
Combining these inequalities with  Proposition \ref{GFB_Decomposition} we obtain that:
\begin{itemize}
\item If $0\leq p< 2^{m_1-1}$, then 
\begin{equation}
2^{m_1-1}=n_b^\mathit{gfb}\leq n_b\leq n_a\leq n_a^\mathit{gfb}= 2^{m_1-1}+p<2^{m_1},
\label{eqn:maxsitem1}
\end{equation}
and then, in this case, $s_1=t_1=m_1-1$.

\item If $2^{m_1-1}\leq p<2^{m_1}$, then 
\begin{equation}
2^{m_1-1}\leq p
=n_b^\mathit{gfb}\leq n_b\leq n_a\leq n_a^\mathit{gfb}=2^{m_1}
\label{eqn:maxsitem2}
\end{equation}
and then either $n_a=2^{m_1}=n_a^\mathit{gfb}$, in which case $n_b=p=n_b^\mathit{gfb}<2^{m_1}$, $s_1=m_1$, and $t_1=m_1-1$,
or $2^{m_1-1}\leq n_b\leq n_a<n_a^\mathit{gfb}=2^{m_1}$, in which case $s_1=t_1=m_1-1$.
\end{itemize}
So, in particular, $t_1$ is always $m_1-1$, and $s_1$ is  $m_1$,  when $n_a=2^{m_1}=n_a^\mathit{gfb}$, and  $m_1-1$ otherwise. Moreover, since $n_a+n_b=n$, it always happens that $\min\{s_{\ell_a}, t_{\ell_b}\} \leq m_\ell$. 

Now, in this case we have again that $s(T_a)=s(T_{n_a}^\mathit{gfb})$ and $s(T_b)=s(T_{n_b}^\mathit{gfb})$, because if, say, $s(T_a)<s(T_{n_a}^\mathit{gfb})$ and if we replace in $T$ its maximal pending subtree $T_a$ by $T_{n_a}^\mathit{gfb}$, then, by Proposition \ref{lem:charmin1} (and recalling that, since $T\in \widetilde{\mathcal{MC}}_n$, by that very proposition we have that $(n_a,n_b)\in QB(n)$), the resulting tree $T'=(T_{n_a}^\mathit{gfb},T_b)$ is still minimal Colless and, by Eqn. (\ref{s_decomp}),
$$
s(T) =s(T_a)+s(T_b)<s(T_{n_a}^\mathit{gfb})+s(T_b)\leq s(T')
$$
against the assumption that $s(T)$ is maximum in $\widetilde{\mathcal{MC}}_n$. So,
$T_a=T_{n_a}^\mathit{gfb}$ and $T_b=T_{n_b}^\mathit{gfb}$. 
It remains to prove that $n_a=n_a^\mathit{gfb}$ and $n_b=n_b^\mathit{gfb}$.

By Eqn. (\ref{s_decomp}) and (a), we have that
\begin{align}
s(T) & =s(T_a) + s(T_b)  =s(T_{n_a}^\mathit{gfb}) + s(T_{n_b}^\mathit{gfb})\nonumber\\
& = n_a-1-(s_1-s_{\ell_a})+n_b-1-(t_1-t_{\ell_b})\nonumber\\
& = n-2-(s_1-s_{\ell_a})-(m_1-1-t_{\ell_b})\nonumber\\
& = n-1-(m_1-m_\ell)-(s_1+m_\ell-s_{\ell_a}-t_{\ell_b})\nonumber\\
& =s(T_{n}^\mathit{gfb})-(s_1+m_\ell-s_{\ell_a}-t_{\ell_b})
\label{eqn:finalmaxs}
\end{align}
We consider now several possibilities:
\begin{itemize}
\item If $s_{\ell_a}=s_1$, then $n_a=2^{s_1}$, where $s_1$ is $m_1-1$ or $m_1$. Now, since $2^{m_1-1}\leq n_b\leq n_a$ and $n_a+n_b=2^{m_1}+p$, if we had $n_a=2^{m_1-1}$, we would also have $n_b=2^{m_1-1}$  and $p=0$, and then
$n_a^\mathit{gfb}=n_a=n_b=n_b^\mathit{gfb}$; but then
$T_a=T_{n_a}^\mathit{gfb}$ and $T_b=T_{n_b}^\mathit{gfb}$ would be isomorphic to the fully symmetric tree $T_{m_1-1}^\mathit{fs}$  and hence the root of $T$ would be a symmetry vertex, against the current  assumption that it is not so. 

So, in this case we have $n_a=2^{m_1}$. By  properties  (\ref{eqn:maxsitem1}) and (\ref{eqn:maxsitem2}), 
it can only happen when $2^{m_1-1}\leq p$ and $n_a= n_a^\mathit{gfb}$, and then $n_b= n_b^\mathit{gfb}$, too.

\item If $s_{\ell_a}<s_1$ and $t_{\ell_b}\leq s_{\ell_a}$, then  $t_{\ell_b}=\min\{s_{\ell_a}, t_{\ell_b}\} \leq m_\ell$ and hence, by (\ref{eqn:finalmaxs}),
$$
s(T)=s(T_{n}^\mathit{gfb})-(s_1+m_\ell-s_{\ell_a}-t_{\ell_b})<s(T_{n}^\mathit{gfb}),  
$$
against the assumption that $s(T)$ is maximum among all minimal Colless trees with $n$ leaves.
    
\item If $s_{\ell_a}<s_1$ and $s_{\ell_a}<t_{\ell_b}$, then  $s_{\ell_a}=\min\{s_{\ell_a}, t_{\ell_b}\} \leq m_\ell$. Since in this case $n_a$ is not a power of 2, we have $s_1=t_1=m_1-1$ and then
$$
s_1+m_\ell-s_{\ell_a}-t_{\ell_b}\geq m_1-1-t_{\ell_b}=t_1-t_{\ell_b}\geq 0.
$$
If one of these inequalities is strict, we deduce again that
$$
s(T)=s(T_{n}^\mathit{gfb})-(s_1+m_\ell-s_{\ell_a}-t_{\ell_b})<s(T_{n}^\mathit{gfb}),  
$$
reaching the same contradiction as before. Therefore, both  inequalities are equalities and hence $s_{\ell_a}=m_\ell$ and $t_{\ell_b}=t_1=m_1-1$, which implies in particular that $n_b=2^{m_1-1}$. By (\ref{eqn:maxsitem1}) and (\ref{eqn:maxsitem2}), this can only happen when $0\leq p\leq 2^{m_1-1}$ and $n_b=2^{m_1-1}=n_b^\mathit{gfb}$.
\end{itemize}
This finishes the proof of (b). \qed
\end{proof}
}

{\subsection*{A.4 Computation of some probabilities under the $\beta$-model}

Aldous' $\beta$ model \citep{Ald1}  is a probabilistic model of bifurcating phylogenetic trees that depends on one parameter $\beta\in (-2,\infty)$. As any other such probabilistic model, it yields a probabilistic model of bifurcating unlabeled trees, by defining the probability $P_\beta(T)$ of a tree $T\in \TT_n$ as the sum of the probabilities of all phylogenetic trees on $n$ leaves with shape $T$. This probabilistic model of bifurcating unlabeled trees satisfies the following Markovian recurrence. For every $m\geq 2$ and $k=1,\ldots,m-1$, let 
$$
q_{m,\beta}(k)=\frac{1}{a_m(\beta)}\cdot \frac{\Gamma(\beta+k+1)\Gamma(\beta+m-k+1)}{\Gamma(k+1)\Gamma(m-k+1)},
$$
where $a_m(\beta)$ is a suitable normalizing constant so that $\sum\limits_{a=1}^{m-1} q_{m,\beta}(a)=1$, and
$$
\widehat{q}_{m,\beta}(k)=\left\{\begin{array}{ll}
q_{m,\beta}(k)+q_{m,\beta}(m-k)=2q_{m,\beta}(k) & \mbox{ if $k\neq m/2$}\\
q_{m,\beta}(k) & \mbox{ if $k= m/2$}
\end{array}\right.
$$
Then, if $T=(T_{a},T_b)\in \TT_n$ with $T_a\in \TT_{n_a}$, $T_b\in \TT_{n_b}$ and $n_a\geq n_b$,
\begin{equation}
P_\beta(T)=\widehat{q}_{n,\beta}(n_a)\cdot P_\beta(T_a)P_\beta(T_b).
\label{eqn:markovbeta}
\end{equation}
Recall that the Gamma function $\Gamma$ satisfies the recurrence $\Gamma(x+1)=x\Gamma(x)$ and that, for every $n\in \NN_{\geq 1}$,  $\Gamma(n)=(n-1)!$.

We want to compute the probabilities under this model of $T_6^\mathit{mb}$ and $T_6^\mathit{gfb}$. To do that, we shall need to compute all values $q_{6,\beta}(k)$ (we need all of them in order to compute the normalizing constant $a_6(\beta)$):
\begin{align*}
q_{6,\beta}(1) & =q_{6,\beta}(5)=\frac{1}{a_6(\beta)}\cdot \frac{\Gamma(\beta+2)\Gamma(\beta+6)}{\Gamma(2)\Gamma(6)}\\ & =
\frac{1}{a_6(\beta)}\cdot \frac{(\beta+5)(\beta+4)(\beta+3)(\beta+2)\Gamma(\beta+2)^2}{5!}\\
q_{6,\beta}(2) & =q_{6,\beta}(4)=\frac{1}{a_6(\beta)}\cdot \frac{\Gamma(\beta+3)\Gamma(\beta+5)}{\Gamma(3)\Gamma(5)}\\ & =
\frac{1}{a_6(\beta)}\cdot \frac{(\beta+4)(\beta+3)(\beta+2)^2\Gamma(\beta+2)^2}{2\cdot 4!}\\
q_{6,\beta}(3) & =\frac{1}{a_6(\beta)}\cdot \frac{\Gamma(\beta+4)^2}{\Gamma(4)^2}\\ & =
\frac{1}{a_6(\beta)}\cdot \frac{(\beta+3)^2(\beta+2)^2\Gamma(\beta+2)^2}{3!^2}
\end{align*}
Imposing now $\sum_{k=1}^5 q_{6,\beta}(k)=1$, i.e., 
\begin{align*}
1 & =\frac{(\beta+3)(\beta+2)\Gamma(\beta+2)^2}{a_6(\beta)}\\
&\qquad\cdot\Big(\frac{2(\beta+5)(\beta+4)}{5!}+\frac{2(\beta+4)(\beta+2)}{2\cdot 4!}+\frac{(\beta+3)(\beta+2)}{3!^2}\Big)\\
&=\frac{(\beta+3)(\beta+2)\Gamma(\beta+2)^2(31 \beta^2 + 194\beta + 300)}{a_6(\beta)\cdot 3\cdot 5!}
\end{align*}
and solving for $a_6(\beta)$, we obtain
$$
a_6(\beta)=\frac{(\beta+3)(\beta+2)\Gamma(\beta+2)^2(31 \beta^2 + 194\beta + 300)}{3\cdot 5!}.
$$
We can compute now the desired probabilities:
\begin{itemize}
\item As far as $P_\beta(T_6^\mathit{mb})$ goes, by Eqn. (\ref{eqn:markovbeta}) we have that
$$
P_\beta(T_6^\mathit{mb})=q_{6,\beta}(3)\cdot P_\beta(T_3^\mathit{mb})^2=q_{6,\beta}(3)
$$
because $\TT_3=\{T_3^\mathit{mb}\}$ and hence $P_\beta(T_3^\mathit{mb})=1$. So,
\begin{align*}
P_\beta(T_6^\mathit{mb}) & = 
\frac{3\cdot 5!\cdot (\beta+3)^2(\beta+2)^2\Gamma(\beta+2)^2}{3!^2(\beta+3)(\beta+2)\Gamma(\beta+2)^2(31 \beta^2 + 194\beta + 300)}\\
& = 
\frac{10 (\beta+3)(\beta+2)}{31 \beta^2 + 194\beta + 300}.
\end{align*}

\item As to $P_\beta(T_6^\mathit{gfb})$ goes, by Eqn. (\ref{eqn:markovbeta}) (and the fact that $T_6^\mathit{gfb} = (T_4^\mathit{gfb}, T_2^\mathit{gfb})$ by Lemma \ref{gfb_subtrees} and Proposition \ref{GFB_Decomposition}) we have that
$$
P_\beta(T_6^\mathit{gfb})=2q_{6,\beta}(4)\cdot P_\beta(T_2^\mathit{gfb})P_\beta(T_4^\mathit{gfb})
$$
where $P_\beta(T_2^\mathit{gfb})=1$, because $T_2^\mathit{gfb}$ is the only tree in $\TT_2$;
$$
P_\beta(T_4^\mathit{gfb})=
P_\beta(T_4^\mathit{mb})=\frac{3(\beta+2)}{7\beta+18}
$$
by Lemma 4 in \citep{CMR}; and
\begin{align*}
q_{6,\beta}(4)& = \frac{3\cdot 5!\cdot (\beta+4)(\beta+3)(\beta+2)^2\Gamma(\beta+2)^2}{2\cdot 4!\cdot (\beta+3)(\beta+2)\Gamma(\beta+2)^2(31 \beta^2 + 194\beta + 300)}\\
& = \frac{15 (\beta+4)(\beta+2)}{2(31 \beta^2 + 194\beta + 300)}.
\end{align*}
So, finally,
\begin{align*}
P_\beta(T_6^\mathit{gfb}) & =2\cdot \frac{15 (\beta+4)(\beta+2)}{2(31 \beta^2 + 194\beta + 300)}\cdot \frac{3(\beta+2)}{7\beta+18}\\
& = \frac{45 (\beta+4)(\beta+2)^2}{(31 \beta^2 + 194\beta + 300)(7\beta+18)}.
\end{align*}
\end{itemize}

}

\end{document}